\documentclass[11pt]{article}

\usepackage{graphicx}
\usepackage{citesort}
\usepackage{amssymb}
\usepackage{amsfonts}
\usepackage{amsmath}
\usepackage{epsfig}
\usepackage{color}
\usepackage{fancybox}
\usepackage{textcomp}
\usepackage{multirow}
\usepackage{appendix}
\usepackage{setspace}
\usepackage{psfrag}

\renewcommand{\baselinestretch}{1.6}
\newtheorem{Theorem}{Theorem}

\newtheorem{Proposition}{Proposition}
\newtheorem{Algorithm}{Algorithm}
\newtheorem{Lemma}{Lemma}
\newtheorem{Assumption}{Assumption}

\newtheorem{Remark}{Remark}

    \def\ang#1{\mbox{$\langle #1 \rangle$}}
    \def\aang#1{\mbox{$\langle\!\langle #1 \rangle\!\rangle$}}

    \def\mnorm#1{\mbox{$\left|\!\left|\!\left|\, #1 \,\right|\!\right|\!\right|$}}

    \def\Complex{{\rm\rule[.23ex]{.03em}{1.1ex}\kern-.3em{C}}}

    \def\largeN{\mbox{${\cal N}$}}

    \newcommand{\be}{\begin{equation}} \newcommand{\ee}{\end{equation}}
    \newcommand{\bea}{\begin{eqnarray}} \newcommand{\eea}{\end{eqnarray}}
    \newcommand{\benum}{\begin{enumerate}} \newcommand{\eenum}{\end{enumerate}}

    \newcommand{\qa}{{\bf a}}
    \newcommand{\qb}{{\bf b}}

    \newcommand{\qv}{{\bf v}}
    
    \newcommand{\qx}{{\bf x}}

    \newcommand{\qA}{{\bf A}}
    \newcommand{\qB}{{\bf B}}
    \newcommand{\qC}{{\bf C}}
    
    \newcommand{\qE}{{\bf E}}
    \newcommand{\qF}{{\bf F}}
    \newcommand{\qG}{{\bf G}}
    \newcommand{\qH}{{\bf H}}
    \newcommand{\qI}{{\bf I}}
    
    \newcommand{\qK}{{\bf K}}
    
    \newcommand{\qM}{{\bf M}}

    \newcommand{\qP}{{\bf P}}
    \newcommand{\qQ}{{\bf Q}}
    \newcommand{\qR}{{\bf R}}
    
    \newcommand{\qT}{{\bf T}}
    \newcommand{\qU}{{\bf U}}
    \newcommand{\qV}{{\bf V}}
    
    \newcommand{\qX}{{\bf X}}
    \newcommand{\qY}{{\bf Y}}
    
    \newcommand{\qzero}{{\bf 0}}
    \newcommand{\qone}{{\bf 1}}

    \newcommand{\qDelta}{{\boldsymbol \Delta}}
    \newcommand{\qPsi}{{\boldsymbol \Psi}}
    \newcommand{\qPhi}{{\boldsymbol \Phi}}
    \newcommand{\qXi}{{\boldsymbol \Xi}}
    \newcommand{\qTheta}{{\boldsymbol \Theta}}
    \newcommand{\qLambda}{{\boldsymbol \Lambda}}
    
    \newcommand{\qGamma}{{\boldsymbol \Gamma}}
    \newcommand{\qUpsilon}{{\boldsymbol \Upsilon}}
    \newcommand{\qOmega}{{\boldsymbol \Omega}}
    \newcommand{\qPi}{{\boldsymbol \Pi}}

    \newcommand{\qdelta}{{\boldsymbol \delta}}

    \newcommand{\qepsilon}{{\boldsymbol \epsilon}}
    \newcommand{\qeta}{{\boldsymbol \eta}}
    
    \newcommand{\qmu}{{\boldsymbol \mu}}

    \newcommand{\qxi}{{\boldsymbol \xi}}
    \newcommand{\qzeta}{{\boldsymbol \zeta}}

    \newcommand{\tH}{{\tilde{H}}}
    \newcommand{\talpha}{{\tilde{\alpha}}}
    \newcommand{\ttau}{{\tilde{\tau}}}
    \newcommand{\te}{{\tilde{e}}}
    \newcommand{\tb}{{\tilde{b}}}
    \newcommand{\tv}{{\tilde{v}}}
    \newcommand{\trho}{{\tilde{\rho}}}
    \newcommand{\tvarepsilon}{{\tilde{\varepsilon}}}

    \newcommand{\tqH}{{\tilde{\qH}}}

    \newcommand{\tqQ}{{\tilde{\qQ}}}
    \newcommand{\tqTheta}{{\tilde{\qTheta}}}
    \newcommand{\tqDelta}{{\tilde{\qDelta}}}
    \newcommand{\tqPsi}{{\tilde{\qPsi}}}
    \newcommand{\tqPhi}{{\tilde{\qPhi}}}
    \newcommand{\tcalqS}{{\tilde{\calqS}}}
    \newcommand{\tqXi}{{\tilde{\qXi}}}
    \newcommand{\tqv}{{\tilde{\qv}}}

    \newcommand{\bH}{{\bar{H}}}
    \newcommand{\bqH}{{\bar{\qH}}}

    \newcommand{\uqB}{{\underline{\qB}}}
    \newcommand{\uH}{{\underline{H}}}
    \newcommand{\utH}{{\underline{\tH}}}
    \newcommand{\ubH}{{\underline{\bH}}}
    \newcommand{\uR}{{\underline{R}}}
    \newcommand{\uT}{{\underline{T}}}
    \newcommand{\uqH}{{\underline{\qH}}}
    \newcommand{\uqR}{{\underline{\qR}}}
    \newcommand{\uqT}{{\underline{\qT}}}
    \newcommand{\ubqH}{{\underline{\bqH}}}
    \newcommand{\utqH}{{\underline{\tqH}}}
    \newcommand{\ucalqX}{{\underline{\calqX}}}
    \newcommand{\uqX}{{\underline{\qX}}}

    \newcommand{\bbR}{{\mathbb R}}
    \newcommand{\bbC}{{\mathbb C}}
    \newcommand{\bbS}{{\mathbb S}}
    
    \newcommand{\bbQ}{{\mathbb Q}}

    \newcommand{\calI}{{\mathcal I}}

    \newcommand{\calS}{{\mathcal S}}
    
    \newcommand{\calV}{{\mathcal V}}
    \newcommand{\calX}{{\mathcal X}}

    \newcommand{\calqB}{\boldsymbol{\cal B}}
    \newcommand{\calqX}{\boldsymbol{\cal X}}
    \newcommand{\calqS}{\boldsymbol{\cal S}}

    \newcommand{\diag}{{\sf diag}}
    
    \newcommand{\tr}{{\sf tr}}
    \newcommand{\Ex}{{\sf E}}
    
    \newcommand{\Varx}{{\sf Var}}

    \newcommand{\aaa}{\buildrel{\circ}\over{a}}
    \newcommand{\aeta}{\buildrel{\circ}\over{\eta}}
    \newcommand{\arho}{\buildrel{\circ}\over{\rho}}
    
    \newcommand{\ateta}{\buildrel{\circ}\over{\tilde\eta}}
    
    \newcommand{\atrho}{\buildrel{\circ}\over{\tilde\rho}}

\vspace{-.5in}
\title{On Capacity of Large-Scale MIMO Multiple Access Channels with\\Distributed Sets of Correlated Antennas
\footnote{The work of J. Zhang, S. Jin, and X. Q. Gao was supported by the National Natural Science Foundation of China under Grants 60902009, 61222102 and 61201171, the Natural Science Foundation of Jiangsu Province under Grants BK2012021, and the Supporting Program for New Century Excellent Talents in University (NCET-11-0090).
The work of C.-K. Wen was supported by the National Science Council, Taiwan, under grant NSC100-2221-E-110-052-MY3.}}

\author{Jun Zhang\footnote{J. Zhang, S. Jin, and X. Q. Gao are with the National Mobile Communications Research Laboratory, Southeast University, Nanjing 210096, China, Email: \{mtzhangjun, jinshi, xqgao\}@seu.edu.cn.}, Chao-Kai Wen\footnote{Institute of Communications Engineering, National Sun Yat-sen University, Taiwan. Email: ckwen@ieee.org.}, Shi Jin$^\dag$,  Xiqi Gao$^\dag$, and Kai-Kit Wong\footnote{Department of Electronic and Electrical Engineering, University College London, UK, Email: k.wong@ee.ucl.ac.uk.} }

\date{}

\begin{document}
\maketitle

\vspace{-.5in}
\begin{abstract}
In this paper, a deterministic equivalent of ergodic sum rate and an algorithm for evaluating the capacity-achieving input covariance matrices for the uplink
large-scale multiple-input multiple-output (MIMO) antenna channels are proposed. We consider a large-scale MIMO system consisting of multiple users and one base
station with several distributed antenna sets. Each link between a user and an antenna set forms a two-sided spatially correlated MIMO channel with line-of-sight
(LOS) components. Our derivations are based on novel techniques from large dimensional random matrix theory (RMT) under the assumption that the numbers of antennas
at the terminals approach to infinity with a fixed ratio. The deterministic equivalent results (the deterministic equivalent of ergodic sum rate and the
capacity-achieving input covariance matrices) are easy to compute and shown to be accurate for realistic system dimensions. In addition, they are shown to be
invariant to several types of fading distribution.
\end{abstract}

%\vspace{0.05in}
{\bf Index Terms}---Deterministic equivalent, large dimensional RMT, large-scale MIMO, Stieltjes transform.
%\vspace{0.05in}

\vspace{-.2in}
\section{Introduction}
To achieve higher rates, much efforts have been put to improving the spectral efficiency and data throughput of wireless communication systems. The multi-antenna
technology is one key technology for wireless communication and is envisaged to be adopted ubiquitously. With the number of antennas at the base stations (BSs) and
user equipments (UEs) being increased, communications systems will have better rate and link reliability \cite{Fos-98,Tel-99}. However, the actual achievable
spectral efficiency could be greatly compromised by interference arising from simultaneous communications in neighboring areas.

A promising solution to interference management is the large-scale multiple-input multiple-output (MIMO) technology, e.g.,
\cite{Marzetta10TWC,Jose11TWC,Ngo11ACASP,Hoydis11ACC,Rusek11SPM}. Figure \ref{fig:LargeScaleMIMO} illustrates a possible scenario where the antenna array of a BS
is composed of multiple geographically distributed low-power antenna sets, installed onto a ring of high-speed fibre-bus, and this BS is communicating with several
multi-antenna UEs. The large-scale MIMO setting is beneficial not only in terms of communication performances (such as better coverage and efficient radio resource
utilization) but also in terms of energy-saving.\footnote{Using the setting, the number of BS can be greatly reduced. Note that the energy consumption for air
conditioning  for each BS is consuming up to 20,000 kWh each year on average which is sometimes higher than other equipments in a BS \cite{ZTE-01}.} In this
complex system model, a number of practical factors such as correlation effects and line-of-sight (LOS) components need to be included, which occur due to the
space limitation of UEs and the densification of the antenna arrays resulting in a visible propagation path from the UEs, respectively. For typical systems of tens
of distributed antenna sets and hundreds of UEs, even computer simulations become challenging \cite{Huh-11}, which makes performance analysis of such large-scale
MIMO systems an important and a new subject of research.

\begin{figure}
\begin{center}
\resizebox{4.00in}{!}{%
\includegraphics*{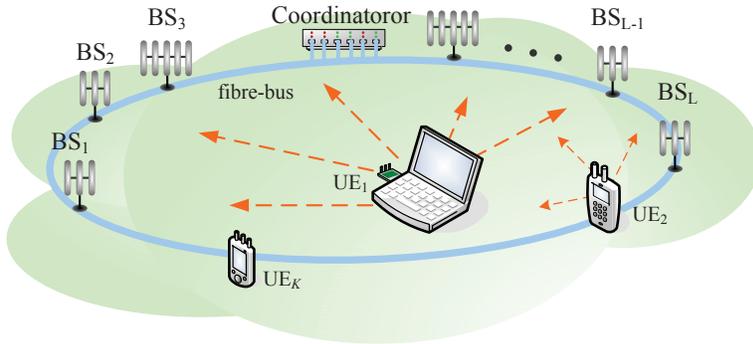} }%
\caption{A vision for a possible large-scale MIMO system.}\label{fig:LargeScaleMIMO}
\end{center}
\end{figure}

When a system is large, exact performance analysis is no longer suitable because an exact analytical expression would be too complex to appreciate. Hence, alternatives have emerged and the large dimensional random matrix theory (RMT) \cite{Moustakas03IT,Tulino-04,Hachem-07AAP,Hachem-08IT,Taricco08IT,Dumont-10IT,Couillet-11IT,Couillet-11Book,Dupuy-11IT,Wen-11IT,Hoydis11ACC} provides a powerful tool in dealing with large-scale MIMO systems. Utilizing the large dimensional RMT, this paper aims to derive information-theoretic results of the large-scale MIMO systems. In particular, our focus is on the uplink large-scale MIMO systems consisting of $K$ UEs and a BS with $L$ distributed sets of multiple antennas. Let $n_k$ and $N_l$ denote, respectively, the numbers of antennas at the $k$-th UE and the $l$-th antenna set of the BS receiver. The channel between the $k$-th UE and the $l$-th antenna set is modeled as the $N_l \times n_k$ complex matrix
%\begin{equation}
$\qH_{l,k}=\qR_{l,k}^\frac{1}{2} \qX_{l,k} \qT_{l,k}^\frac{1}{2} +\bqH_{l,k}$,
%\end{equation}
where $\qX_{l,k}$'s are statistically independent random matrices of independent and identically distributed (i.i.d.) entries (but {\em not} necessary Gaussian\footnote{Despite the Rayleigh or Rician distribution being the most popular distributions for small-scale amplitude fading, there are other classes of fading distributions which serve as better models under certain circumstances \cite{Molisch05,Foerster03}.}), $\bqH_{l,k}$ is a deterministic matrix reflecting the LOS components of the channel, and $\qR_{l,k}$ and $\qT_{l,k}$, respectively, characterize the spatial correlation structures at the receiver and transmitter sides separatively. Since the signals from multiple antenna sets are collected into a BS, the corresponding channel matrix of UE $k$ can be expressed as $\qH_k\triangleq
[\qH_{1,k}^T \cdots \qH_{L,k}^T]^T$. An important objective of this study is to obtain a deterministic equivalent of the ergodic sum rate for the distributed uplink MIMO channel $\sum_{k=1}^{K} \qH_k \qH_k^H$ so that the system sum rate can be efficiently and accurately computed.

Although there have been quite many such results on MIMO capacity analysis utilizing large dimensional RMT
\cite{Moustakas03IT,Hachem-07AAP,Hachem-08IT,Taricco08IT,Dumont-10IT,Couillet-11IT}, the general model studied in this paper has not been addressed. To appreciate the objective of this paper, it is important to understand the limitations of the existing results. First, previous works in the large-scale MIMO systems usually assumed $n_k=1$ and $N_l=1$ for all $k,l$. That is, the UEs have only one antenna each and the BS is equipped with completely distributed antennas (i.e., one antenna in each antenna set). The elements of this channel matrix merely reflect the path loss differences between the links. Regarding the channel model (the channel with a variance profile), the most relevant work is \cite{Hachem-07AAP} (or \cite[Theorem 3.8]{Tulino-04} without the LOS components). In \cite{Hachem-07AAP}, a deterministic equivalent of the mutual information\footnote{Formally, it should be read as the mutual information between the input and output over the channel with a variance profile. In this paper, we often simply refer to it as ``the mutual information'' if no confusion would occur.} was derived based on the Bai-and-Silverstein method \cite{Bai-10} (or \cite[Chapter 6.2.1]{Couillet-11Book}). In fact, the results of \cite{Hachem-07AAP} can be easily extended to the case with $n_k \geq 1$ and $N_l \geq 1$ but those spatial correlation matrices $\qT_{l,k}$'s and $\qR_{l,k}$'s are required to be {\em diagonal}.

The deterministic approximations of \cite{Hachem-07AAP,Tulino-04} have found many applications in various system optimization designs such as scheduling \cite{Huh-11,Huh-12}, training length designs \cite{Hoy-11c}, cell planning \cite{Hoy-12}, and many others \cite{Couillet-11Book,Tulino-04}. This is because the designs based on the deterministic approximations not only can provide an efficient computation method but also give insight into what the optimal strategies look like. However, inheriting from the limitations of \cite{Hachem-07AAP,Tulino-04}, these results do not allow the UEs or each antenna set of the BS to be equipped with multiple \emph{spatially correlated} antennas. Because of the potential applicability of deterministic equivalent results to system designs, there is a strong desire to deriving new deterministic equivalents as those given in \cite{Hachem-07AAP} for the general model of our interest. However, even for an extension to the one-sided spatially correlated case, there will be several obstacles when one intends to get the deterministic equivalent of mutual information by using the Bai-and-Silverstein method and alike.\footnote{If $\bqH_{l,k}=\qzero~\forall l,k$, a partial generalization is possible by the Bai-and-Silverstein method. Specifically, with minor modifications for the case in \cite{Wagner-11IT,Couillet-11IT}, the asymptotic mutual information can be obtained for the case that $\qR_{l,k}$'s were permitted to be nonnegative definite, while $\bqH_{l,k}=\qzero$ and $\qT_{l,k}$'s are diagonal. If $\qT_{l,k}$'s are generally nonnegative definite, difficulties arise.}

To date, there are only very few results dealing with random matrix models where the entries are correlated across both rows and columns. Most studies only considered random matrices with independent complex Gaussian random variables and used the fact that the correlated Gaussian random matrix can be transformed to an uncorrelated one with non-identically distributed entries without changing the concerned objects (e.g., the eigenvalue distribution and the mutual information). For convenience, we will refer to this transformation as the {\em decorrelation procedure}. Because of the assumption of Gaussianity the entries are in fact uncorrelated, and so the Bai-and-Silverstein method can be used. For the latest results using this trick, refer to, e.g., \cite{Couillet-11IT}. Unfortunately, the channel model of our interest (i.e., $\qH_k$) cannot be transformed to a Gaussian random matrix with uncorrelated columns even if $\qX_{l,k}$'s are assumed to be Gaussian. For this to be possible, it would require that $\qT_{1,k}, \dots, \qT_{L,k}$ be simultaneously unitarily diagonalizable for every $k$. Clearly, this restriction in the model does not permit UEs to have multiple spatially correlated antennas, which is unrealistic and greatly limits the significance of the model.

If the entries of the random matrices are Gaussian, then an alternative method, known as the Gaussian method \cite{Pastur-05AAP} (the integration by part formula and Poincar\'e-Nash inequality), is much more useful. In this context, Hachem {\em et al.}~\cite{Hachem-08IT,Dumont-10IT} have succeeded in obtaining the deterministic equivalent of mutual information for Kronecker (or separately) correlated Rayleigh and Rician MIMO channels. Compared to the Bai-and-Silverstein method, the Gaussian method is only suited to random matrices with Gaussian entries. However, one may extend the results obtained for matrices with Gaussian entries to any random matrices with independent entries following two recent developments, the Lindeberg principle \cite{Chatterjee-06AAP} and the interpolation trick \cite{Lytova-09AAP}. For the latest results, see, e.g., \cite{Wen-11IT}, where the Lindeberg principle is applied.

Early analyses using the Gaussian method were only for the typical Kronecker MIMO channel \cite{Hachem-08IT,Dumont-10IT}.\footnote{In this paper, the typical Kronecker MIMO channel means that $K=L=1$.} In that case, the correlated Gaussian random matrix was transformed into an uncorrelated one, and the decorrelation procedure was employed. As such, the Gaussian method was merely an alternative tool to study large dimensional random matrices. Its superiority in dealing with random matrices with correlated pattern is largely unexplored until most recently, Dupuy and Loubaton in \cite{Dupuy-11IT} derived the deterministic equivalent of average mutual information for a frequency selective MIMO channel, in which the decorrelation procedure could not be applied. We believe that the Gaussian method can be useful to treat other random matrices with involved correlation. With the aid of the Lindeberg principle, one may further extend the results obtained for matrices with Gaussian entries to any random matrices. Following this approach, this paper combines the two techniques to get the deterministic equivalents for the concerned channel model.

In particular, we first use the Gaussian method to derive the deterministic equivalent of ergodic sum rate for the large-scale MIMO multiple access channel (MAC) when $\qX_{l,k}$'s are Gaussian distributed. Our results are much more general and can cope with several complex applications. As a special case, this contribution complements the results of \cite{Dupuy-11IT} by extending the analysis to the case with LOS components. This extension is non-trivial.\footnote{Using the Gaussian tools, the asymptotic mutual information expressions for Rayleigh fading Kronecker MIMO channels were first proved by \cite{Hachem-08IT}. Two years later, the authors in the same group generalized the results to Rician fading channels \cite{Dumont-10IT}. This in some ways reflects the difficulty of such extension even for the typical Kronecker MIMO channel.} Next, by the generalized Lindeberg principle \cite{Chatterjee-06AAP,Korada-11IT}, we generalize the deterministic equivalent for random matrices with Gaussian entries to those with non-Gaussian entries. Simulation results reveal that even for systems with realistic system dimensions, the deterministic approximation of ergodic sum rate provides reliable estimates to those obtained by Monte-Carlo simulations. Then, we apply the approximation to design the input covariances that tend to maximize the ergodic sum rate of the large-scale MIMO MAC, and provide an iterative water-filling optimization algorithm when only the statistical CSI at the transmitter (precisely, $\qT_{l,k}$'s, $\qR_{l,k}$'s, and $\bqH_{l,k}$'s) is available. Finally, we conduct several simulations to confirm the comparability between results by our approach and those by the true (but time-consuming) optimization procedure under several types of fading distribution.

\emph{Notations}---We use uppercase and lowercase boldface letters to denote matrices and vectors, respectively. ${\bf{I}}_N$ denotes an $N \times N$ identity matrix while an all-zero matrix is denoted by ${\bf{0}}$, and an all-one matrix is denoted by ${\bf{1}}$. The matrix inequality $\succeq$ shows the positive semi-definiteness. The superscripts $(\cdot)^{H}$, $(\cdot)^{T}$, and $(\cdot)^{*}$ represent the conjugate-transpose, transpose, and conjugate operations, respectively. Also, we use $\Ex\{\cdot\}$ to denote expectation with respect to all random variables within the brackets; $\log(\cdot)$ is the natural logarithm; $\rho(\cdot)$ denotes the spectral radius (i.e., the largest absolute value of the eigenvalues) of a matrix. $\|\cdot\|$ represents the Euclidean norm of an input vector or the spectral norm of an input matrix, while $\|\cdot\|_{\rm F}$ denotes the Frobenius norm of a matrix, and $\mnorm{\cdot}_{\infty}$ represents the maximum row sum matrix norm. The complex number field is denoted by $\mathbb{C}$. For any matrix $\qA \in \mathbb{C}^{N \times n}$, we use $[{\bf A}]_{lk}$, $[{\bf A}]_{l,k}$ or $A_{kl}$ to denote the ($l$,$k$)-th entry, and $a_k$ denotes the $k$-th entry of the column vector $\bf{a}$. The operators $(\cdot)^{\frac{1}{2}}$, $(\cdot)^{-1}$, ${\tr}(\cdot)$ and $\det(\cdot)$ represent the matrix principal square root, inverse, trace and determinant, respectively. In addition, $\diag(\bf{x})$ denotes a diagonal matrix with an input vector $\bf{x}$ representing its diagonal elements.

\section{Channel Model and Problem Statement}
\subsection{Uplink Large MIMO}
As shown in Figure \ref{fig:LargeScaleMIMO}, we consider the large-scale MIMO MAC with $K$ UEs, labeled as ${\sf UE}_1,\dots,{\sf UE}_K$, which are equipped with $n_1,\dots,n_K$ antennas, respectively. The $K$ UEs transmit simultaneously to a central coordinator with $L$ distributed antenna sets, labeled as ${\sf BS}_1,\dots,{\sf BS}_L$, which are equipped with $N_1,\dots,N_L$ antennas, respectively. In this paper, we use the Kronecker model to characterize the spatial correlation of the MIMO channel for each MIMO link so that the correlation at an antenna set and a UE is modeled separately, as in \cite{Shiu-00TCOM}. Specifically, the channel from ${\sf UE}_k$ to ${\sf BS}_l$, $\qH_{l,k}\in\bbC^{N_l \times n_k}$, can be written as
\begin{equation}\label{eq:Spatial_Cov}
\qH_{l,k} = \tqH_{l,k}  + \bqH_{l,k}\equiv\qR_{l,k}^\frac{1}{2} \qX_{l,k} \qT_{l,k}^\frac{1}{2} + \bqH_{l,k},
\end{equation}
where $\qR_{l,k} \in\bbC^{N_l \times N_l}$ and $\qT_{l,k} \in\bbC^{n_k \times n_k}$ are deterministic nonnegative definite matrices, characterizing the spatial correlation of the received signals across the antenna elements of ${\sf BS}_l$ and that of the transmitted signals across the antenna elements of ${\sf UE}_k$, respectively; $\qX_{l,k} \equiv [\frac{1}{\sqrt{n_k}} X_{ij}^{(l,k)}] \in\bbC^{N_l \times n_k}$ consists of the random components of the channel in which the elements are i.i.d.~complex random variables with zero mean and unit variance; and $\bqH_{l,k} \in\bbC^{N_l \times n_k}$ is a deterministic matrix corresponding to the channel LOS.

With the channel given above, we define the Rician factor between ${\sf UE}_k$ and ${\sf BS}_l$ as
\begin{equation}
\kappa_{l,k} = \frac{ \|\bqH_{l,k}\|_{\rm F}^2 }{ \Ex{\{ \|\tqH_{l,k}\|_{\rm F}^2\} } }.
\end{equation}
We also denote the distance-dependent pathloss of the $(l,k)$-th pair by $g_{l,k} = \Ex{ \left\{ \|\qH_{l,k}\|_{\rm F}^2 \right\} }/N_l$ given by
\begin{equation}
 \Ex{ \left\{ \|\qH_{l,k}\|_{\rm F}^2 \right\} } = \frac{1}{n_k} \tr{\left(\qR_{l,k}\right)} \tr{\left(\qT_{l,k}\right)} + \tr{\left( \bqH_{l,k}\bqH_{l,k}^H \right)}.
\end{equation}
Following the standard conventions \cite{Taricco08IT}, $\qR_{l,k}$, $\qT_{l,k}$, and $\bqH_{l,k}$ are normalized such that
\begin{equation}
\left\{
\begin{aligned}
\tr{(\qR_{l,k})} &= \frac{1}{\kappa_{l,k}+1} g_{l,k} N_l, \\
\tr{(\qT_{l,k})} &= n_k, \\
\tr{\left( \bqH_{l,k} \bqH_{l,k}^H \right)} &= \frac{\kappa_{l,k}}{\kappa_{l,k}+1} g_{l,k} N_l.
\end{aligned}
\right.\end{equation}
It is noted that $\kappa_{l,k} $ and $g_{l,k}$ are independent from the matrix dimensions. Therefore, the normalization is valid for all possible correlation patterns and imposes no restriction on practical applications. Although for convenience purpose we will simply set the same noise level (i.e., $\sigma^2$) at all the receivers, it imposes no restriction since one can adjust $g_{l,k}$ to get an arbitrary signal-to-noise ratio (SNR) of the $(l,k)$-th pair. In addition, the setting implies that the LOS components of some link pairs are allowed to be absent.

\subsection{Problem Formulation}
The sum rate has been a key metric for performance analysis of a MAC. We begin with the sum rate formulation of the large-scale MIMO system and then explain its relation to RMT. For ease of exposition, we define $N \triangleq\sum_{l=1}^{L} N_l$, $n\triangleq\sum_{k=1}^{K} n_k$, $\qH_k\triangleq \left[ \qH_{1,k}^T\cdots\qH_{L,k}^T \right]^T \in \bbC^{N \times n_k}$, $\bqH_k \triangleq \left[ \bqH_{1,k}^T \cdots \bqH_{L,k}^T \right]^T \in \bbC^{N \times n_k}$, $\qH\triangleq\left[\qH_1 \cdots \qH_K\right]  \in \bbC^{N \times n}$, and $\bqH \triangleq\left[\bqH_1 \cdots \bqH_K\right] \in \bbC^{N \times n}$. The channel $\qH_k$ represents the joint channel between ${\sf UE}_k$ and the $L$ distributed antenna sets interconnected at the BS. Then, the ergodic sum rate of the MIMO MAC can be expressed as \cite{Goldsmith-03JSAC}
\begin{equation}
\calV_{\qB_N}(\sigma^2) \equiv \frac{1}{N} \Ex \left\{ \log\det\left(\qI_N + \frac{1}{\sigma^2} \qB_N \right) \right\}
\end{equation}
where $\sigma^2$ is the noise variance at the receivers and
\begin{equation}
 \qB_N \triangleq \sum_{k=1}^{K} \qH_k \qH_k^H  ~\in \bbC^{N \times N}. \label{eq:main_model}
\end{equation}
Specifically, $\calV_{\qB_N}(\sigma^2)$ provides a performance metric regarding the total number of nats (or bits if in base 2 of logarithm) per antenna that can be transmitted reliably over the channel matrices $\{\qH_k\}_{k =1,\dots,K}$.

The derivative of $\calV_{\qB_N}(\sigma^2)$ with respect to $\sigma^2$ is given by
\begin{equation}
\frac{\partial \calV_{\qB_N}(\sigma^2)}{\partial \sigma^2} = \frac{1}{N} \Ex\left\{ \tr\left[ \left(\qI_N + \frac{1}{\sigma^2} \qB_N \right)^{-1} \right] \right\} - \frac{1}{\sigma^2}.
\end{equation}
By Fubini's theorem, we have \cite[page 891]{Hachem-07AAP}
\begin{equation}\label{eq:shannonTrans}
\calV_{\qB_N}(\sigma^2) = \int_{\sigma^2}^{\infty} \left( \frac{1}{\omega} - \Ex\{ m_{\qB_N}(\omega) \} \right) d\omega,
\end{equation}
where
\begin{equation}\label{eq:stjTrans}
m_{\qB_N}(\omega)\triangleq \frac{1}{N} {\sf tr} \left( \qB_N + \omega \qI_{N} \right)^{-1}.
\end{equation}
In RMT, $m_{\qB_N}$ is referred to as the Stieltjes transform of $\qB_N$ at point $-\omega$, which provides a convenient tool to study the behavior of large dimensional random matrices. The relationship by which the mutual information is expressed as a functional of the Stieltjes transform is called the Shannon transform \cite[Section 2.2.3]{Tulino-04}.

In this paper, we are interested in understanding the ergodic sum capacity of the MIMO MAC by using large dimensional RMT. In particular, we consider that $L$, $K$ are fixed but $N_1, \dots, N_L, n_1, \dots, n_K$ all go to infinity with ratios $\{\beta_{l,k}(N) \equiv \frac{N_l}{n_k}\}$ such that
\begin{equation}\label{eq:asymptoticregime}
0 <\min_{l,k}\,\liminf_{N} \beta_{l,k}(N) < \max_{l,k}\,\limsup_{N} \beta_{l,k}(N) < \infty.
\end{equation}
For convenience,  we refer to this large dimensional regime simply as $\largeN \rightarrow \infty$ in the sequel. To this end, in the next section, we first find a {\em deterministic} matrix-valued function $\qPsi(\omega) \in \bbC^{N \times N}$ (to be done later) such that
\begin{equation} \label{eq:AidstjCong}
\Ex\left\{m_{\qB_N}(\omega)\right\} - \frac{1}{N} \tr{(\qPsi(\omega))} \xrightarrow{\largeN \rightarrow \infty} 0~~\mbox{for } \omega \in \bbR^{+}.
\end{equation}
Following \cite{Hachem-07AAP} (or \cite[Definition 6.1]{Couillet-11Book}), we refer to $\frac{1}{N} \tr(\qPsi(\omega))$ as the deterministic equivalent of $\Ex\left\{m_{\qB_N}(\omega)\right\}$. To appreciate the contributions of this paper, it is worth emphasizing that $\qH_k$, in general, cannot be written in the form \eqref{eq:Spatial_Cov} using the separable correlation model, because different antenna sets have different spatial correlations, and this is the main obstacle of this class of random matrices -- otherwise, there are some existing results
\cite{Moustakas03IT,Tulino-04,Hachem-07AAP,Hachem-08IT,Taricco08IT,Dumont-10IT,Couillet-11IT,Dupuy-11IT,Wen-11IT,Hoydis11ACC}. Next, using the Shannon transform (\ref{eq:shannonTrans}), we will find $\calV_N(\sigma^2)$ so that $\Ex \{\calV_{\qB_N}(\sigma^2)\} - \calV_N(\sigma^2) \rightarrow 0$ as $\largeN\rightarrow\infty$. Finally, we will use $\calV_N(\sigma^2)$ to obtain the optimal input covariance matrices that maximize the deterministic approximation of the ergodic sum rate.

\section{Deterministic Equivalents and Ergodic Capacity}
\subsection{Deterministic Equivalents}
We first state the assumptions imposed in our system model.
\begin{Assumption} \label{Ass:1}
Let $\qX_{l,k} \equiv [\frac{1}{\sqrt{n_k}} X_{ij}^{(l,k)}] \in \mathbb{C}^{N_l \times n_k}$, where $X_{ij}^{(l,k)}$'s are i.i.d.~complex random variables with
independent real and imaginary parts such that
\begin{equation} \label{eq:assX}
\Ex\{ X_{11}^{(l,k)}\}=0,~\mbox{and}~\Ex\{|X_{11}^{(l,k)}|^2\}=1,
\end{equation}
and have finite $6$-th order moment.
\end{Assumption}

\begin{Assumption} \label{Ass:2}
The family of deterministic matrices $\{\qT_{l,k},\qR_{l,k}\}_{\forall l,k}$ is nonnegative definite. In addition, the spectral norms of $\qR_{l,k}$, $\qT_{l,k}$,
and $\bqH_{l,k}\bqH_{l,k}^H$ are bounded by a constant, i.e.,
\begin{equation}\label{eq:asumRTH}
   \max_{k,l}\,\max\{\|\qR_{l,k}\|,\|\qT_{l,k}\|, \|\bqH_{l,k} \bqH_{l,k}^H\| \} \leq C_{\rm max}.
\end{equation}
\end{Assumption}

To facilitate our expressions, we define the notation $\ang{\qA}_k$ that returns the submatrix of $\qA$ obtained by extracting the elements of the rows and columns with indices from $\sum_{i=1}^{k-1} n_i+1$ to $\sum_{i=1}^{k} n_i$. Similarly, the notation $\aang{\qA}_l$ returns the submatrix of $\qA$ obtained by extracting the elements of the rows and columns with indices from $\sum_{j=1}^{l-1} N_j+1$ to $\sum_{j=1}^{l} N_j$. Also, for convenience, in the paper, we often omit $\omega$ when writing $m_{\qB_N}, \qPsi, \tqPsi, \qPhi, \tqPhi, \qPhi_l, \tqPhi_k, e_{l,k}, \te_{l,k}$, and denote $\sum_{l,k} \equiv\sum_{l=1}^{L}\sum_{k=1}^{K}$.

\begin{Theorem} \label{mainTh_Stj}
Let $\beta_{l,k} = \frac{N_l}{n_k}$. Under Assumption \ref{Ass:2}, the deterministic system of the $L \times K$ equations
\begin{subequations} \label{eq:Solutionete}
\begin{align}
e_{l,k} &= \frac{1}{N_l} \tr{\left( \qR_{l,k} \aang{\qPsi}_l \right)},\\
\te_{l,k} &= \frac{1}{n_k} \tr{\left( \qT_{l,k} \ang{\tqPsi}_k \right)},
\end{align}
\end{subequations}
for $1 \leq l \leq L$ and $1 \leq k \leq K$, where
\begin{subequations} \label{eq:PsiS}
\begin{align}
\qPsi &=  {\left( \qPhi^{-1} + \omega \bqH \tqPhi \bqH^H \right)^{-1}}, \label{eq:Psi} \\
\tqPsi &= {\left( \tqPhi^{-1} + \omega \bqH^H \qPhi \bqH \right)^{-1}}, \label{eq:tPsi} \\
\qPhi &= \diag{( \qPhi_1, \ldots, \qPhi_L )}, \\
\tqPhi &= \diag{(\tqPhi_1, \ldots, \tqPhi_K )},\\
\qPhi_l & = \left(\omega \qI_{N_l} + \omega \sum_{k=1}^{K}{ {\te_{l,k}} \qR_{l,k} } \right)^{-1},\\ % l=1,\ldots,L,
\tqPhi_k & =\left(\omega\qI_{n_k} + \omega \sum_{l=1}^{L}{ { \beta_{l,k} e_{l,k}} \qT_{l,k}}\right)^{-1} %, k=1,\ldots,K
\end{align}
\end{subequations}
have a unique solution for $\omega \in \bbR^+$.

Under Assumptions \ref{Ass:1} and \ref{Ass:2}, as $\largeN \rightarrow \infty$, we then have
\begin{equation} \label{eq:stjNG}
\Ex \left\{ m_{\qB_N} \right\} - \frac{1}{N} \tr{(\qPsi)} = O\left( \frac{1}{\sqrt{N}}\right), ~~ \mbox{for}~\omega \in \bbR^+.
\end{equation}
Furthermore, if $\qX_{l,k}$'s are Gaussian, we have
\begin{equation} \label{eq:stjG}
\Ex \left\{ m_{\qB_N} \right\} - \frac{1}{N} \tr{(\qPsi)} = O\left( \frac{1}{N^2}\right), ~~ \mbox{for}~\omega \in \bbR^+.
\end{equation}
\end{Theorem}

\begin{proof}
Here, for ease of understanding, we give an outline of the proof. Our strategy is to show that the deterministic equivalent of $\Ex \left\{ m_{\qB_N}\right\}$ [i.e. $\frac{1}{N} \tr{(\qPsi)}$] can be found for the Gaussian random matrices and then we prove that the result is also applied for the non-Gaussian distributions.

Let $\calqB_N$ be an $N \times N$ matrix obtained from $\qB_N$ in (\ref{eq:main_model}) with all $\qX_{l,k}$'s replaced by $\calqX_{l,k}$'s, where $\calqX_{l,k}$'s are matrices with entries being independent standard {\em Gaussian}. Using the Gaussian method \cite{Pastur-05AAP} (the integration by part formula and Poincar\'{e}-Nash inequality), we can show that the error term $ \Ex\{ m_{\calqB_N}\} - \frac{1}{N} \tr{\left(\qPsi\right)}$ is of order $O\left(\frac{1}{N^2}\right)$. The detailed derivation is given in Appendix \ref{Appendix: Proof of result Gaussian}.

Next, applying the Lindeberg principle \cite[Theorem 2]{Korada-11IT}, we prove that $\Ex \{ m_{\qB_N} \} - \Ex \{ m_{\calqB_N}\} =
O\left(\frac{1}{\sqrt{N}}\right)$. The detailed derivation using the Lindeberg principle is provided in Appendix \ref{Appendix: Proof of GtoNG}. Together with the result for the Gaussian case, the proof of \eqref{eq:stjNG} can be accomplished by noting that
\begin{equation*}
  \Ex \left\{ m_{\qB_N}\right\} - \frac{1}{N} \tr\left(\qPsi\right)
  = \underbrace{\Big(\Ex \left\{ m_{\qB_N} \right\} - \Ex \left\{ m_{\calqB_N}\right\}\Big)}_{=O\left(\frac{1}{\sqrt{N}}\right)}
  + \underbrace{\Big(\Ex \left\{ m_{\calqB_N}\right\} - \frac{1}{N} \tr\left(\qPsi\right)\Big)}_{=O\left(\frac{1}{N^2}\right)}.
\end{equation*}

Finally, we consider the existence and uniqueness of the solution to \eqref{eq:Solutionete} in Appendix \ref{Appendix: Existence and Uniquenss}.
\end{proof}

\begin{Remark}
If $X_{ij}^{(l,k)}$'s are Gaussian, the assumption that $X_{ij}^{(l,k)}$'s have finite $6$-th order moment is naturally satisfied. When the amplitudes of the channel fading coefficients follow the Nakagami and log-normal distributions, Theorem \ref{mainTh_Stj} is applicable since these distributions have finite $6$-th order moment. In Appendix \ref{Appendix: Proof of GtoNG}, the proof of $\Ex \left\{ m_{\qB_N} \right\} - \frac{1}{N} \tr{(\qPsi)} =O\left(\frac{1}{\sqrt{N}}\right)$ was given under the assumption that $X_{ij}^{(l,k)}$'s have finite $6$-th order moment. In fact, with additional arguments, the more general case can be obtained. Specifically, if $X_{ij}^{(l,k)}$'s have only finite second moment, we can prove that $\Ex \left\{ m_{\qB_N}(\omega)\right\} -\frac{1}{N} \tr{(\qPsi(\omega))} = O\left( \varepsilon_n \right)$, where $\varepsilon_n$ is a positive sequence converging to zero. However, it should be noted that with the finite $6$-th order moment assumption, the proof of $\Ex \left\{ m_{\qB_N} \right\} - \frac{1}{N} \tr{(\qPsi)} = O\left(\frac{1}{\sqrt{N}}\right)$ is much simpler than the latter general case. The proof of the general case requires some additional truncation, centralization, and rescaling techniques together with some careful derivations as those in \cite{Wen-11IT}. Since these are beyond the scope of this paper, we do not show the detail proof regarding this general case. Interested readers can refer to \cite{Wen-11IT}.
\end{Remark}

\begin{Remark}
Theorem \ref{mainTh_Stj} is developed under the asymptotic regime where $L$, $K$ are fixed but $\{N_l, n_k\}$'s all grow to infinity with fixed ratios. For other applications, we might be interested in the cases with fixed $\{N_l, n_k\}$'s while $L$ and $K$ grow to infinity. In this case, the entries of $\qX_{l,k}$'s will be normalized by $\sqrt{n}$ rather than $\sqrt{n_k}$ and a similar deterministic equivalent result as that of Theorem \ref{mainTh_Stj} can be obtained.\footnote{Only different in some scalar adjustment.}

\end{Remark}
We then derive a deterministic equivalent of the ergodic sum rate of the large-scale MIMO MAC in the following theorem.
\begin{Theorem} \label{mainTh_Cap}
Assuming that $\qB_N$ follows the hypotheses of Theorem \ref{mainTh_Stj}, as $\largeN \rightarrow \infty$, the Shannon transform of $\qB_N$ satisfies
\begin{equation} \label{eq:a4}
  \Ex\{ \calV_{\qB_N}(\sigma^2) \} - \calV_N(\sigma^2) = O\left( \frac{1}{\sqrt{N}}\right),
\end{equation}
where
\begin{subequations} \label{eq:AsyShannon}
\begin{align}
  \calV_N(\sigma^2)&= \frac{1}{N} \log\det{\left( \frac{\qPsi(\sigma^2)^{-1}}{\sigma^2} \right)} + \frac{1}{N} \sum_{k=1}^{K}{{\log\det\left( \frac{\tqPhi_k(\sigma^2)^{-1}}{\sigma^2} \right)} } -  \frac{\sigma^2}{N}\sum_{l,k}{ N_l e_{l,k}(\sigma^2) \te_{l,k}(\sigma^2)},\label{eq:AsyShannon1}\\
&= \frac{1}{N} \log\det{\left( \frac{\tqPsi(\sigma^2)^{-1}}{\sigma^2} \right)} + \frac{1}{N} \sum_{l=1}^{L}{{\log\det\left( \frac{\qPhi_l(\sigma^2)^{-1}}{\sigma^2} \right)} } -  \frac{\sigma^2}{N}\sum_{l,k}{ N_l e_{l,k}(\sigma^2) \te_{l,k}(\sigma^2)}.\label{eq:AsyShannon2}
\end{align}
\end{subequations}
Furthermore, if $\qX_{l,k}$'s are Gaussian, we have, as $\largeN \rightarrow \infty$,
\begin{equation} \label{eq:AsyShannonG}
  N \left( \Ex\{ \calV_{\qB_N}(\sigma^2) \} - \calV_N{(\sigma^2)} \right) = O\left( \frac{1}{N}\right).
\end{equation}
\end{Theorem}

\begin{proof}
By \eqref{eq:stjNG} in Theorem \ref{mainTh_Stj} together with the dominated convergence theorem, \eqref{eq:a4} is obtained. Then, we show that $\int^{\infty}_{\sigma^2} \left(\frac{1}{\omega} - \frac{1}{N}\tr{\left(\qPsi(\omega)\right)}\right)d\omega$ can be written more explicitly as \eqref{eq:AsyShannon1}. The details of the proof are similar to those in \cite[Theorem 3]{Wen-11IT}, and thus omitted. Since
$\det\left(\qI+\qA\qB\right)=\det\left(\qI+\qB\qA\right)$, we then have \eqref{eq:AsyShannon2}. On the other hand, \eqref{eq:AsyShannonG} can be obtained by \eqref{eq:stjG} in Theorem \ref{mainTh_Stj}.
\end{proof}

\begin{Remark}
With \eqref{eq:a4}, we can get the deterministic equivalent of the ergodic sum rate regarding the number of nats per antenna. However, \eqref{eq:AsyShannonG} shows the convergence regarding the total ergodic sum rate and as a consequence has a wider range of applications for the performance evaluation criteria.
\end{Remark}

Over the last few years, there have been quite many deterministic equivalent results obtained by using large dimensional RMT (e.g.,
\cite{Moustakas03IT,Tulino-04,Hachem-07AAP,Hachem-08IT,Taricco08IT,Dumont-10IT,Couillet-11IT,Dupuy-11IT,Wen-11IT,Hoydis11ACC}). Since our model is fairly general,
Theorem \ref{mainTh_Cap} may be interpreted as a unified formula that encompasses many such results. For the case with $K = 1$ and $\bqH = \qzero$,
$\calV_N(\sigma^2)$ agrees with that in \cite[Theorem 2]{Dupuy-11IT}, in which $\{ \qX_{l,1} \}_{\forall l}$ are assumed to be Gaussian. Theorem \ref{mainTh_Cap}
thus extends its application to the non-Gaussian scenarios in this sense. Indeed, if $\bqH = \qzero$, (\ref{eq:AsyShannon}) was first presented in
\cite[(23)]{Wen-07TCOM}, where the replica method was used. Also, for the case with $K=2$, $L=1$, and $\{\qR_{l,k}=\qR\}_{\forall k}$, Theorem
\ref{mainTh_Cap} is consistent with the results in \cite{Taricco-11IT} by the replica method which is however mathematically incomplete. In contrast, Theorem
\ref{mainTh_Cap} is not only mathematically rigorous but also more general than the proposition in \cite{Wen-07TCOM} in the sense that $\bqH \neq \qzero$ and there
is no requirement on the Gaussian distribution on the entries of $\qX_{l,k}$. Finally, if $n_k= 1$ and $N_l= 1$ for all $k,l$, then Theorem \ref{mainTh_Cap}
degenerates to that in \cite{Hachem-07AAP} (or \cite{Tulino-04} without the LOS components). Clearly, in contrast with \cite{Hachem-07AAP,Tulino-04}, Theorem
\ref{mainTh_Cap} allows the UEs and each antenna set of the BS to be equipped with multiple spatially correlated antennas.

As mentioned before, deterministic equivalent results together with optimization approaches have found numerous applications in system optimization designs \cite{Huh-11,Huh-12,Hoy-11c,Hoy-12}. For example, based on the deterministic equivalent result of \cite{Tulino-04}, the authors of \cite{Huh-11} devised an algorithm to compute the ergodic sum rate subject to a general fairness criterion. Also, based on \cite{Tulino-04}, the authors of \cite{Huh-12} derived an analytical expression of a system spectral efficiency when multiple BSs employ joint transmission with linear zero-forcing beamforming. They also developed a downlink scheduling scheme under a fairness criterion. Our deterministic equivalent results provide a promising foundation to these applications while under the more general large-scale MIMO system. In addition, a deterministic equivalent for the SINR at the output of the MMSE receiver can be derived using our deterministic equivalent results. Due to space limitations, such applications through Theorems \ref{mainTh_Stj}--\ref{mainTh_Cap} are left out. In the next subsection, our aim is to answer one of the fundamental questions: How should the input covariances be designed so that the ergodic sum rate can be maximized?

\subsection{Ergodic Capacity}
It is well known that the ergodic sum capacity of a MIMO MAC is achieved by selecting proper input covariance matrices so that the ergodic sum rate is maximized \cite{Yu-04IT}. In this subsection, we aim to design the optimal covariance matrices using the deterministic equivalent results. Firstly, we state that the covariance matrices maximizing the deterministic equivalent of the ergodic sum rate yield a result which converges to the ergodic capacity. After that, these optimal covariance matrices will be shown to be structurally equivalent to an iterative waterfilling procedure over a deterministic channel. Finally, we propose an iterative waterfilling algorithm for finding the capacity-achieving input covariance matrices.

Let $\qQ_k$ be the input covariance matrix of ${\sf UE}_k$ which satisfies $\tr(\qQ_k)\leq n_k$.\footnote{The power constraint can be replaced by $\tr(\qQ_k)\leq P_k n_k$ with $P_k$ being any finite positive value independent from the matrix dimension. Note that the current setting $\tr(\qQ_k)\leq n_k$ is for notational brevity only.} With the input covariance matrices $\qQ \triangleq \diag\left( \qQ_1, \ldots, \qQ_K \right)$, we thus write the ergodic sum rate of the large-scale MIMO MAC as
\begin{equation}
   \Ex \left\{\calV_{\qB_N}(\sigma^2, \qQ_1, \dots, \qQ_K)\right\} = \frac{1}{N} \Ex{ \left\{ \log\det{\left(\qI_N + \frac{1}{\sigma^2} \qH\qQ\qH^H \right)} \right\} }.
\end{equation}
Then, the ergodic capacity under the power constraint is given by
\begin{equation}\label{eq:maxergodic}
   \max_{\qQ_k \in \bbQ_k, \forall k} \Ex{ \left\{\calV_{\qB_N}{(\sigma^2, \qQ_1, \ldots, \qQ_K)} \right\} },
\end{equation}
where $$\bbQ_k \triangleq \Big\{\qQ_k \Big| \: \tr{(\qQ_k)}\leq n_k ~\mbox{and}~ \qQ_k \succeq \bf0 \Big\}$$ is the feasible set of $\qQ_k$. The problem
\eqref{eq:maxergodic} is convex and can be solved using stochastic programming based on convex optimization with Monte-Carlo methods \cite{Boyd04}. Specifically,
we can apply the method in \cite{Vu-05PAC} (called the Vu-Paulraj algorithm), which was developed based on the barrier method \cite[Chap. 11]{Boyd04} where the
related gradient and Hessian are approximated by Monte-Carlo methods. Since $\qQ_k$ is a Hermitian matrix of size $n_k \times n_k$, the optimization involves $n_k$
real entries on the diagonal and $n_k(n_k-1)/2$ complex entries in the upper triangle. The complexity of such algorithm is high and requires long execution time.
We thus propose an approximate approach using the deterministic equivalent results in Theorem \ref{mainTh_Cap}.

In Theorem \ref{mainTh_Cap}, we have shown that the deterministic equivalent results are invariant to the type of fading distribution. As a result, the asymptotic
optimal input covariances, which are designed based on the deterministic equivalent results, are also invariant to the type of fading distribution. To get the
deterministic equivalent of $\Ex\left\{\calV_{\qB_N}(\sigma^2, \qQ_1, \ldots, \qQ_K)\right\}$, the effect of $\qQ_k$ has to be included in $\calV_N(\sigma^2)$.
With Theorem \ref{mainTh_Cap}, this can be easily accomplished by the following replacements: for $1 \leq l \leq L$,
\begin{equation} \label{eq:replacements}
 \qT_{l,k}:= \qQ_k^{\frac{1}{2}} \qT_{l,k} \qQ_k^{\frac{1}{2}},~\mbox{and}~~
 \bqH_{l,k} :=\bqH_{l,k} \qQ_k^{\frac{1}{2}}.
\end{equation}
Now, let $\calV_N(\sigma^2, \qQ_1, \dots, \qQ_K )$ be the result obtained from $\calV_N(\sigma^2)$ with $\qT_{l,k}$ and $\bqH_{l,k}$ based on the above
replacements. Then, \eqref{eq:AsyShannon2} becomes
\begin{equation}\label{eq:AsyShannonQ}
\calV_N(\sigma^2, \qQ_1, \dots, \qQ_K)=\frac{1}{N} \log\det{\left( \qI_n +  \qF \qQ \right)} + \frac{1}{N} \sum_{l=1}^{L}{ \log\det{\left( \frac{\qPhi_l(\sigma^2)^{-1}}{\sigma^2} \right)} }-\frac{\sigma^2}{N}\sum_{l,k}{ N_l e_{l,k}(\sigma^2) \te_{l,k}(\sigma^2)},
\end{equation}
where
\begin{equation} \label{eq:defF}
   \qF = \diag\left(\left\{\sum_{l=1}^{L}{ { \beta_{l,k} e_{l,k}(\sigma^2)} \qT_{l,k}}\right\}_{\forall k}\right) +  \bqH^H \qPhi(\sigma^2) \bqH.
\end{equation}
Note that $\qQ_k$'s appear in $\te_{l,k}(\sigma^2)$'s, i.e., $\te_{l,k}(\omega) = \frac{1}{n_k} \tr (\qQ_k^{\frac{1}{2}} \qT_{l,k}\qQ_k^{\frac{1}{2}}
\ang{\tqPsi(\omega)}_k)$  and thus are involved in all the three terms of \eqref{eq:AsyShannonQ}. Using the deterministic equivalent result, we have the
optimization problem:
\begin{equation}\label{eq:maxAsy}
   \max_{\qQ_k \in \bbQ_k, \forall k} \calV_{N}(\sigma^2, \qQ_1, \dots, \qQ_K).
\end{equation}

Before solving the above problem, two important issues must be resolved. One is to establish the concavity of $\calV_{N}(\sigma^2, \qQ_1, \dots, \qQ_K)$ with
respect to $(\qQ_1, \dots, \qQ_K)$, and the other one is to ensure that $\Ex \{\calV_{\qB_N}(\sigma^2, \qQ_1^{\circ}, \dots,
\qQ_K^{\circ})\}-\calV_{N}(\sigma^2,\qQ_1^{\star}, \dots, \qQ_K^{\star})$ goes asymptotically to zero, where, $(\qQ_1^{\circ}, \dots, \qQ_K^{\circ})$ and
$(\qQ_1^{\star}, \dots, \qQ_K^{\star})$ are the maximizers of \eqref{eq:maxergodic} and \eqref{eq:maxAsy}, respectively. The required results are described by the
following proposition.

\begin{Proposition} \label{pro1}
We have:
\begin{enumerate}
\item The function $(\qQ_1, \ldots, \qQ_K) \mapsto \calV_{N}(\sigma^2, \qQ_1, \dots, \qQ_K)$ is strictly concave on $(\bbQ_1, \dots, \bbQ_K)$.
\item In addition to Assumption \ref{Ass:2}, suppose that $\qQ_k^{\circ}$'s and $\qQ_k^{\star}$'s lay within a set of positive semi-definite matrices with bounded spectral norm. The, we have
\begin{equation} \label{eq:PropAsyOptimal}
\Ex \{\calV_{\qB_N}{(\sigma^2, \qQ_1^{\circ}, \dots, \qQ_K^{\circ})}\}-\calV_{N}{(\sigma^2, \qQ_1^{\star}, \dots, \qQ_K^{\star})} = O\left( \frac{1}{\sqrt{N}}\right).
\end{equation}
Furthermore, if $\qX_{l,k}$'s are Gaussian, then (\ref{eq:PropAsyOptimal}) becomes $O\left( \frac{1}{N^2}\right)$.
\end{enumerate}
\end{Proposition}

\begin{proof}
The proof is similar to that in \cite[Theorem 4 and Proposition 3]{Dumont-10IT} and \cite[Theorem 3 and Proposition 4]{Dupuy-11IT}, and therefore omitted.
\end{proof}

So far, we have stated that $(\qQ_1^{\star}, \dots, \qQ_K^{\star})$ yield a result which converges to the ergodic capacity. Next, by using tools from convex
optimization \cite{Boyd04}, we will gain a better understanding on the structure of $(\qQ_1^{\star}, \dots, \qQ_K^{\star})$. In particular, our next proposition
will state that the optimal covariance matrices are structurally equivalent to an iterative waterfilling procedure over a deterministic equivalent channel.

To that end, we start with defining the Lagrangians of the optimization problem \eqref{eq:maxAsy} as
\begin{equation}
\mathcal{L}\left(\qQ, \qUpsilon, \qmu \right) =  - \calV_{N}(\sigma^2, \qQ_1, \ldots, \qQ_K) + \sum_{k=1}^K{ \tr{ \left(\qUpsilon_k\qQ_k\right) } } + \sum_{k=1}^K{ \mu_k \left(n_k - \tr(\qQ_k) \right) },
\end{equation}
where $\qUpsilon \triangleq \{ \qUpsilon_k\}_{\forall k}$ and $\qmu \triangleq \{ \mu_k\}_{\forall k}$ are the Lagrange multipliers associated with the problem
constraints. In order to express the partial derivative of $\calV_{N}(\sigma^2, \qQ_1, \dots, \qQ_K)$ with respect to $\qQ_k$, i.e. $\frac{\partial
\calV_{N}}{\partial \qQ_k}$, we define $\mathcal{I}(\sigma^2, \qQ_1, \dots, \qQ_K) = \frac{1}{N} \log\det{\left( \qI_n + \qF \qQ \right)}$. From
\eqref{eq:AsyShannonQ}, it is noted that the parameters affected by the perturbation of $\qQ_k$ are $\calI(\sigma^2, \qQ_1, \dots, \qQ_K), \{e_{l,k}\}_{\forall
l,k}$, and $\{\te_{l,k}\}_{\forall l,k}$. As a result, we have
\begin{equation}
\frac{\partial \calV_{N}}{\partial \qQ_k} =  \frac{\partial \calV_{N}}{\partial \calI}\frac{\partial \calI}{\partial \qQ_k} +  \sum_{l,k} \frac{\partial \calV_{N}}{\partial e_{l,k}}\frac{\partial e_{l,k}}{\partial \qQ_k} + \sum_{l,k} \frac{\partial \calV_{N}}{\partial \te_{l,k}}\frac{\partial \te_{l,k}}{\partial \qQ_k}.
\end{equation}
It can be checked that $\frac{\partial \calV_{N}}{\partial e_{l,k}} = 0$ and $\frac{\partial \calV_{N}}{\partial \te_{l,k}} = 0, \forall l,k$. Therefore, the
Karush-Kuhn-Tucker (KKT) conditions of \eqref{eq:maxAsy} are
\begin{equation} \label{eq:KKTconditions}
\left\{ \begin{aligned}
         &-\frac{1}{N}\ang{\left( \qI_n + \qF \qQ \right)^{-1}\qF}_k + \qUpsilon_k - \mu_k \qI_{n_k} =0,  \\
         &\tr \left(\qUpsilon_k \qQ_k \right) = 0, ~\qUpsilon_k \succeq 0, ~\qQ_k  \succeq 0,\\
         &\mu_k \left(n_k - \tr(\qQ_k) \right) = 0, ~\mu_k\geq 0,
        \end{aligned} \right.
\end{equation}
for $k=1, \dots, K$.

Since \eqref{eq:maxAsy} is a convex optimization problem with constraints satisfying Slater's condition, the optimal $\qQ_k$'s can be found by solving the KKT
conditions \cite{Boyd04}. Using Lemma \ref{Lemma:optmalCovLemma}, the first line of \eqref{eq:KKTconditions} can be rewritten as
\begin{equation}\label{eq:KKTconditionsfirstline}
-\frac{1}{N}\left( \qI_{n_k} + \qP_k \qQ_k \right)^{-1}\qP_k + \qUpsilon_k - \mu_k \qI_{n_k} =0,
\end{equation}
where
\begin{align}
\qP_k \triangleq & \ang{\left( \qI_n + \qF \qQ_{\backslash k} \right)^{-1}\qF}_k, \label{eq:qPk} \\
\qQ_{\backslash k} \triangleq & \diag\left(\qQ_1, \dots, \qQ_{k-1}, \qzero, \qQ_{k+1}, \ldots,\qQ_K\right).\label{eq:qQbackslash}
\end{align}
Note that $\qP_k$ is a function of $(\qQ_1, \ldots, \qQ_K)$ rather than only $\qQ_{\backslash k}$, as $\qF$ defined in \eqref{eq:defF} includes the whole
$\qQ_k$'s. For brevity, we have omitted its argument when writing $\qP_k$. Substituting \eqref{eq:KKTconditionsfirstline} for the first line of
\eqref{eq:KKTconditions}, the KKT conditions \eqref{eq:KKTconditions} are now equivalent to those of the following optimization problem:
\begin{equation}\label{eq:maxuerk}
\max_{\qQ_k \in \bbQ_k} ~ \frac{1}{N}\log\det{\left( \qI_{n_k} + \qP_k \qQ_k \right)},
\end{equation}
which can be solved by a standard iterative waterfilling procedure. Thus, we get the next proposition.

\begin{Proposition}\label{Propo:oprimalQ}
Let $\qP_k^{\star}$ be the matrix in \eqref{eq:qPk} by replacing $(\qQ_1, \ldots, \qQ_k, \ldots \qQ_K)$ with $(\qQ_1, \ldots, \qQ_k^{\star}, \ldots \qQ_K)$ and
$\qP_k^{\star}=\qV_{P_k}\qLambda_{P_k} \qU_{P_k}^H$. The eigenvectors of $\qQ_k^{\star}$ coincide with the right singular vectors of matrix $\qP_k^{\star}$, i.e.,
\begin{equation}\label{eq:OptimalQk}
    \qQ_k^{\star} = \qU_{P_k} \qLambda_{Q_k}^{\star} \qU_{P_k}^H,
\end{equation}
and the eigenvalues are given by
\begin{equation}\label{eq:OptimalQkEigenvalue}
    \qLambda_{Q_k}^{\star} = \left(\frac{1}{\mu_k}\qI_{n_k} - \qLambda_{P_k}^{-1}\right)^{+},
\end{equation}
where $\left(a\right)^{+} = \max \{0,a\}$ and $\mu_k$ is chosen to satisfy the power constraints $\tr(\qQ_k^{\star}) = n_k$.
\end{Proposition}

Using Proposition \ref{Propo:oprimalQ}, we have the following observations:
\begin{itemize}
\item $\bqH = \qzero$ --- In this case, $\qP_k = \sum_{l=1}^{L} \beta_{l,k} e_{l,k}(\sigma^2) \qT_{l,k}$. Therefore, the optimal transmit directions align with the eigenvectors of some weighted sum of $\qT_{l,k}$'s. As such, $\beta_{l,k} e_{l,k}(\sigma^2)$ can be understood as the equivalent channel gain contributed by ${\sf BS}_l$.
\item $\qH = \bqH$ --- This implies that the channels are deterministic. In this case,
\begin{equation}
  \qP_k = \bqH_k^H \left(\sigma^2\qI_N+\sum_{j\neq k} \bqH_j \qQ_j \bqH_j^H \right)^{-1} \bqH_k.
\end{equation}
It shows that the optimal input covariance matrix of each user follows the water-filling principle that treats the other users as noise. This characteristics
agrees with that for finite-size systems \cite{Yu-04IT}.
\item $K = 1$ --- In this case, we have
\begin{equation}
\qP_1 = \sum_{l=1}^{L} \left(\beta_{l,1}e_{l,1}\qT_{l,1}+\frac{1}{\sigma^2}\bqH_{l,1}^H \left(\te_{l,1}\qR_{l,1}+\qI_{N_l}\right)^{-1} \bqH_{l,1} \right). \nonumber \\
\end{equation}
If $\{ \qR_{l,1}=\qI_{N_l} \}_{\forall l}$, the optimal transmit directions thus align with the eigenvectors of some weighted sum of $\qT_{l,1}$'s and
$\bqH_{l,1}^H\bqH_{l,1}$'s. While if $\qR_{l,1} \neq \qI_{N_l}$, the impact of $\qR_{l,1}$ on the optimal transmit directions is involved by $\bqH_{l,1}$ via
$\bqH_{l,1}^H \left(\te_{l,1}\qR_{l,1}+\qI_{N_l}\right)^{-1} \bqH_{l,1}$. It appears that if the link pair does not have LOS, the corresponding correlation
pattern at the receiver side does not provide a ``direct'' impact on the structure of the optimal transmit directions. Nevertheless, this inference is not
entirely true, since the optimal transmit directions still can be changed by the correlation pattern at the receiver side through $\beta_{l,1}e_{l,1}$. We will
illustrate this phenomenon by an example in the simulation results.
\end{itemize}

Through the observations above, Proposition \ref{Propo:oprimalQ} shows its potential in understanding the impact of antenna correlations and LOS components on the
structure of the optimal transmit directions. We now introduce an iterative algorithm for optimizing $\calV_{N}(\sigma^2, \qQ_1, \dots, \qQ_K)$ which adapts
parameters $\qQ$ and $\{e_{l,k}\}_{\forall l,k}, \{\te_{l,k}\}_{\forall l,k}$ separately.
\begin{Algorithm}\label{Algorithm:IWFA}
(Optimization for $\qQ$)
\begin{itemize}
\item Initialization: $\qQ_k^{(0)} = \qI_{n_k}, e_{l,k}^{(0)} = 1$ and $\te_{l,k}^{(0)} = 1$ for $k=1, \dots, K$ and $k=1,\dots, K$.
\item Iteration $t$:
\begin{itemize}
\item Given that $\qQ_k^{(t-1)}, e_{l,k}^{(t-1)}$ and $\te_{l,k}^{(t-1)}$ are available, for $l=1,\dots, L$ and $k=1,\dots, K$;
\item Calculate $\qT_{l,k}$ and $\bqH_{l,k}$ by the replacements of \eqref{eq:replacements} for $l=1,\dots, L$ and $k=1,\dots, K$. Then, $\{e_{l,k}^{(t)}\}_{\forall l,k}, \{\te_{l,k}^{(t)}\}_{\forall l,k}$ are obtained by
            \begin{align}
             e_{l,k}^{(t)} &= \frac{1}{N_l} \tr{\left( \qR_{l,k} \aang{\qPsi^{(t-1)}}_l \right)}, \nonumber \\
             \te_{l,k}^{(t)} &= \frac{1}{n_k} \tr{\left( \qT_{l,k} \ang{\tqPsi^{(t-1)}}_k \right)}, \nonumber
            \end{align}
            where
            \begin{align}
            \qPsi^{(t-1)} &=  {\left( \left(\qPhi^{(t-1)}\right)^{-1} + \sigma^2 \bqH \tqPhi^{(t-1)} \bqH^H \right)^{-1}},  \nonumber \\
            \tqPsi^{(t-1)} &= {\left( \left(\tqPhi^{(t-1)}\right)^{-1} + \sigma^2 \bqH^H \qPhi^{(t-1)} \bqH \right)^{-1}},  \nonumber \\
            \qPhi^{(t-1)} & = \diag\left( \left\{ \left(\sigma^2 \qI_{N_l} + \sigma^2 \sum_{k=1}^{K} { {\te_{l,k}^{(t-1)}} \qR_{l,k} } \right)^{-1} \right\}_{\forall l} \right), \nonumber\\
            \tqPhi^{(t-1)}& = \diag\left( \left\{ \left(\sigma^2 \qI_{n_k} + \sigma^2 \sum_{l=1}^{L} { {\beta_{l,k} e_{l,k}^{(t-1)}} \qT_{l,k}}\right)^{-1} \right\}_{\forall k} \right); \nonumber
            \end{align}
\item Calculate $\qP_k^{(t)}$ based on \eqref{eq:qPk}, for $k=1, \dots, K$;
\item Calculate $\qQ_k^{(t)}$ based on Proposition \ref{Propo:oprimalQ}, for $k=1, \dots, K$.
\end{itemize}
\item Update $t:= t+1$ until $\left|\calV_{N}(\sigma^2, \qQ_1^{(t)}, \dots, \qQ_K^{(t)})-\calV_{N}(\sigma^2, \qQ_1^{(t-1)}, \dots, \qQ_K^{(t-1)})\right|$ is small enough.
\end{itemize}
\end{Algorithm}

A similar iteration procedure was adopted by \cite{Taricco-11IT}. For the case with $K=L=1$ and $\bqH = \qzero$, the convergence of Algorithm \ref{Algorithm:IWFA} has been proved in \cite{Taricco-11IT}. Note that Algorithm \ref{Algorithm:IWFA} is slightly different from those in \cite{Wen-11TCOM,Dumont-10IT,Dupuy-11IT,Couillet-11IT,Wen06TCOM}, named the frozen water-filling. For the frozen water-filling, $\{e_{l,k}^{(t)}\}_{\forall l,k}, \{\te_{l,k}^{(t)}\}_{\forall l,k}$ are defined as the unique solutions of \eqref{eq:Solutionete} at every iteration step $t$, while in Algorithm \ref{Algorithm:IWFA}, $\{e_{l,k}^{(t)}\}_{\forall l,k}, \{\te_{l,k}^{(t)}\}_{\forall l,k}$ are obtained by performing a single update. It was pointed out in \cite{Taricco-11IT} that the frozen water-filling algorithm does not always converge.\footnote{Note that an example of oscillating behavior of the frozen water-filling algorithm is artificially constructed in \cite{Taricco-11IT}. However, there is no known condition (e.g., spatial correlation pattern) to exclude such behavior of the frozen water-filling algorithm.} The convergence proof of Algorithm \ref{Algorithm:IWFA} is still an open challenge now.

\section{Simulation Results}

\begin{table}
\caption{Angular parameters.}
\begin{center}
\begin{tabular}{|c|c|c|c|c|c|c|c|}
\hline
$\theta_{1,1}^{\sf R}$ & $\theta_{2,1}^{\sf R}$ & $\theta_{1,2}^{\sf R}$ & $\theta_{2,2}^{\sf R}$ & $\theta_{1,1}^{\sf T}$ & $\theta_{2,1}^{\sf T}$ & $\theta_{1,2}^{\sf T}$ & $\theta_{2,2}^{\sf T}$ \\
\hline
$10^{\circ}$ & $20^{\circ}$ & $30^{\circ}$ & $40^{\circ}$ & $15^{\circ}$ & $25^{\circ}$ & $35^{\circ}$ & $45^{\circ}$ \\
\hline \hline
$\delta_{1,1}^{\sf T}$ & $\delta_{2,1}^{\sf T}$ & $\delta_{1,2}^{\sf T}$ & $\delta_{2,2}^{\sf T}$ & $\delta_{1,1}^{\sf T}$ & $\delta_{2,1}^{\sf T}$ & $\delta_{1,2}^{\sf T}$ & $\delta_{2,2}^{\sf T}$ \\
\hline
$0.01$ & $0.02$ & $0.03$ & $0.04$ & $0.04$ & $0.03$ & $0.02$ & $0.01$ \\
\hline  \hline
$\bar{\theta}_{1,1}^{\sf R}$ & $\bar{\theta}_{2,1}^{\sf R}$ & $\bar{\theta}_{1,2}^{\sf R}$ & $\bar{\theta}_{2,2}^{\sf R}$ & $\bar{\theta}_{1,1}^{\sf T}$ & $\bar{\theta}_{2,1}^{\sf T}$ & $\bar{\theta}_{1,2}^{\sf T}$ & $\bar{\theta}_{2,2}^{\sf T}$ \\
\hline
$10^{\circ}$ & $20^{\circ}$ & $30^{\circ}$ & $40^{\circ}$ & $40^{\circ}$ & $30^{\circ}$ & $20^{\circ}$ & $10^{\circ}$ \\
\hline
\end{tabular}
\end{center}\label{tab:compar_Q}
\end{table}

\begin{figure}
\centering
\resizebox{12cm}{!}{\includegraphics{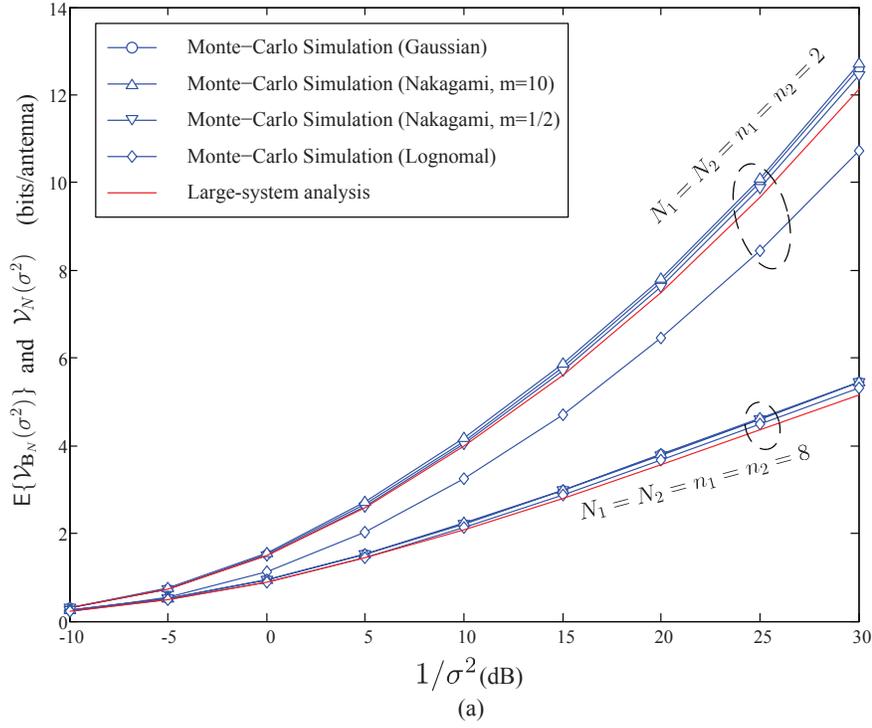}}
\resizebox{12cm}{!}{\includegraphics{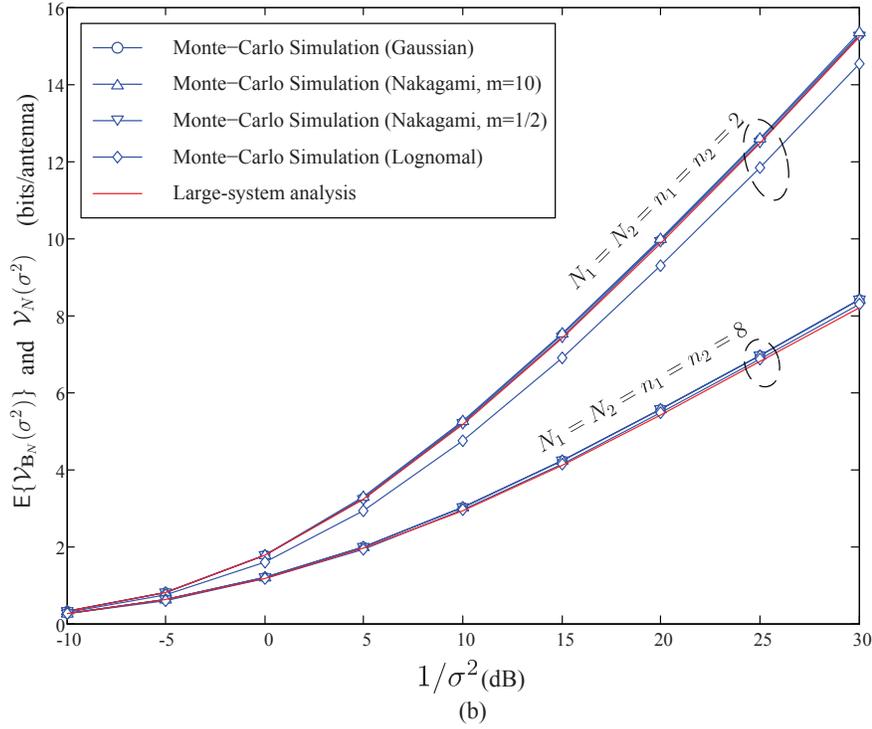}}
\caption{Ergodic sum rate versus SNRs with $N_1 = N_2 = n_1 = n_2 = 2$ and $N_1 = N_2 = n_1 = n_2 = 8$
for a) $\{ \kappa_{l,k} = 0, ~\forall l,k \}$ and  b) $\{ \kappa_{l,k} = 1, ~\forall l,k \}$.
The solid lines plot the deterministic equivalent results, while the markers plot the Monte-Carlo simulation results under different different fading distributions.}\label{fig:SimVsAna}
\end{figure}

In this section, computer simulations are conducted to evaluate the accuracy of the approximation $\calV_N (\sigma^2)$ in Theorem \ref{mainTh_Cap}, and the effectiveness of the iterative algorithm developed in Algorithm \ref{Algorithm:IWFA}. In particular, we are interested in their performances when the numbers of antennas are not so large. The simulation settings are based on the propagation model introduced in \cite{Bol-02TCOM}, in which the spatial correlation is generated from a uniform linear array with half wavelength spacing in a wireless scenario where there is one propagation path cluster with Gaussian power azimuthal distribution having mean angle of $\theta_{k,l}$ and root-mean-square spread of $\delta_{k,l}$. Specifically, we take the correlation matrix with elements
\cite{Moustakas03IT}
\begin{equation} \label{eq: arrayPattern}
[{\bf T}_{l,k}]_{m,n}~(\mbox{or }[{\bf R}_{l,k}]_{m,n}) =\int_{-180}^{180}{ \frac{d\phi}{\sqrt{2\pi\delta_{k}^2}}e^{{\sf j}\pi(m-n)\sin\left(\frac{\pi\phi}{180}\right)-\frac{(\phi- \theta_{l,k})^2}{2\delta_{l,k}^2}} }
\end{equation}
with $m,n$ being the indices of antennas. In addition, we use the superscripts ${\sf T}$ and ${\sf R}$, respectively, to refer to the corresponding values at the transmit and receive sides. The LOS matrix $\bqH_{l,k}$ is generated according to $\bqH_{l,k} = \qa_{{\sf R}, l} (\bar{\theta}_{l,k}^{\sf R}) \qa_{{\sf T},k}(\bar{\theta}_{l,k}^{\sf T})^H$ where
\begin{align*}
\qa_{{\sf R}, l}(\bar{\theta}_{l,k}^{\sf R}) &= \left[1 ~e^{{\sf j}\pi \sin\left(\frac{\bar{\theta}_{l,k}^{\sf R}}{180}\pi\right)} ~\cdots~ e^{{\sf j}\pi(N_l-1)
\sin\left(\frac{\bar{\theta}_{l,k}^{\sf R}}{180}\pi \right)} \right]^T, \\
\qa_{{\sf T},k} (\bar{\theta}_{l,k}^{\sf T}) &= \left[1 ~e^{-{\sf j}\pi \sin\left(\frac{\bar{\theta}_{l,k}^{\sf T}}{180}\pi\right)} ~\cdots~ e^{-{\sf j}\pi(n_k-1)
\sin\left(\frac{\bar{\theta}_{l,k}^{\sf T}}{180}\pi \right)} \right]^T.
\end{align*}
Regarding the fading distribution, we assume that $X_{ij}^{(l,k)}$ is of the form $W_{{\rm R},ij}^{(l,k)} \cos(\theta_{{\rm R},ij}^{(l,k)}) + {\sf j} W_{{\rm
I},ij}^{(l,k)} \sin(\theta_{{\rm I},ij}^{(l,k)})$ \cite{Choi07ICC}, where $\theta_{{\rm R},ij}^{(l,k)}$'s (and $\theta_{{\rm I},ij}^{(l,k)}$'s) are the phases
modeled as i.i.d.~uniform random variables over $[0,2\pi]$, and those $W_{{\rm R},ij}^{(l,k)}$'s (and $W_{{\rm I},ij}^{(l,k)}$'s) are the amplitude fading drawn
from a distribution with $\Ex\{(W_{{\rm R},ij}^{(l,k)})^2\}=1$. The typical probability distributions of $W_{{\rm R},ij}^{(l,k)}$ include the Rayleigh, Nakagami,
and log-normal distributions \cite{Molisch05,Foerster03}. Throughout this section, all the expected values (e.g., $\Ex\{\calV_{\qB_N}(\sigma^2)\}$) are obtained by
the Monte-Carlo method in which $10,000$ independent realizations of $\qH$ are used for averaging.

In Theorem 2, we have shown that in the \emph{large-system limit} the ergodic sum rate is invariant in distribution and can be well approximated by $\calV_N (\sigma^2)$. Therefore, it is important to see how well $\calV_N (\sigma^2)$ in \eqref{eq:AsyShannon} approximates to the ergodic sum rate $\Ex\{\calV_{\qB_N}(\sigma^2)\}$ when the dimensions of the system are not so large. For this purpose, Figure \ref{fig:SimVsAna} compares the results of $\Ex\left\{\calV_{\qB_N}(\sigma^2)\right\}$ with $\calV_N (\sigma^2)$ for $K=2$ and $L=2$ under different fading distributions. Their mean arrival/departure angles and angular spreads are given in Table \ref{tab:compar_Q} and their distance-dependent pathlosses are $g_{1,1} = g_{2,2} = 1$ and $g_{1,2} = g_{2,1} = 0.25$. We see that $\calV_N (\sigma^2)$ produces very good estimates for $\Ex\left\{\calV_{\qB_N}(\sigma^2)\right\}$ even when only a few antenna elements (e.g., $N_1= N_2 =n_1= n_2 = 2$) are located at each UE and antenna set. As expected, when the number of antennas grows large  (e.g., $N_1= N_2 = n_1 = n_2 = 8$) all curves tend to overlap regardless of the distributions. In addition, we notice that for the Nakagami-$m$ distribution, the difference between the case $m = 0.5$ and $m = 10$ is small even when there are only a few antenna elements.

\begin{table}
\caption{Average execution time in seconds.}
\begin{center}
\begin{tabular}{|c|c|c|c|}
\hline
  & $L=10$, $K=20$   &  $L=20$, $K=40$  & $L=30$, $K=60$    \\
\hline
Monte-Carlo simulation &  490 & 1941 & 4541 \\
\hline
Deterministic approximation & 0.5 & 3.2 & 8.2  \\
\hline
\end{tabular}
\end{center}\label{tab:exTime}
\end{table}

In the above experiments, we have shown that $\calV_N (\sigma^2)$ provides a very good approximation for the sum rate of finite-dimensional systems. Before
proceeding, it it useful to discuss the computational efficiency of evaluating $\Ex\left\{\calV_{\qB_N}(\sigma^2)\right\}$ through $\calV_N (\sigma^2)$. For the
considered scenarios in Figure \ref{fig:SimVsAna} with $K=2$ and $L=2$, the execution time for evaluating $\Ex\left\{\calV_{\qB_N}(\sigma^2)\right\}$ is at the
order of decasecond (i.e., $10^1$ seconds). Although the execution time for $\calV_N (\sigma^2)$ is only at the order of centisecond (i.e., $10^{-2}$ seconds), one
may not be convinced to use $\calV_N (\sigma^2)$ since writing a program to perform $\Ex\left\{\calV_{\qB_N}(\sigma^2)\right\}$ is much easier than that for
$\calV_N (\sigma^2)$. However, when the numbers of $K$ and $L$ grow, the Monte-Carlo simulations will become very demanding. Table \ref{tab:exTime} gives the
average execution times on a 2.93 GHz Intel CPU with 4 GB of RAM under various system sizes. Here, we set $\{ N_l = n_k = 2\}_{\forall l,k}$, and the spatial
correlation and LOS are generated from an arbitrary pattern. For typical systems with twenties of distributed antenna sets and forties of users, the simulations
become prohibitive, ruling out the possibility for other system optimization designs such as scheduling \cite{Huh-11,Huh-12}. Clearly, the proposed deterministic
equivalent result is much more efficient in this sense and provides a promising foundation to further applications of system optimization.

\begin{figure}
\centering
\resizebox{12cm}{!}{\includegraphics{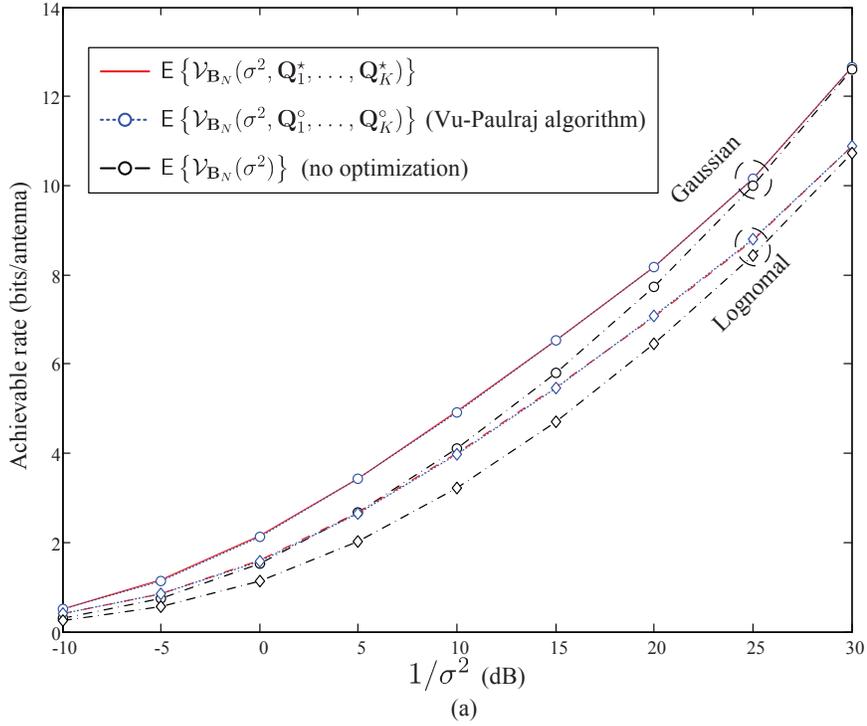}}
\resizebox{12cm}{!}{\includegraphics{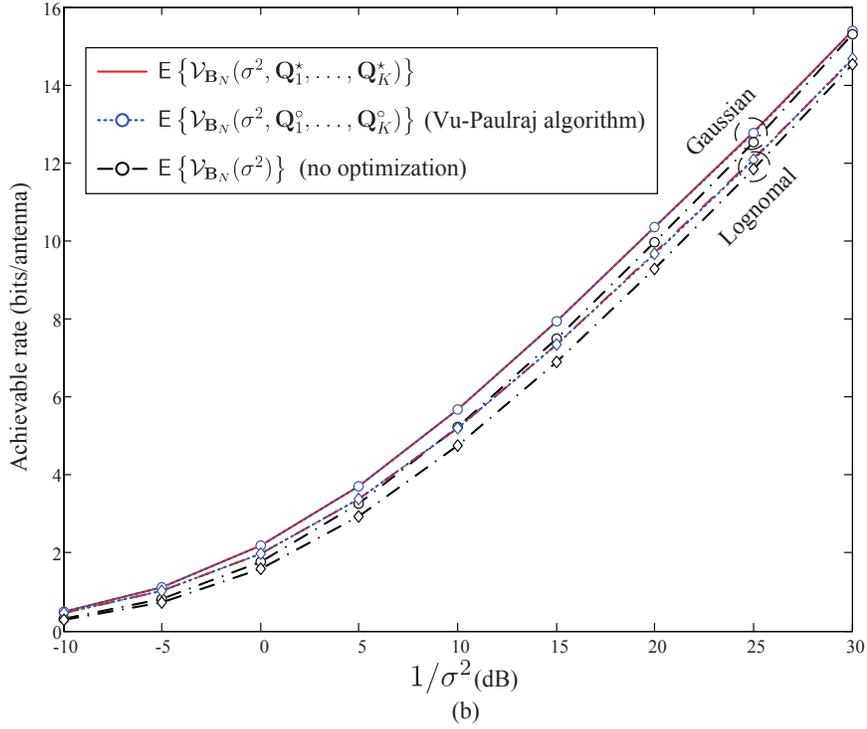}}
\caption{ Achievable rates versus SNRs with $N_1 = N_2 = n_1 = n_2 = 2$ for a) $\{ \kappa_{l,k} = 0, ~\forall l,k \}$ and  b) $\{ \kappa_{l,k} = 1, ~\forall l,k \}$.
The lines plot the results based on the deterministic equivalent, while the markers on dotted line plot the results for the Vu-Paulraj algorithm.}\label{fig:optCov}
\end{figure}

Next, we examine if the input covariance design based on the deterministic equivalent results performs well under different fading distributions when the numbers of antennas are not so large. Recall that $\{\qQ_1^{\circ}, \dots, \qQ_K^{\circ}\}$ denote the optimal solutions of \eqref{eq:maxergodic} that maximize the ergodic sum rate; and $\{\qQ_1^{\star}, \dots, \qQ_K^{\star}\}$ denote the optimal solutions of \eqref{eq:maxAsy} that maximize the deterministic equivalent of the ergodic sum rate. Algorithm \ref{Algorithm:IWFA} is used for solving $\{\qQ_1^{\star}, \dots, \qQ_K^{\star}\}$, while, $\{\qQ_1^{\circ}, \dots, \qQ_K^{\circ}\}$ is solved by the Vu-Paulraj algorithm \cite{Vu-05PAC} which is based on the barrier method where the ergodic sum rate and their first and second derivatives are calculated by the Monte-Carlo method. In contrast to $\{\qQ_1^{\circ}, \dots, \qQ_K^{\circ}\}$, $\{\qQ_1^{\star}, \dots, \qQ_K^{\star}\}$ is independent from the true distributions of $\qX_{l,k}$'s. In Figure \ref{fig:optCov}, we depict $\Ex\{ \calV_{\qB_N}(\sigma^2, \qQ_1, \dots, \qQ_K)\}$ when the input covariance matrices are $\{\qQ_1^{\circ}, \dots, \qQ_K^{\circ}\}$, $\{\qQ_1^{\star}, \dots, \qQ_K^{\star}\}$, and identity matrices, when the amplitude fading distributions are either Rayleigh or log-normal. The reason for considering the two distributions is because, from Figure \ref{fig:SimVsAna}, the values of $\Ex\left\{\calV_{\qB_N}(\sigma^2)\right\}$ for the two distributions are significantly different and $\calV_N (\sigma^2)$ does not get very good estimation on $\Ex\left\{\calV_{\qB_N}(\sigma^2)\right\}$ when the amplitude fading distribution is log-normal. However, regardless of Rayleigh or log-normal distributions, the ergodic sum rate based on $\{\qQ_1^{\star}, \dots, \qQ_K^{\star}\}$ provides indistinguishable results to that based on $\{\qQ_1^{\circ}, \dots, \qQ_K^{\circ}\}$. In addition to its ability of providing good performance, Algorithm \ref{Algorithm:IWFA} is computationally much more efficient than the Vu-Paulraj algorithm.

\begin{figure}
\begin{center}
\resizebox{3.2in}{!}{%
\includegraphics*{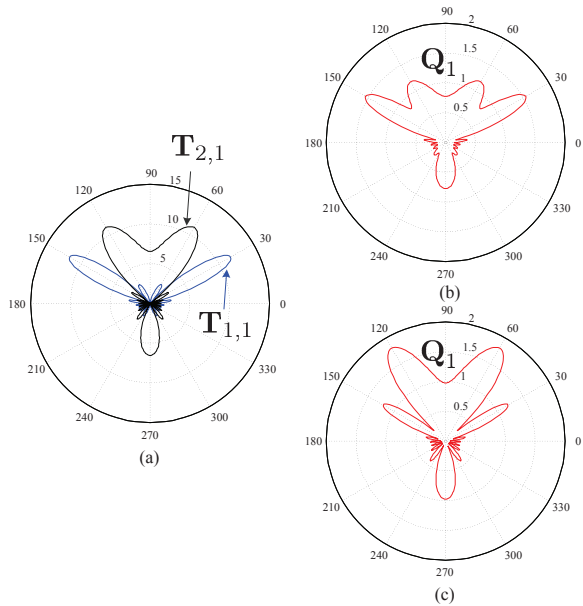} }%
\caption{Antenna radiation patterns.}\label{fig:beamPattern-Ex2}
\end{center}
\end{figure}

Finally, we discuss the fact mentioned in Section 3.2 that the optimal transmit directions can be changed by the correlation pattern at the receiver side through $\beta_{l,1}e_{l,1}$. To understand this better, we consider two scenarios with $K=1$, $L=2$ and $\bqH = \qzero$. The two scenarios use the same parameters except that the radiation patterns at the receiver of the second antenna set have different beam-widthes. Specifically, the radiation patterns at the receiver of the second antenna set have $\delta_{2,1}^{\sf R} = 0.01$ for scenario 1 and $\delta_{2,1}^{\sf R} = 0.1$ for scenario 2. We find it useful to observe the array patterns by plotting its array factor\footnote{Consider a uniform linear array with half wavelength spacing. Given a vector $\qa \in \bbC^{n\times 1}$, we can get its array factor in direction $\phi$ by
\begin{equation*}
    f(\phi) = \sum_{l=1}^{n} a_i e^{-{\sf j}\pi l \sin\left(\frac{\pi\phi}{180}\right) }.
\end{equation*}}
in all directions. The array patterns of $\qT_{1,1}$ and $\qT_{2,1}$ are depicted in Figure \ref{fig:beamPattern-Ex2}(a), where $\theta_{1,1}^{\sf T} =30^{\circ}$, $\theta_{2,1}^{\sf T} = 60^{\circ}$, $\delta_{1,1}^{\sf T} = 0.04$, and $\delta_{2,1}^{\sf T}= 0.03$. The array patterns of the optimal input covariance $\qQ_1$ for the two scenarios are given in Figure \ref{fig:beamPattern-Ex2}(b) and (c), respectively. Figure \ref{fig:beamPattern-Ex2}(c) corresponds to the setting with the broader beamwidth of $\qR_{2,1}$. In this case, the optimal covariance is shown to feed the signal largely according to $\qT_{2,1}$, showing that the optimal transmit directions can be changed by the correlation pattern at the receiver side.

\section{Conclusion}
By using the large dimensional RMT, this paper investigated the deterministic equivalents for the large-scale MIMO MAC. The considered model includes the large-scale MIMO channel such as the general spatial correlation, the LOS components, and the channel entries being non-Gaussian. In particular, we derived the deterministic equivalent of the ergodic sum rate of the large-scale MIMO MAC. In addition, through the deterministic equivalent of the ergodic sum rate, we investigated the capacity-achieving input covariance matrices for the the large-scale MIMO MAC and proposed the iterative waterfilling algorithm for finding them. Finally, computer simulations were conducted to conclude the following three facts: First, the deterministic equivalent of the ergodic sum rate provides a very good approximation even when the numbers of antennas are of practical size. Second, calculating the ergodic sum rate by using the deterministic equivalent result is much more efficient than that by using the Monte-Carlo method when the system sizes are large. Hence, the deterministic equivalent result is of interest to addressing complex system optimization problems. Third, the optimal input covariance matrices predicted by the deterministic equivalent result are indeed remarkably close to those obtained by the corresponding finite-dimensional optimization approach, but in a much more efficient manner.

Investigation of the central limit theorem of the sum rate for the large-scale MIMO MAC by using the mathematical framework in \cite{Hachem-08IT}, as well as application of the deterministic equivalent results to system-level designs \cite{Huh-11,Huh-12,Hoy-11c,Hoy-12}, are promising topics for future research.

\appendix
\section*{Appendix}
%\appendices

\renewcommand{\baselinestretch}{1.2}
\topmargin=-1.05truein \textheight=9.75truein \oddsidemargin=-.35truein\textwidth=7.5truein
%\topmargin=-0.37truein \textheight=9.75truein\oddsidemargin=-.35truein\textwidth=7.5truein

\section{Proof of $\Ex\{ m_{\calqB_N}\} - \frac{1}{N} \tr\left(\qPsi\right) = O\left(\frac{1}{N^2}\right)$ in Theorem \ref{mainTh_Stj}}\label{Appendix: Proof of result Gaussian}
We start the proof by reformulating the channel model so that the derivation can be performed systematically. To this end, we denote
\begin{align}
  \uqR_{l,k} &\triangleq \diag \left(\qzero_{N_1}, \dots, \qzero_{N_{l-1}},\qR_{l,k}, \qzero_{N_{l+1}}, \dots, \qzero_{N_L} \right), \\
  \uqT_{l,k} &\triangleq \diag \left(\qzero_{n_1}, \dots, \qzero_{n_{k-1}},\qT_{l,k}, \qzero_{n_{k+1}}, \dots, \qzero_{n_K} \right).
\end{align}
Also, let $\ubqH_{l,k}$ be the all-zero $N \times n$ matrix except that $\qH_{l,k}$ is used for its $(\sum_{i=1}^{l-1} N_i + 1)$ to $(\sum_{i=1}^{l} N_i )$-th row
and $(\sum_{j=1}^{k-1} n_j +1)$ to $(\sum_{j=1}^{k} n_j)$-th column. As a result, $\qH$ is statistically equivalent to
\begin{equation} \label{eq:multiPathMod}
  \qH = \sum_{l,k} \uqH_{l,k} = \sum_{l,k} \left( \utqH_{l,k}+ \ubqH_{l,k}\right),
\end{equation}
where $\utqH_{l,k} = \uqR_{l,k}^{\frac{1}{2}} \uqX_{l,k} \uqT_{l,k}^{\frac{1}{2}} \in \bbC^{N\times n}$, and $\uqX_{l,k} \equiv \left[\frac{1}{\sqrt{n_k}}
X_{ij}^{(l,k)}\right] \in \bbC^{N\times n}$ consists of the random components of the channel. Here, $\uqX_{l,k}$'s are assumed to be mutually independent. From
\eqref{eq:main_model}, we have
\begin{equation} \label{eq:BN}
  \qB_N = \left(\sum_{l,k} \left( \uqR_{l,k}^{\frac{1}{2}} \uqX_{l,k} \uqT_{l,k}^{\frac{1}{2}}+ \ubqH_{l,k}\right)\right) \left( \sum_{l,k} \left( \uqR_{l,k}^{\frac{1}{2}} \uqX_{l,k} \uqT_{l,k}^{\frac{1}{2}}+ \ubqH_{l,k} \right)\right)^H
\end{equation}
and
\begin{equation} \label{eq:BN_Gau}
  \calqB_N = \left(\sum_{l,k} \left( \uqR_{l,k}^{\frac{1}{2}} \ucalqX_{l,k} \uqT_{l,k}^{\frac{1}{2}}+ \ubqH_{l,k}\right) \right) \left(\sum_{l,k} \left( \uqR_{l,k}^{\frac{1}{2}} \ucalqX_{l,k} \uqT_{l,k}^{\frac{1}{2}}+ \ubqH_{l,k} \right)\right)^H,
\end{equation}
where $\uqX_{l,k}$'s and $\ucalqX_{l,k}$'s are matrices with entries satisfying Assumption \ref{Ass:1} but $\ucalqX_{l,k}$'s are Gaussian.

Let $\calqS$ and $\tcalqS$ be the resolvents of matrices  $\qH\qH^H$ and $\qH^H\qH$, respectively, given by
\begin{align}
  \calqS &\triangleq  \left( \qH\qH^H + \omega\qI_N \right)^{-1}, \label{eq:def S} \\
  \tcalqS &\triangleq  \left( \qH^H\qH + \omega\qI_n \right)^{-1}. \label{eq:def tilde S}
\end{align}
These resolvents clearly satisfy the following useful properties:
\begin{equation}\label{eq:Sproperty}
  \calqS \preceq  \frac{1}{\omega}\qI_N, ~~\mbox{and}~~
  \tcalqS \preceq  \frac{1}{\omega}\qI_n.
\end{equation}

To facilitate our notations, we use $\aaa$ to denote the zero-mean random variable $a-\Ex\{a\}$, where $a$ is a random variable. To accomplish the proof, the
following two lemmas are useful.

\begin{Lemma}\label{Lma:IPFGF}
{\rm (Integration by Parts Formula for Gaussian Functionals) (see, e.g., \cite[Proposition 2.4]{Pastur-05AAP})} Let $\qxi = [\xi_1, \ldots, \xi_M]^T$ be a complex
Gaussian random vector such that $\Ex\{\qxi\}=\qzero$ and $\Ex\{\qxi\qxi^H\}=\qOmega$. Denoting by $\Gamma(\qxi)$ a complex function polynomially bounded with its
derivatives, we have
\begin{equation} \label{eq:intByPart}
    \Ex\left\{ \xi_p \Gamma(\qxi) \right\} = \sum_{m=1}^{M} \Omega_{pm} \Ex\left\{ \frac{\Gamma(\qxi)}{\partial \xi_m^*} \right\}.
\end{equation}
\end{Lemma}

\begin{Lemma} \label{Lma:PNineq}
{\rm (The Poincar\'{e}-Nash Inequality) (see, e.g., \cite[Proposition 2.5]{Pastur-05AAP})} Let $\qxi = [\xi_1, \ldots, \xi_M]^T$ be a complex Gaussian random
vector such that $\Ex\{\qxi\}=\qzero$ and $\Ex\{\qxi\qxi^H\}=\qOmega$. Denoting by $\Gamma(\qxi)$ a complex function polynomially bounded with its derivatives, the
following inequality holds true:
\begin{equation} \label{eq:PNinequality}
    \Varx\left( \Gamma(\qxi) \right) \leq  \Ex\left\{ \left(\bigtriangledown_{\xi} \Gamma(\qxi)\right)^T \qOmega \left(\bigtriangledown_{\xi} \Gamma(\qxi)\right)^*\right\} + \Ex\left\{ \left(\bigtriangledown_{\xi^*} \Gamma(\qxi)\right)^H \qOmega \left(\bigtriangledown_{\xi^*} \Gamma(\qxi)\right) \right\},
\end{equation}
where $\bigtriangledown_{\xi} \Gamma(\qxi) = \left[\frac{\partial\Gamma}{\partial\xi_1}, \dots, \frac{\partial\Gamma}{\partial\xi_M}\right]^T$ and
$\bigtriangledown_{\xi^*} \Gamma(\qxi) = \left[\frac{\partial\Gamma}{\partial\xi_1^*}, \dots, \frac{\partial\Gamma}{\partial\xi_M^*}\right]^T$.
\end{Lemma}

The rigorous proof of Theorem \ref{mainTh_Stj} is rather complex. Although a standard procedure for the MIMO channel without the LOS components \cite{Dupuy-11IT}
is used, several additional manipulations for the LOS components to our present argument are required. To show this, we split the proof into two steps: First, we
prove that $\tr\left(\Ex \{ \calqS \} - \qPsi \right) \rightarrow 0$; secondly, we refine the convergence rate that $\frac{1}{N}\left(\tr\left(\Ex \{ \calqS \} -
\qPsi \right)\right) = O\left(\frac{1}{N^2}\right)$. However, it is difficult to prove directly that $\tr\left(\Ex \{ \calqS \} - \qPsi \right) \rightarrow 0$. To
that end, we employ an intermediate quantity between $\Ex \{ \calqS \}$ and $\qPsi$ and establish the following two propositions.

\begin{Proposition}\label{Prop:SXi}
As $\largeN \rightarrow \infty$, we have
\begin{subequations}
\begin{align}
\tr{ (\Ex{ \{\calqS  \}} - \qXi  )} &\longrightarrow 0, \label{eq:Proposition1a} \\
\tr{ (\Ex{ \{\tcalqS \}} - \tqXi )} &\longrightarrow 0,\label{eq:Proposition1b}
\end{align}
\end{subequations}
where
\begin{subequations}
\begin{align}
  \qXi &\triangleq  \left[ \omega \left(\qI_N + \diag\left( \left\{\sum_{k=1}^{K} \talpha_{l,k} \qR_{l,k}\right\}_{\forall l}\right) +  \bqH \tqTheta \bqH^H \right) \right]^{-1}, \label{eq:Xi}\\
  \tqXi &\triangleq \left[ \omega \left(\qI_n + \diag\left( \left\{\sum_{l=1}^{L} \alpha_{l,k} \qT_{l,k}\right\}_{\forall k}\right) + \bqH^H \qTheta \bqH \right) \right]^{-1},\label{eq:tildeXi}\\
  \qTheta &\triangleq \diag\left(\qTheta_1,\dots,\qTheta_L\right), \label{eq:Theta}\\
  \tqTheta &\triangleq \diag (\tqTheta_1,\dots,\tqTheta_K ),\\
  \qTheta_l &\triangleq \left[ \omega \left( \qI+ \sum_{k=1}^{K} \talpha_{l,k}\qR_{l,k} \right) \right]^{-1}, ~\mbox{for }l=1,\dots,L,\label{eq:Theta_l}\\
  \tqTheta_k &\triangleq \left[ \omega \left( \qI+\sum_{l=1}^{L} \alpha_{l,k}\qT_{l,k} \right) \right]^{-1}, ~\mbox{for }k=1,\dots,K,\label{eq:tildeTheta_k}\\
  \alpha_{l,k} &\triangleq \frac{1}{n_k} \tr{ ( \qR_{l,k} \Ex{\{\aang{\calqS}_l\}}  )}, \label{eq:alpha_lk}\\
  \talpha_{l,k} &\triangleq \frac{1}{n_k} \tr{ ( \qT_{l,k} \Ex{\{\ang{\tcalqS}_k\}} )}. \label{eq:tildealpha_lk}
\end{align}
\end{subequations}
\end{Proposition}

\begin{proof}
See Appendix \ref{Appendix:Proof of Proposition 1}.
\end{proof}

\begin{Proposition}\label{Prop:XiPsi}
As $\largeN \rightarrow \infty$, we have
\begin{subequations}
\begin{align}
\tr{ (\qXi - \qPsi  )} &\longrightarrow 0, \label{eq:Proposition2a} \\
\tr{ (\tqXi - \tqPsi)} &\longrightarrow 0. \label{eq:Proposition2b}
\end{align}
\end{subequations}
\end{Proposition}

\begin{proof}
See Appendix \ref{Appendix:Proof of Proposition 2}.
\end{proof}

From \eqref{eq:Proposition1a} to \eqref{eq:Proposition2a}, the proof of $ \Ex\{ m_{\calqB_N}\} - \frac{1}{N} \tr\left(\qPsi\right)\rightarrow 0$ can be accomplished. %In the remainder of this proof, we will proceed to prove that the convergence order is $O\left(\frac{1}{N^2}\right)$ and the results are justified by the following proposition.

\begin{Proposition}\label{Prop:SXiPsi_order}
As $\largeN \rightarrow \infty$, we have
\begin{subequations}
\begin{align}
\frac{1}{N} \tr{\left(\Ex{\left\{\calqS \right\}} - \qXi \right)} & = O\left(\frac{1}{N^2}\right), \label{eq:Proposition3a} \\
\frac{1}{N} \tr{ (\Ex{ \{ \tcalqS \}} - \tqXi  )}  &= O\left(\frac{1}{N^2}\right),\label{eq:Proposition3b} \\
\frac{1}{N} \tr{\left(\qXi - \qPsi \right)} &= O\left(\frac{1}{N^2}\right), \label{eq:Proposition3c}\\
\frac{1}{N} \tr{ (\tqXi - \tqPsi  )} &= O\left(\frac{1}{N^2}\right).\label{eq:Proposition3d}
\end{align}
\end{subequations}
\end{Proposition}

\begin{proof}
See Appendix \ref{Appendix:Proof of Proposition 3}.
\end{proof}

Consequently, \eqref{eq:stjG} then follows from \eqref{eq:Proposition3a} and \eqref{eq:Proposition3c}. The proof is complete.

\subsection{Proof of Proposition \ref{Prop:SXi}}\label{Appendix:Proof of Proposition 1}
From \eqref{eq:multiPathMod} and \eqref{eq:def S}, we have
\begin{equation}
    \calqS = \frac{1}{\omega}\qI_N - \frac{1}{\omega} {\calqS \qH\qH^H}   %\label{eq:resolS1}\\
           = \frac{1}{\omega}\qI_N - \frac{1}{\omega} \sum_{l_1,k_1} \sum_{l,k} \calqS \uqH_{l_1,k_1}\uqH_{l,k}^H, \label{eq:resolS2}
\end{equation}
and
\begin{align}
    \Ex\{\calS_{pq}\} =& \frac{1}{\omega} \delta_{pq} - \frac{1}{\omega} \Ex{\left\{ [\calqS \qH\qH^H]_{pq} \right\}} \label{eq:resolSpq1}\\
                      =& \frac{1}{\omega} \delta_{pq} - \frac{1}{\omega} \sum_{l,k} \Ex{\left\{ [ \calqS \uqH_{l,k} \uqH_{l,k}^H ]_{pq} \right\}} - \frac{1}{\omega} \sum_{l\neq l_1}^{L}\sum_{k\neq k_1}^{K} \Ex{\left\{ [ \calqS \uqH_{l,k} \uqH_{l_1,k_1}^H ]_{pq} \right\}}.\label{eq:resolSpq2}
\end{align}

We first calculate $\Ex{\left\{ [\calqS \qH\qH^H]_{pq} \right\}}$. Using the integration by parts formula (\ref{eq:intByPart}), we write
\begin{align}
    \Ex{\left\{ \calS_{pi} \uH_{ij}^{(l,k)} \uH_{qr}^{(l,k)*}  \right\}} =& \Ex{\left\{ \calS_{pi} \utH_{ij}^{(l,k)} \uH_{qr}^{(l,k)*}\right\}} + \Ex{\left\{ \calS_{pi} \utH_{qr}^{(l,k)*} \right\}} \underline{\bar{H}}_{ij}^{(l,k)} + \Ex{\left\{\calS_{pi}\right\}} \ubH_{ij}^{(l,k)} \ubH_{qr}^{(l,k)*} \nonumber \\
    =& \frac{1}{n_k} \sum_{m,n} \uR_{im}^{(l,k)} \uT_{jn}^{(l,k)*} \Ex{\left\{ \frac{\partial \calS_{pi} \uH_{qr}^{(l,k)*}  }{\partial \utH_{mn}^{(l,k)*} } \right\}} + \frac{1}{n_k} \sum_{m,n} \uR_{qm}^{(l,k)*} \uT_{rn}^{(l,k)} \Ex{\left\{ \frac{\partial \calS_{pi} }{\partial \utH_{mn}^{(l,k)}} \right\}} \ubH_{ij}^{(l,k)} \nonumber \\
     & + \Ex{\{\calS_{pi}\}} \ubH_{ij}^{(l,k)} \ubH_{qr}^{(l,k)*},
\end{align}
and similarly,
\begin{multline}
    \Ex{\left\{ \calS_{pi} \uH_{ij}^{(l,k)} \uH_{qr}^{(l_1,k_1)*}  \right\}}
    =\frac{1}{n_k} \sum_{m,n} \uR_{im}^{(l,k)} \uT_{jn}^{(l,k)*} \Ex{\left\{ \frac{\partial \calS_{pi} }{\partial \utH_{mn}^{(l,k)*} } \right\} \uH_{qr}^{(l_1,k_1)*}}\\
       +\frac{1}{n_{k_1}} \sum_{m,n} \uR_{qm}^{(l_1,k_1)*} \uT_{rn}^{(l_1,k_1)} \Ex{\left\{ \frac{\partial \calS_{pi} }{\partial \utH_{mn}^{(l_1,k_1)}} \right\}} \ubH_{ij}^{(l,k)} + \Ex{\{\calS_{pi}\}} \ubH_{ij}^{(l,k)} \ubH_{qr}^{(l_1,k_1)*}.
\end{multline}
Now, using the fact that
\begin{subequations}
\begin{align}
    \Ex \left\{ \frac{\partial \calS_{pi} }{\partial \utH_{mn}^{(l,k)}} \right\}  &= - \Ex \left\{ \calS_{pm}[\qH^H\calqS]_{ni}  \right\}, \label{eq:parSparH}\\
    \Ex \left\{ \frac{\partial \calS_{pi} }{\partial \utH_{mn}^{(l,k)*}} \right\}  &= - \Ex \left\{ \calS_{mi}[\calqS\qH]_{pn}  \right\}, \label{eq:parSparHc}\\
    \Ex \left\{ \frac{\partial \calS_{pi} \uH_{qr}^{(l,k)*}  }{\partial \utH_{mn}^{(l,k)*} } \right\} &= \Ex \left\{ \calS_{pi} \delta_{qm} \delta_{rn} - \calS_{mi} [\calqS \qH]_{pn} \uH_{qr}^{(l,k)*} \right\},
\end{align}
\end{subequations}
we have
\begin{multline}
    \Ex{\left\{ \calS_{pi} \uH_{ij}^{(l,k)} \uH_{qr}^{(l,k)*}  \right\}} =\frac{1}{n_k} \uR_{iq}^{(l,k)} \uT_{jr}^{(l,k)*} \Ex{\{\calS_{pi}\}} - \frac{1}{n_k} \Ex{\left\{ [\uqR_{l,k}\calqS]_{ii} [\calqS \qH \uqT_{l,k}]_{pj} \uH_{qr}^{(l,k)*}\right\}}\\
- \frac{1}{n_k} \Ex{\left\{ [\calS \uqR_{l,k}]_{pq}[\uqT_{l,k}\qH^H\calqS]_{ri}  \right\}} \ubH_{ij}^{(l,k)} + \Ex{\{\calS_{pi}\}} \ubH_{ij}^{(l,k)}
\ubH_{qr}^{(l,k)*},
\end{multline}
and
\begin{multline}
    \Ex{\left\{ \calS_{pi} \uH_{ij}^{(l,k)} \uH_{qr}^{(l_1,k_1)*}  \right\}} = - \frac{1}{n_k} \Ex{\left\{ [\uqR_{l,k}\calqS]_{ii} [\calqS \qH \uqT_{l,k}]_{pj} \uH_{qr}^{(l_1,k_1)*}\right\}}\\
      - \frac{1}{n_{k_1}} \Ex{\left\{ [\calS \uqR_{l_1,k_1}]_{pq}[\uqT_{l_1,k_1}\qH^H\calqS]_{ri}  \right\}} \ubH_{ij}^{(l,k)} + \Ex{\{\calS_{pi}\}} \ubH_{ij}^{(l,k)} \ubH_{qr}^{(l_1,k_1)*}.
\end{multline}
Then, summing over $i$, we have
\begin{multline}
    \Ex{\left\{ [\calqS \uqH_{l,k}]_{pj} \uH_{qr}^{(l,k)*}  \right\}} =\frac{1}{n_k} \uT_{jr}^{(l,k)*} \Ex{\{[ \calqS \uqR_{l,k} ]_{pq}\}}  -  \frac{1}{n_k} \Ex{\left\{ \tr(\uqR_{l,k}\calqS) [\calqS \qH \uqT_{l,k}]_{pj} \uH_{qr}^{(l,k)*}\right\}}\\
 - \frac{1}{n_k}  \Ex{\left\{ [\calqS \uqR_{l,k}]_{pq} [\uqT_{l,k}\qH^H\calqS \ubqH_{l,k} ]_{rj} \right\}} + \Ex\{[\calqS \ubqH_{l,k}]_{pj} \} \ubH_{qr}^{(l,k)*},
\end{multline}
and
\begin{multline}
    \Ex{\left\{ [\calqS \uqH_{l,k}]_{pj} \uH_{qr}^{(l_1,k_1)*}  \right\}} =-  \frac{1}{n_k} \Ex{\left\{ \tr(\uqR_{l,k}\calqS) [\calqS \qH \uqT_{l,k}]_{pj} \uH_{qr}^{(l_1,k_1)*}\right\}} \\
      -\frac{1}{n_{k_1}}  \Ex{\left\{ [\calqS \uqR_{l_1,k_1}]_{pq} [\uqT_{l_1,k_1}\qH^H\calqS \ubqH_{l,k} ]_{rj} \calS_{pm} \right\}} + \Ex \left\{[\calqS \ubqH_{l,k}]_{pj} \right\} \ubH_{qr}^{(l_1,k_1)*}.
\end{multline}

Let $\alpha_{l,k} \triangleq \frac{1}{n_k}\tr{\left(\uqR_{l,k}\Ex{\{\calqS\}}\right)} = \frac{1}{n_k}\tr{\left(\qR_{l,k}\Ex{\{\aang{\calqS}_l\}}\right)}$ and
$\aeta_{l,k} \triangleq \frac{1}{n_k} \tr{\left(\uqR_{l,k} \calqS \right)} - \alpha_{l,k}$. Then, we get
\begin{multline}
    \Ex{\left\{ [\calqS \uqH_{l,k}]_{pj} \uH_{qr}^{(l,k)*} \right\}} = \frac{1}{n_k} \uT_{jr}^{(l,k)*} \Ex{\{ [ \calqS \uqR_{l,k} ]_{pq}\}} -  \alpha_{l,k} \Ex{\left\{ [\calqS \qH \uqT_{l,k}]_{pj} \uH_{qr}^{(l,k)*} \right\}} -  \Ex{\left\{ \aeta_{l,k} [\calqS \qH \uqT_{l,k}]_{pj} \uH_{qr}^{(l,k)*}\right\}}\\
    - \frac{1}{n_k}  \Ex{\left\{ [\calqS \uqR_{l,k}]_{pq} [\uqT_{l,k}\qH^H\calqS \ubqH_{l,k} ]_{rj} \right\}} + \Ex\{[\calqS \ubqH_{l,k}]_{pj} \} \ubH_{qr}^{(l,k)*},\label{eq:SHlk}
\end{multline}
and
\begin{multline}
    \Ex{\left\{ [\calqS \uqH_{l,k}]_{pj} \uH_{qr}^{(l_1,k_1)*}  \right\}} = - \alpha_{l,k} \Ex{\left\{ [\calqS \qH \uqT_{l,k}]_{pj} \uH_{qr}^{(l_1,k_1)*}\right\}} - \Ex{\left\{\aeta_{l,k} [\calqS \qH \uqT_{l,k}]_{pj} \uH_{qr}^{(l_1,k_1)*}\right\}} \\
     - \frac{1}{n_{k_1}}  \Ex{\left\{ [\calqS \uqR_{l_1,k_1}]_{pq} [\uqT_{l_1,k_1}\qH^H\calqS \ubqH_{l,k} ]_{rj} \right\}} + \Ex\{[\calqS \ubqH_{l,k}]_{pj} \} \ubH_{qr}^{(l_1,k_1)*}.\label{eq:SHlkl1k1}
\end{multline}
From \eqref{eq:SHlk} and \eqref{eq:SHlkl1k1}, we obtain
\begin{multline}\label{eq:SHh}
    \Ex{\left\{ [\calqS \qH]_{pj} H_{qr}^{*} \right\}} = \sum_{l,k} \frac{1}{n_k} \uT_{jr}^{(l,k)*} \Ex{\{ [ \calqS \uqR_{l,k} ]_{pq}\}} -  \sum_{l,k} \alpha_{l,k} \Ex{\left\{ [\calqS \qH \uqT_{l,k}]_{pj} H_{qr}^{*} \right\}}\\
     - \sum_{l,k} \Ex{\left\{ \aeta_{l,k} [\calqS \qH \uqT_{l,k}]_{pj} H_{qr}^{*}\right\}} - \sum_{l_1,k_1} \frac{1}{n_{k_1}}  \Ex{\left\{ [\calqS \uqR_{l_1,k_1}]_{pq} [\uqT_{l_1,k_1}\qH^H\calqS \bqH ]_{rj} \right\}} + \Ex\{[\calqS \bqH]_{pj} \} \bH_{qr}^{*}.
\end{multline}

Defining
\begin{equation}
     \tqTheta \triangleq \left[ \omega \left(\qI_n + \sum_{l,k} \alpha_{l,k}\uqT_{l,k}\right)\right]^{-1} =\diag\left(\tqTheta_1,\ldots,\tqTheta_K\right),
\end{equation}
where $\tqTheta_k$ is given by \eqref{eq:tildeTheta_k}. Multiplying both sides of \eqref{eq:SHh} by $[\tqTheta_k]_{jr}$ and summing over $j$ and $r$, we get
\begin{multline}
    \Ex{\left\{ [\calqS \qH\qH^H]_{pq} \right\}} =\omega \sum_{l,k} \frac{1}{n_k}\tr( \uqT_{l,k} \tqTheta ) \Ex{\{[\calqS \uqR_{l,k}]_{pq}\}} -  \omega \sum_{l,k} \Ex{\left\{ \aeta_{l,k} [\calqS \qH \uqT_{l,k} \tqTheta \qH^H ]_{pq} \right\}} \\
     - \omega \sum_{l,k} \frac{1}{n_k} \Ex \left\{\tr( \uqT_{l,k} \qH^H \calqS \bqH \tqTheta) [\calqS \uqR_{l,k}]_{pq} \right\}  + \omega \Ex \left\{ [\calqS \bqH \tqTheta \bqH^H ]_{pq} \right\}.
\end{multline}
This, together with (\ref{eq:resolSpq1}), yields
\begin{align}
    \Ex\left\{ \calqS \right\} =& \frac{1}{\omega} \qI_N - \sum_{l,k} \frac{1}{n_k}\tr( \uqT_{l,k} \tqTheta) \Ex{\{\calqS \uqR_{l,k}\}} + \sum_{l,k} \Ex{\left\{ \aeta_{l,k} \calqS \qH \uqT_{l,k} \tqTheta \qH^H \right\}} \nonumber\\
    & + \sum_{l,k} \frac{1}{n_k} \Ex{\left\{ \tr( \uqT_{l,k} \qH^H\calqS \bqH \tqTheta) \calqS \uqR_{l,k}  \right\}} - \Ex{\left\{ \calqS \bqH \tqTheta \bqH^H \right\}} \nonumber\\
    =& \frac{1}{\omega} \qI_N - \sum_{l,k} \frac{1}{n_k}\tr( \uqT_{l,k} \tqTheta) \Ex{\{\calqS \uqR_{l,k}\}} + \sum_{l,k} \Ex{\left\{ \aeta_{l,k} \calqS \qH \uqT_{l,k} \tqTheta \qH^H \right\}} \nonumber\\
    & + \sum_{l,k} \frac{1}{n_k} \Ex{\left\{ \tr(\calqS \bqH \tqTheta \uqT_{l,k} (\tqH + \bqH)^H) \calqS \uqR_{l,k}  \right\}} - \Ex{\left\{ \calqS \bqH \tqTheta \bqH^H \right\}} \nonumber\\
    =& \frac{1}{\omega} \qI_N - \sum_{l,k} \frac{1}{n_k}\tr( \uqT_{l,k} \tqTheta) \Ex{\{\calqS \uqR_{l,k}\}} + \sum_{l,k} \Ex{\left\{ \aeta_{l,k} \calqS \qH \uqT_{l,k} \tqTheta \qH^H \right\}} \nonumber\\
    & + \sum_{l,k} \Ex\left\{  \frac{1}{n_k}\tr\left(\calqS \bqH \tqTheta \uqT_{l,k} \tqH^H \right) \right\} \Ex{\left\{\calqS \uqR_{l,k} \right\}} + \sum_{l,k} \Ex{\left\{\arho_{l,k}^{(1)}\calqS \uqR_{l,k} \right\}}  \nonumber\\
    & + \sum_{l,k} \Ex\left\{  \frac{1}{n_k}\tr\left(\calqS \bqH \tqTheta \uqT_{l,k} \bqH^H \right) \right\} \Ex{\left\{\calqS \uqR_{l,k} \right\}}  + \sum_{l,k} \Ex\left\{\arho_{l,k}^{(2)} \calqS \uqR_{l,k} \right\} - \Ex{\left\{ \calqS \bqH \tqTheta \bqH^H \right\}},
    \label{eq:tempAA1}
\end{align}
where the third equality follows from the following definitions
\begin{subequations}
\begin{align}
    \rho_{l,k}^{(1)} \triangleq \frac{1}{n_k}\tr\left(\calqS \bqH \tqTheta \uqT_{l,k} \tqH^H \right),
    \rho_{l,k}^{(2)} \triangleq \frac{1}{n_k}\tr\left(\calqS \bqH \tqTheta \uqT_{l,k} \bqH^H \right).
\end{align}
\end{subequations}
Before proceeding, we establish the following lemma.

\begin{Lemma}
\begin{align}
    \Ex{\left\{ \frac{1}{n_k}\tr{\left(\calqS \bqH \tqTheta \uqT_{l,k} \tqH^H \right)} \right\}} =& - \omega \sum_{l_1,k_1} \alpha_{l_1,k_1} \Ex{\left\{ \frac{1}{n_k} \tr{\left(\calqS \bqH \tqTheta \uqT_{l,k} \tqTheta \uqT_{l_1,k_1} \bqH^H \right)} \right\}} %\nonumber\\
    - \omega \sum_{l_1,k_1} \Ex{\left\{ \aeta_{l_1,k_1} \arho_{lk,l_1k_1}^{(3)}  \right\}}, \label{eq:lemma SHThertaTH}
\end{align}
where
\begin{equation}\label{eq:rho_3k}
    \rho_{lk,l_1k_1}^{(3)} \triangleq \frac{1}{n_k} \tr{\left(\calqS \bqH \tqTheta \uqT_{l,k} \tqTheta \uqT_{l_1,k_1} \qH^H \right)}.
\end{equation}
\end{Lemma}

\begin{proof}
Using the integration by parts formula (\ref{eq:intByPart}), we write
\begin{align}
    \Ex{\left\{ \calS_{pi} [ \bqH \tqTheta \uqT_{l,k}]_{ij} \tH_{pr}^{*}  \right\}} & =  \sum_{l_1,k_1} \Ex{\left\{ \calS_{pi}\tH_{pr}^{(l_1,k_1)*}\right\}}  [\bqH \tqTheta \uqT_{l,k}]_{ij}   \nonumber\\
     &=   \sum_{l_1,k_1} \frac{1}{n_{k_1}} \sum_{m,n} \uR_{pm}^{(l_1,k_1)*} \uT_{rn}^{(l_1,k_1)} \Ex{\left\{ \frac{\partial \calS_{pi}}{\partial \tH_{mn}^{(l_1,k_1)} } \right\}}   [\bqH \tqTheta \uqT_{l,k}]_{ij}   \nonumber\\
     &= - \sum_{l_1,k_1} \frac{1}{n_{k_1}} \Ex{\left\{ [\calqS \uqR_{l_1,k_1}]_{pp} [\uqT_{l_1,k_1}\qH^H \calqS]_{ri} \right\}} [\bqH \tqTheta \uqT_{l,k}]_{ij}.
\end{align}
Summing over $p$, we have
\begin{align}
    \Ex{\left\{ [\bqH \tqTheta \uqT_{l,k}]_{ij} [\tqH^H\calqS]_{ri} \right\}} = &  - \sum_{l_1,k_1} \frac{1}{n_{k_1}} [\bqH \tqTheta \uqT_{l,k}]_{ij} \Ex{\left\{ \tr(\calqS \uqR_{l_1,k_1}) [\uqT_{l_1,k_1}\qH^H\calqS]_{ri} \right\}}    \nonumber\\
    = & -\sum_{l_1,k_1} \alpha_{l_1,k_1} [\bqH \tqTheta \uqT_{l,k}]_{ij} \Ex{\left\{[\uqR_{l_1,k_1} \tqH^H \calqS]_{ri} \right\}}     \nonumber\\
      &- \sum_{l_1,k_1} \Ex{\left\{ \aeta_{l_1,k_1} [\bqH \tqTheta \uqT_{l,k}]_{ij} [\uqT_{l_1,k_1} \tqH^H \calqS]_{ri} \right\}}     \nonumber\\
      &- \sum_{l_1,k_1} \frac{1}{n_{k_1}} [\bqH \tqTheta \uqT_{l,k}]_{ij} \Ex{\left\{ \tr(\calqS \uqR_{l_1,k_1}) [\uqT_{l_1,k_1} \bqH^H \calqS]_{ri} \right\}} .
\end{align}
After simple algebraic operations and summing over $i, j$ and $l$, we then get
\begin{multline}
    \Ex{\left\{ \tr \left( \calqS \bqH \tqTheta \uqT_{l,k} \tqH^H \right) \right\}}  = - \omega \sum_{l_1,k_1} \Ex{\left\{ \aeta_{l_1,k_1} \tr\left( \calqS \bqH \tqTheta \uqT_{l,k} \tqTheta \uqT_{l_1,k_1} \tqH^H \right)  \right\}}\\
       - \omega \sum_{l_1,k_1} \frac{1}{n_{k_1}} \Ex{\left\{ \tr\left( \calqS \uqT_{l_1,k_1} \right) \tr\left(\calqS \bqH \tqTheta \uqT_{l,k} \tqTheta  \uqT_{l_1,k_1} \bqH^H\right)  \right\}}.
\end{multline}
Therefore, we have
\begin{equation}
    \Ex{\left\{ \frac{1}{n_k} \tr \left(  \calqS \bqH \tqTheta \uqT_{l,k} \tqH^H \right) \right\}} =- \omega \sum_{l_1,k_1} \alpha_{l_1,k_1} \Ex{\left\{ \frac{1}{n_k} \tr\left(\calqS \bqH \tqTheta \uqT_{l,k} \tqTheta \uqT_{l_1,k_1} \bar{\qH}^H\right) \right\}}
    - \omega \sum_{l_1,k_1} \Ex{\left\{ \aeta_{l_1,k_1} \rho_{lk,l_1k_1}^{(3)}  \right\}},
\end{equation}
where $\rho_{lk,l_1k_1}^{(3)}$ is given by \eqref{eq:rho_3k}. Using the fact that $ \Ex{\left\{ \aeta_{l_1,k_1} \rho_{lk,l_1k_1}^{(3)}  \right\}}= \Ex{\left\{
\aeta_{l_1,k_1} \arho_{lk,l_1k_1}^{(3)}  \right\}}$, we obtain \eqref{eq:lemma SHThertaTH}.
\end{proof}

Applying this lemma to \eqref{eq:tempAA1}, we get
\begin{multline}
    \Ex\left\{ \calqS \right\}
    = \frac{1}{\omega} \qI_N - \sum_{l,k} \frac{1}{n_k}\tr{( \uqT_{l,k} \tqTheta)} \Ex{\{\calqS \uqR_{l,k}\}} + \sum_{l,k} \Ex{\left\{ \aeta_{l,k} \calqS \qH \uqT_{l,k} \tqTheta \qH^H \right\}}\\
    + \sum_{l,k} \Ex{\left\{ \frac{1}{n_k}\tr{\left(\calqS \bqH \tqTheta \uqT_{l,k} \bqH^H \right)} \right\}} \Ex{\left\{\calqS \uqR_{l,k} \right\}} - \omega \sum_{l,l_1}^{L}\sum_{k,k_1}^{K}  \Ex{\left\{ \aeta_{l_1,k_1} \arho_{lk,l_1k_1}^{(3)}  \right\}} \Ex{\left\{\calqS \uqR_{l,k} \right\}} \\
    - \omega \sum_{l,l_1}^{L}\sum_{k,k_1}^{K}  \alpha_{l_1,k_1} \Ex{\left\{ \frac{1}{n_k} \tr{\left(\calqS \bqH \tqTheta \uqT_{l,k} \tqTheta \uqT_{l_1,k_1} \bqH^H \right)} \right\}} \Ex{\left\{\calqS \uqR_{l,k} \right\}}  \\
    + \sum_{l,k} \Ex{\left\{\arho_{l,k}^{(1)}\calqS \uqR_{l,k} \right\}} + \sum_{l,k} \Ex{\left\{\arho_{l,k}^{(2)} \calqS \uqR_{l,k} \right\}} - \Ex{\left\{ \calqS \bqH \tqTheta \bqH^H \right\}}.
\end{multline}
Define
\begin{multline}\label{eq:Delta}
    \qDelta \triangleq \sum_{l,k} \Ex{\left\{ \aeta_{l,k} \calqS \qH \uqT_{l,k} \tqTheta \qH^H \right\}} - \omega \sum_{l,l_1}^{L}\sum_{k,k_1}^{K}  \Ex{\left\{ \aeta_{l_1,k_1} \arho_{lk,l_1k_1}^{(3)}  \right\}} \Ex{\left\{\calqS \uqR_{l,k} \right\}} \\
    +  \sum_{l,k} \Ex{\left\{\arho_{l,k}^{(1)}\calqS \uqR_{l,k} \right\}} + \sum_{l,k} \Ex{\left\{\arho_{l,k}^{(2)} \calqS \uqR_{l,k} \right\}}.
\end{multline}
Noting that
\begin{multline}
    \omega  \sum_{l,l_1}^{L}\sum_{k,k_1}^{K} \alpha_{l_1,k_1} \Ex{\left\{ \frac{1}{n_k} \tr{\left(\calqS \bqH \tqTheta \uqT_{l,k} \tqTheta \uqT_{l_1,k_1} \bqH^H \right)} \right\}} \Ex{\left\{\calqS \uqR_{l,k} \right\}}\\
    = \sum_{l,k} \Ex{\left\{ \frac{1}{n_k}\tr{\left(\calqS \bqH \tqTheta \uqT_{l,k} \bqH^H \right)} \right\}} \Ex{\left\{\calqS \uqR_{l,k} \right\}}
     - \omega \sum_{l,k} \Ex{\left\{ \frac{1}{n_k} \tr{\left(\calqS \bqH \tqTheta \uqT_{l,k} \tqTheta \bqH^H \right)} \right\}} \Ex{\left\{\calqS \uqR_{l,k} \right\}},
\end{multline}
we therefore get
\begin{multline}
    \Ex\left\{ \calqS \right\} = \frac{1}{\omega} \qI_N - \sum_{l,k} \frac{1}{n_k}\tr{(\uqT_{l,k} \tqTheta)} \Ex{\{\calqS \uqR_{l,k}\}} + \omega \sum_{l,k} \Ex{\left\{ \frac{1}{n_k} \tr{\left(\calqS \bqH \tqTheta \uqT_{l,k} \tqTheta \bqH^H \right)} \right\}} \Ex{\left\{\calqS \uqR_{l,k} \right\}} \\
     - \Ex{\left\{ \calqS \bqH \tqTheta \bqH^H \right\}} + \qDelta.
\end{multline}
Writing
\begin{equation} \label{eq:tildetau}
    \ttau_{l,k} \triangleq \frac{1}{n_k}\tr{(\uqT_{l,k} \tqTheta)} - \omega \frac{1}{n_k} \tr{\left(\Ex{\left\{\calqS \right\}} \bqH \tqTheta \uqT_{l,k} \tqTheta \bqH^H \right)} = \frac{1}{n_k}\tr{ \left( \uqT_{l,k} \tqTheta\left( \qI_n - \omega  \bqH^H \Ex{\left\{\calqS \right\}} \bqH \tqTheta\right) \right)},
\end{equation}
we have
\begin{equation}
   \Ex{\left\{\calqS \right\}} \left( \qI_N  +  \bqH \tqTheta \bqH^H  \right) = \frac{1}{\omega} \qI_N - \sum_{l,k} \ttau_k \Ex{\left\{\calqS \right\}} \uqR_{l,k} +  \qDelta,
\end{equation}
and then
\begin{equation}
   \Ex{\left\{\calqS \right\}} \left( \qI_N + \sum_{l,k} \talpha_{l,k} \uqR_{l,k} +  \bqH \tqTheta \bqH^H  \right)
    = \frac{1}{\omega} \qI_N + \sum_{l,k} (\talpha_{l,k} - \ttau_{l,k})  \Ex\left\{\calqS \right\} \uqR_{l,k} + \qDelta.
\end{equation}
As a result, we then get
\begin{equation}\label{eq:S}
   \Ex{\left\{\calqS \right\}} =\qXi + \omega \sum_{l,k} (\talpha_{l,k} - \ttau_{l,k})  \Ex\left\{\calqS \right\} \uqR_{l,k}\qXi + \omega\qDelta \qXi,
\end{equation}
where
\begin{equation}\label{eq:qXidef}
    \qXi =  \left( \qI_N + \sum_{l,k} \talpha_{l,k} \uqR_{l,k} +  \bqH \tqTheta \bqH^H  \right)^{-1}
\end{equation}
and
\begin{equation}\label{eq:}
    \talpha_{l,k} \triangleq \frac{1}{n_k}\tr{\left( \uqT_{l,k}\Ex{\{\tcalqS\}}\right)} = \frac{1}{n_k}\tr{\left(\qT_{l,k}\Ex{\{\aang{\tcalqS}_l\}}\right)}.
\end{equation}
To get Proposition \ref{Prop:SXi}, it remains to show that $\talpha_{l,k} - \ttau_{l,k} \rightarrow 0$ and $\tr(\qDelta \qXi) \rightarrow 0$. To that end, we have
to get a similar expression of $\Ex \{\tcalqS \}$ as that of \eqref{eq:S}. Following the same derivation of \eqref{eq:S} from the beginning, we can get
\begin{equation} \label{eq:tildeS}
   \Ex \{\tcalqS \} =  \tqXi +   \omega \sum_{l,k} (\alpha_{l,k} - \tau_{l,k}) \Ex \{ \tcalqS \}\uqT_{l,k} \tqXi + \omega \tqDelta \tqXi,
\end{equation}
where
\begin{subequations}
\begin{align}
%    \tqXi  &\triangleq \left( \qI_n + \sum_{l,k} \alpha_{l,k} \uqT_{l,k} +  \bqH^H \qTheta \bqH  \right)^{-1},\\
    \tqDelta &\triangleq \sum_{l,k} \Ex \left\{ \ateta_{l,k} \tcalqS \qH^H  \uqR_{l,k} \qTheta \qH \right\} - \omega\sum_{l,l_1}^{L}\sum_{k,k_1}^{K} \Ex\left\{ \ateta_{l_1,k_1} \atrho_{lk,l_1k_1}^{(3)} \right\} \Ex\left\{ \tcalqS \uqT_{l,k} \right\} \nonumber \\
               &~~~~+ \sum_{l,k} \Ex\left\{ \atrho_{l,k}^{(1)} \tcalqS \uqT_{l,k} \right\} + \sum_{l,k} \Ex\left\{ \atrho_{l,k}^{(2)} \tcalqS \uqT_{l,k} \right\}, \label{eq:tlideDelta} \\
    \tau_{l,k} &\triangleq \frac{1}{n_k}\tr{ \left( \uqR_{l,k} \qTheta \left( \qI_N - \omega\bqH \Ex{\{\tcalqS\}} \bqH^H \qTheta \right) \right)}, \label{eq:tau}  \\
%    \qTheta &\triangleq \left[ \omega \left(\qI_N +\sum_{l,k}\talpha_{l,k}\uqR_{l,k}\right)\right]^{-1} = \diag\left(\qTheta_1,\ldots,\qTheta_L\right),\\
    \ateta_{l,k} &\triangleq \frac{1}{n_k} \tr\left(\uqT_{l,k}\tcalqS\right) - \talpha_{l,k},\\
    \trho_{l,k}^{(1)} &\triangleq \frac{1}{n_k} \tr\left(\tcalqS \bqH^H \qTheta \uqR_{l,k} \bqH\right), ~~~~
    \trho_{l,k}^{(2)} \triangleq \frac{1}{n_k} \tr\left(\tcalqS \bqH^H \qTheta \uqR_{l,k} \tqH\right),\\
    \trho_{lk,l_1k_1}^{(3)} &\triangleq \frac{1}{n_k} \tr\left( \tcalqS \bqH^H \qTheta \uqR_{l,k}\qTheta \uqR_{l_1,k_1}\qH\right),
\end{align}
\end{subequations}
and those $\tqXi$ and $\qTheta$ are given by \eqref{eq:tildeXi} and \eqref{eq:Theta} respectively.

From \eqref{eq:tildealpha_lk}, \eqref{eq:tildetau}, \eqref{eq:S} and \eqref{eq:tildeS}, write
\begin{equation}
    \talpha_{l,k} =  \frac{1}{n_k} \tr{\left( \uqT_{l,k} \tqXi \right)}+  \frac{\omega}{n_k} \sum_{i,j} (\alpha_{i,j} - \tau_{i,j}) \tr{\left(\uqT_{l,k}\Ex{\{ \tcalqS \}}\uqT_{i,j} \tqXi \right)}    + \frac{\omega}{n_k} \tr{\left(\uqT_{l,k} \tqDelta \tqXi \right)}
\end{equation}
and
\begin{align}
    \ttau_{l,k} = & \frac{1}{n_k} \tr \left( \uqT_{l,k} \tqTheta \left( \qI_n - \omega \bqH^H \qXi \bqH \tqTheta \right) \right) - \frac{\omega^2}{n_k} \sum_{i,j} (\talpha_{i,j} - \ttau_{i,j}) \tr \left( \uqT_{l,k} \tqTheta \bqH^H \Ex \left\{\calqS \right\} \uqR_{i,j} \qXi \bqH \tqTheta \right)  \nonumber \\
    & -  \frac{\omega^2}{n_k}\tr \left( \uqT_{l,k} \tqTheta \bqH^H \qDelta \qXi \bqH \tqTheta \right) \nonumber \\
    =& \frac{1}{n_k}\tr{ \left( \qT_{l,k} \tqXi \right)} - \frac{\omega^2}{n_k} \sum_{i,j} (\talpha_{i,j} - \ttau_{i,j}) \tr{\left( \uqT_{l,k} \tqTheta \bqH^H \Ex\left\{ \calqS \right\} \uqR_{i,j} \qXi \bqH \tqTheta \right)}  -  \frac{\omega^2}{n_k}\tr{ \left( \uqT_{l,k} \tqTheta \bqH^H \qDelta \qXi \bqH \tqTheta \right)} \nonumber \\
    =& \talpha_{l,k} - \frac{\omega}{n_k} \sum_{i,j} (\alpha_{i,j} - \tau_{i,j}) \tr{\left(\uqT_{l,k} \Ex{\{ \tcalqS \}} \uqT_{i,j} \tqXi \right)} - \frac{\omega}{n_k} \tr{\left( \uqT_{l,k} \tqDelta \tqXi \right)} \nonumber \\
    & - \frac{\omega^2}{n_k} \sum_{i,j} (\talpha_{i,j} - \ttau_{i,j})  \tr{\left( \uqT_{l,k} \tqTheta \bqH^H \Ex \left\{ \calqS  \right\} \uqR_{i,j} \qXi \bqH \tqTheta \right)}  -  \frac{\omega^2}{n_k}\tr{ \left( \uqT_{l,k} \tqTheta \bqH^H \qDelta \qXi \bqH \tqTheta \right)}. \label{eq:ttauRec}
\end{align}
Similarly,
\begin{equation}
    \alpha_{l,k}= \frac{1}{n_k} \tr\left(\uqR_{l,k} \qXi \right) + \frac{\omega}{n_k} \sum_{i,j} (\talpha_{i,j} - \ttau_{i,j}) \tr{\left( \uqR_{l,k} \Ex{\left\{ \calqS \right\}}\uqR_{i,j} \qXi  \right)} + \frac{\omega}{n_k} \tr\left(\uqR_{l,k} \qDelta \qXi \right),
\end{equation}
and
\begin{multline}\label{eq:tauRec}
    \tau_{l,k} =\alpha_{l,k} - \frac{\omega}{n_k} \sum_{i,j} (\talpha_{i,j} - \ttau_{i,j}) \tr{\left( \uqR_{l,k} \Ex{\left\{ \calqS \right\}}\uqR_{i,j} \qXi  \right)} - \frac{\omega}{n_k} \tr\left(\uqR_{l,k} \qDelta \qXi \right)\\
    - \frac{\omega^2}{n_k} \sum_{i,j} (\alpha_{i,j} - \tau_{i,j}) \tr{ \left( \uqR_{l,k} \qTheta \bqH \Ex \{ \tcalqS \} \uqT_{i,j} \tqXi \bqH^H \qTheta \right)} - \frac{\omega^2}{n_k}\tr{ \left( \uqR_{l,k} \qTheta \bqH \tqDelta \tqXi \tqH^H \qTheta  \right)}.
\end{multline}

Let $\qeta \triangleq \left[ {\tt vec}(\qA_1)^T, {\tt vec}(\qA_2)^T\right]^T$, $\qepsilon \triangleq \left[{\tt vec}(\qC_1)^T, {\tt vec}(\qC_2)^T \right]^T$,
$\qGamma \triangleq \left[
                      \begin{array}{cc}
                          \qGamma_{11} & \qGamma_{12} \\
                          \qGamma_{21} & \qGamma_{22}
                      \end{array} \right]$,
where $\qA_1, \qA_2, \qC_1, \qC_2 \in \mathbb{C}^{L \times K}, \qGamma_{11},\qGamma_{12}, \qGamma_{21}, \qGamma_{22} \in \mathbb{C}^{LK \times LK}$, with
\begin{subequations}
\begin{align}
    [\qA_1]_{l,k}&= \tilde\alpha_{l,k} - \tilde\tau_{l,k}, ~~~ [\qA_2]_{l,k} = \alpha_{l,k} - \tau_{l,k},\\
    [\qC_1]_{l,k}&= \frac{\omega}{n_k} \tr\left(\uqT_{l,k} \tqDelta \tqXi \right) + \frac{\omega^2}{n_k}\tr{ \left( \uqT_{l,k} \tqTheta \bqH^H \qDelta \qXi \bqH \tqTheta \right)}, \\
    [\qC_2]_{l,k}&= \frac{\omega}{n_k} \tr\left(\uqR_{l,k} \qDelta \qXi  \right) +  \frac{\omega^2}{n_k}\tr{ \left( \uqR_{l,k} \qTheta \bqH \tqDelta \tqXi \bqH^H \qTheta \right)},\\
    [\qGamma_{11}]_{lk,ij} &= \left\{
    \begin{array}{ll} -\frac{\omega^2}{n_k} \tr{\left( \uqT_{l,k} \tqTheta \bqH^H \Ex\left\{ \calqS \right\} \uqR_{i,j} \qXi \bqH \tqTheta \right)},  & (i,j)\neq (l,k); \\
     1-\frac{\omega^2}{n_k} \tr{\left( \uqT_{l,k} \tqTheta \bqH^H \Ex\left\{ \calqS \right\} \uqR_{l,k} \qXi \bqH \tqTheta \right)},  & (i,j)= (l,k) ,
    \end{array}
    \right.\\
    [\qGamma_{12}]_{lk,ij} &= -\frac{\omega}{n_k} \tr{\left(\uqT_{l,k} \Ex{ \{ \tcalqS \}}\uqT_{i,j} \tqXi \right)}, ~~~~
    [\qGamma_{21}]_{lk,ij} = -\frac{\omega}{n_k} \tr{\left(\uqR_{l,k} \Ex{\left\{ \calqS \right\}}\uqR_{i,j} \qXi  \right)},\\
    [\qGamma_{22}]_{lk,ij} &= \left\{
    \begin{array}{ll}
     -\frac{\omega^2}{n_k} \tr{ \left( \uqR_{l,k} \qTheta \bqH \Ex \{ \tcalqS \} \uqT_{i,j} \tqXi \bqH^H \qTheta \right)},  & (i,j)\neq (l,k); \\
     1-\frac{\omega^2}{n_k}  \tr{ \left( \uqR_{l,k} \qTheta \bqH \Ex \{ \tcalqS \} \uqT_{l,k} \tqXi \bqH^H \qTheta \right)},  &(i,j) = (l,k).
    \end{array}
    \right.
\end{align}
\end{subequations}

From (\ref{eq:ttauRec}) and (\ref{eq:tauRec}), we get
\begin{equation}\label{eq:qeta1}
    \qGamma \qeta = \qepsilon.
\end{equation}
If we can show that $\qepsilon \rightarrow \qzero$ and $\qGamma$ is invertible, we then get our desired result $\qeta \rightarrow \qzero$. To show that $\qepsilon
\rightarrow \qzero$, we establish the following lemma.

\begin{Lemma}
For any uniformly bounded matrices $\qQ$ and $\tqQ$, we have
\begin{subequations}
\begin{align}
    \frac{1}{n_k} \tr( \qDelta \qQ ) &= O\left(\frac{1}{N^2}\right), \label{eq:DeltaON2}\\
    \frac{1}{n_k} \tr( \tqDelta \tqQ ) &= O\left(\frac{1}{N^2}\right).\label{eq:tDeltaON2}
\end{align}
\end{subequations}
\end{Lemma}

\begin{proof}
From \eqref{eq:Delta}, we write
\begin{multline} \label{eq:trDeltaQ}
    \frac{1}{n_k} \tr\left( \qDelta \qQ \right) = \sum_{l,k} \Ex{\left\{ \aeta_{l,k} \arho_{l,k}^{(4)} \right\}} - \omega \sum_{l,l_1}^{L}\sum_{k,k_1}^{K}  \Ex{\left\{ \aeta_{l_1,k_1} \arho_{lk,l_1k_1}^{(3)} \right\}} \frac{1}{n_k}\tr\left(\Ex{\left\{\calqS \right\}}\underline{\qR}_{l,k}\qQ \right)\\
       +  \sum_{l,k} \Ex{\left\{(\arho_{l,k}^{(1)} + \arho_{l,k}^{(2)}) \frac{1}{n_k} \tr\left(\calqS \uqR_{l,k} \qQ \right)\right\}},
\end{multline}
where
\begin{equation}
    \rho_{l,k}^{(4)}=\frac{1}{n_k}\tr\left(\calqS \qH \uqT_{l,k} \tqTheta \qH^H\qQ \right).
\end{equation}

We first prove the following facts for any uniformly bounded matrices $\qM$,
\begin{equation}\label{eq:trSM}
    \Varx\left(\frac{1}{n_k} \tr\left(\calqS \qM \right) \right)= O\left(\frac{1}{N^2}\right).
\end{equation}
For this, we let $\Gamma\left(\qH \right) \triangleq \frac{1}{n_k} \tr\left(\calqS \qM \right)$ which gives
\begin{subequations}
\begin{align}
    \frac{\partial \Gamma\left(\qH \right)}{\partial \utH_{mn}^{(l,k)}} =& \frac{1}{n_k} \sum_{p,q} M_{pq} \frac{\partial \calS_{pq}}{\partial \utH_{mn}^{(l,k)}}
     =  -\frac{1}{n_k} [\qH^H \calqS \qM \calqS]_{nm},   \\
    \frac{\partial \Gamma\left(\qH \right)}{\partial \utH_{mn}^{(l,k)*}} =& \frac{1}{n_k} \sum_{p,q} M_{pq} \frac{\partial \calS_{pq}}{\partial \utH_{mn}^{(l,k)*}}
     =  -\frac{1}{n_k} [\calqS \qM \calqS\qH]_{mn}.
\end{align}
\end{subequations}
Using Lemma \ref{Lma:PNineq} (the Poincar\'{e}-Nash inequality), we obtain
\begin{align}
    \Varx\left(\frac{1}{n_k} \tr\left(\calqS \qM \right) \right) \leq & \sum_{l,k} \frac{1}{n_k} \sum_{m,n} \sum_{m',n'} \uR_{mm'}^{(l,k)} \uT_{nn'}^{(l,k)*} \Ex \left\{ \frac{1}{n_k} [\qH^H \calqS \qM \calqS]_{nm} \frac{1}{n_k} [\qH^H \calqS \qM \calqS]_{n'm'}^* \right\} \nonumber \\
      &+ \sum_{l,k} \frac{1}{n_k} \sum_{m,n} \sum_{m',n'} \uR_{mm'}^{(l,k)*} \uT_{nn'}^{(l,k)} \Ex \left\{ \frac{1}{n_k} [\calqS \qM \calqS\qH]_{mn}^* \frac{1}{n_k} [\calqS \qM \calqS\qH]_{n'm'} \right\} \nonumber \\
     = & \sum_{l,k} \frac{1}{n_k^3} \Ex \left\{ \tr \left(\qH^H \calqS \qM \calqS \uqR_{l,k} \calqS \qM^H \calqS \qH \uqT_{l,k}\right)  \right\}   \nonumber \\
       &+ \sum_{l,k} \frac{1}{n_k^3} \Ex \left\{ \tr \left(\qH^H \calqS \qM^H \calqS \uqR_{l,k} \calqS \qM \calqS \qH \uqT_{l,k}\right)    \right\}.
\end{align}
Noting the fact that (using $\|\calqS\|\leq \frac{1}{\omega}$, Lemma \ref{Lemma:TraceNorm}, and Lemma \ref{Lemma:SpNormHHupb})
\begin{equation}
    \Ex \left\{\tr \left(\qH^H \calqS \qM \calqS \uqR_{l,k} \calqS \qM^H \calqS \qH \uqT_{l,k}\right) \right\} \leq \frac{2LK\|\qM\|^2C_{\rm max}^2 N}{\omega^4},
\end{equation}
we get
\begin{equation}
    \Varx\left(\frac{1}{n_k} \tr\left(\calqS \qM \right) \right) \leq  \frac{4L^2K^2\|\qM\|^2C_{\rm max}^2 N}{\omega^4 n_k^3}
    = O\left(\frac{1}{N^2}\right).
\end{equation}
It turns out that \eqref{eq:trSM} holds and thus implies that $\Ex{\left\{ \aeta_{l,k}^2 \right\}} = O\left(\frac{1}{N^2}\right)$. Similarly, based on the
Poincar\'{e}-Nash inequality, we have $\Ex{\left\{ {\arho_{l,k}^{(4)}}^2 \right\}} = O\left(\frac{1}{N^2}\right)$. The Cauchy-Schwarz inequality provides the first
term of the right-hand side of \eqref{eq:trDeltaQ} which is a $O\left(\frac{1}{N^2}\right)$ term. Similar calculations allow to show the second and third terms of
the right-hand side of \eqref{eq:trDeltaQ} giving the $O\left(\frac{1}{N^2}\right)$ terms. Therefore, we obtain \eqref{eq:DeltaON2}. Similarly,
\eqref{eq:tDeltaON2} can be proved and the proof is omitted.
\end{proof}

From this lemma, it can be shown that $\qepsilon = O\left(\frac{1}{N^2}\right)\qone$. In addition, we note that
\begin{equation}
    \|\calqS\|, \|\tcalqS\|, \|\qXi\|, \|\tqXi\|, \|\qTheta_l\|, \|\tqTheta_k\| \leq \frac{1}{\omega}.
\end{equation}
Using \eqref{eq:asumRTH}, Lemma \ref{Lemma:TraceNorm} and Lemma \ref{Lemma:SpNormHHupb}, we have
\begin{subequations}
\begin{align}
    [\qGamma_{11}]_{lk,ij} &\geq \left\{
    \begin{array}{ll}
      -\frac{N_l}{n_k} \frac{LK C_{\rm max}^3}{\omega^2},  &(i,j)\neq (l,k), \\
     1-\frac{N_l}{n_k} \frac{LK C_{\rm max}^3}{\omega^2},  & (i,j) = (l,k),
    \end{array}
    \right.\\
    [\qGamma_{12}]_{lk,ij} &\geq  -\frac{C_{\rm max}^2}{\omega}, ~~[\qGamma_{21}]_{lk,ij} \geq -\frac{N_l}{n_k} \frac{C_{\rm max}^2}{\omega},\\
    [\qGamma_{22}]_{lk,ij} &\geq\left\{
    \begin{array}{ll}
      -\frac{N_l}{n_k} \frac{LK C_{\rm max}^3}{\omega^2},  & (i,j)\neq (l,k), \\
      1-\frac{N_l}{n_k} \frac{LK C_{\rm max}^3}{\omega^2},  & (i,j) = (l,k).
    \end{array}
    \right.
\end{align}
\end{subequations}
It is possible to choose $\omega_0$ such that $\omega > \omega_0$ and $\qGamma$ is a strictly diagonally dominant. Thus the eigenvalues of $\qGamma$ are bounded
away from 0 \cite[Theorem 6.1.10]{Horn-90}. It implies that if $\omega > \omega_0$, then $(\alpha_{l,k} - \tau_{l,k})$'s and $(\tilde\alpha_{l,k} -
\tilde\tau_{l,k})$'s are of the same order of magnitude as $O\left(\frac{1}{N^2}\right)$, and therefore converge to 0 when $\largeN \rightarrow \infty$.

In the remaining part, we aim to prove that this convergence still holds for $0 < \omega \leq \omega_0$. Firstly, considering $\alpha_{l,k}$ and
$\tau_{l,k}$ as functions of the parameter $z=-\omega \in \bbR^-$, we extend their domain of validity from $\bbR^-$ to $\bbC-\bbR^+$. Similarly to
\cite[Proposition 11]{Dupuy-11IT}, we have the following lemma.
\begin{Lemma}
$\alpha_{l,k}$ and $\tau_{l,k}$ are analytic over $\bbC-\bbR^+$ and belong to $\bbS(\bbR^+)$ with  $|\alpha_{l,k}| \leq \frac{ \frac{1}{n_k}\tr \uqR_{l,k} }{d(z,\bbR^+)}=\frac{N_lP_{l,k} }{n_k(\kappa_{l,k}+1)d(z,\bbR^+)}$ and $|\tau_{l,k}| \leq \frac{N_lP_{l,k} }{n_k(\kappa_{l,k}+1)d(z,\bbR^+)} \left(1+\frac{|z|LKC_{\rm max}}{(d(z,\bbR^+))^2 }\right)$, where $\bbS(\bbR^+)$ is the class of all Stieltjes transforms of finite positive measures carried by $\bbR^+$.
\end{Lemma}
\begin{proof}
We only prove the results for $\alpha_{l,k}$ since the proof of results of $\tau_{l,k}$ is similar. From the definition of $\calqS$, $\calqS$ is invertible for every $z \in \bbC-\bbR^+$
and $\Ex \{\calqS \}$ is analytic over $\bbC-\bbR^+$. Thus $\alpha_{l,k}$ is analytic over $\bbC-\bbR^+$. Using the fact that $\calqS \preceq \frac{1}{d(z,\bbR^+)}\qI_n$ and Lemma \ref{Lemma:TraceNorm}, we have
\begin{align}
 |\alpha_{l,k}| =  \left|\frac{1}{n_k} \tr (\uqR_{l,k} \Ex \{\calqS \})\right|  \leq  \frac{1}{n_k} \| \Ex \{\calqS \} \| \tr \uqR_{l,k}
             \leq  \frac{\frac{1}{n_k}  \tr \uqR_{l,k}}{d(z,\bbR^+)} = \frac{N_lP_{l,k} }{n_k(\kappa_{l,k}+1)d(z,\bbR^+)}, \nonumber
\end{align}
where the last equality is obtained by $(4)$. In order to state $\alpha_{l,k} \in \bbS(\bbR^+)$, we only check the following three
conditions by \cite[Proposition 10]{Dupuy-11IT}: $1)~\Im\{\alpha_{l,k}(z)\}>0$ if $\Im\{z\}>0$; $2)~\Im\{z \alpha_{l,k}(z)\}>0$ if $\Im\{z\}>0$;
$3)~\lim_{y\rightarrow \infty}|{\sf j}y\alpha_{l,k}({\sf j}y)|<\infty$.

Let us first compute $\Im\{\alpha_{l,k}(z)\}$: For every $z \in
\bbC^+$,
\begin{align}
 \Im\{\alpha_{l,k}(z)\}    = & \Im\left\{\frac{1}{n_k} \tr \left(\uqR_{l,k} \Ex \left\{\calqS \left( \qH^H\qH - z^*\qI_N \right) \calqS^H \right\}\right) \right\}   \nonumber\\
                           = & \Im\left\{\frac{1}{n_k} \tr \left(\uqR_{l,k} \Ex \left\{\calqS \qH^H\qH \calqS^H \right\}\right) \right\} - \Im\left\{\frac{1}{n_k} z^* \tr \left(\uqR_{l,k} \Ex \left\{\calqS\calqS^H\right\}\right) \right\} \nonumber\\
                           = & - \frac{1}{n_k}\Im\left\{ z^* \right\} \tr \left(\uqR_{l,k} \Ex \left\{\calqS \calqS^H\right\}\right) >0.  \nonumber
\end{align}
By similar arguments above, we can prove that $\Im\{z \alpha_{l,k}(z)\}>0$ if $\Im\{z\}>0$. Next, we calculate
\begin{align}
 \lim_{y\rightarrow \infty}|{\sf j}y\alpha_{l,k}({\sf j}y)| %= & \lim_{y\rightarrow \infty} \left|{\sf j}y \frac{1}{n_k} \tr \left(\uqR_{l,k} \Ex \left\{\left( \qH\qH^H - {\sf j}y \qI_N \right)^{-1} \right\}\right)\right| \nonumber\\
                                                            = & \lim_{y\rightarrow \infty} \left|\frac{1}{n_k} \tr \left(\uqR_{l,k} \Ex \left\{ \left( \frac{1}{{\sf j}y} \qH\qH^H - \qI_N \right)^{-1}  \right\} \right)\right| \nonumber\\
                                                            = & \frac{1}{n_k} \tr \uqR_{l,k} < \infty. \nonumber
\end{align}
Since the three sufficient conditions have been verified, we have
$\alpha_{l,k} \in \bbS(\bbR^+)$.
\end{proof}

Using this lemma, $|\alpha_{l,k}-\tau_{l,k}| \leq \frac{N_lP_{l,k} }{n_k(\kappa_{l,k}+1)d(z,\bbR^+)} \left(2+\frac{|z|LKC_{\rm max}}{(d(z,\bbR^+))^2 }\right)$. Moreover,
$\{\alpha_{l,k}-\tau_{l,k}\}_{\forall l,k}$ is a family of analytic functions. By Montel's theorem \cite{Cartan78}, this convergence still holds for
$0 < \omega \leq \omega_0$, and that \eqref{eq:Proposition1a} and \eqref{eq:Proposition1b} hold true.

\subsection{Proof of Proposition \ref{Prop:XiPsi}}\label{Appendix:Proof of Proposition 2}
Using the resolvent identity (Lemma \ref{Lemma:ResolventIdentity}) $\qXi - \qPsi=\qXi\left(\qPsi^{-1} - \qXi^{-1}\right)\qPsi$, we have
\begin{equation}
    \qXi - \qPsi = \omega \qXi \diag\left( \left\{\sum_{k=1}^K {(\te_{l,k} - \talpha_{l,k}) \qR_{l,k}}\right\}_{\forall l}\right)\qPsi
     + \omega^2 \qXi \bqH \diag\left( \left\{\sum_{l=1}^L {(\alpha_{l,k} - \beta_{l,k} e_{l,k}) \tqPhi_k \qT_{l,k} \tqTheta_k}\right\}_{\forall k}\right) \bqH^H \qPsi.
\end{equation}
Similarly,
\begin{equation}
    \tqXi - \tqPsi = \omega \tqXi \diag\left( \left\{\sum_{l=1}^L {( \beta_{l,k} e_{l,k} - \alpha_{l,k}) \qT_{l,k}}\right\}_{\forall k}\right) \tqPsi
     + \omega^2 \tqXi \bqH^H \diag\left( \left\{\sum_{k=1}^K {(\talpha_{l,k} - \te_{l,k}) \qPhi_l \qR_{l,k} \qTheta_l}\right\}_{\forall l}\right) \bqH \tqPsi.
\end{equation}
Taking the trace, we get
\begin{subequations}
\begin{align}
    \tr \left(\qXi - \qPsi\right) = & \omega \sum_{l,k} (\te_{l,k} - \talpha_{l,k})  \tr\left( \qXi \uqR_{l,k} \qPsi \right)
      + \omega^2 \sum_{l,k} (\alpha_{l,k} - \beta_{l,k} e_{l,k}) \tr\left( \qXi \bqH \tqPhi \uqT_{l,k} \tqTheta \bqH^H \qPsi\right),   \label{eq:traceXiPsi} \\
%\end{align}
%and
%\begin{align}
    \tr \left(\tqXi - \tqPsi\right) =& \omega \sum_{l,k}  (\beta_{l,k} e_{l,k} - \alpha_{l,k})  \tr\left( \tqXi \uqT_{l,k} \tqPsi \right)
      + \omega^2 \sum_{l,k} ( \talpha_{l,k} - \te_{l,k}) \tr\left(\tqXi \bqH^H \qPhi \uqR_{l,k} \qTheta \bqH \tqPsi\right).  \label{eq:tracetXiPsi}
\end{align}
\end{subequations}
From Proposition \ref{Prop:SXi}, we have
\begin{subequations}
\begin{align}
    \alpha_{l,k} =&\frac{1}{n_k} \tr ( \qR_{l,k} \aang{\qXi}_{l} ) + \varepsilon_{l,k},  \label{eq:varepsilon}\\
    \talpha_{l,k}=&\frac{1}{n_k} \tr ( \qT_{l,k} \ang{\tqXi}_{k} ) + \tvarepsilon_{l,k}, \label{eq:tvarepsilon}
\end{align}
\end{subequations}
where $\varepsilon_{l,k}$ and $\tvarepsilon_{l,k}$ converge towards $0$. Therefore,
\begin{subequations}
\begin{align}
    \alpha_{l,k} -&\beta_{l,k} e_{l,k} =\frac{1}{n_k} \tr\left( \qR_{l,k} \aang{\qXi-\qPsi}_{l} \right) + \varepsilon_{l,k}  \nonumber \\
       =& \frac{\omega}{n_k} \sum_{i,j} (\te_{i,j} - \talpha_{i,j}) \tr\left( \uqR_{l,k} \qXi \uqR_{i,j} \qPsi \right) %\nonumber \\
     + \frac{\omega^2}{n_k} \sum_{i,j} (\alpha_{i,j} - \beta_{i,j} e_{i,j}) \tr\left( \uqR_{l,k} \qXi \bqH \tqPhi \qT_{i,j} \tqTheta \bqH^H \qPsi\right) + \varepsilon_{l,k}, \label{eq:alphae}\\
%\end{align}
%and
%\begin{align}
    \tilde\alpha_{l,k}-&\tilde e_{l,k} =\frac{1}{n_k} \tr\left( \qT_{l,k} \ang{\tilde\qXi-\tilde\qPsi}_{k} \right) + \tilde\varepsilon_{l,k}  \nonumber \\
       =& \frac{\omega}{n_k} \sum_{i,j} (\beta_{i,j} e_{i,j} - \alpha_{i,j}) \tr\left(\uqT_{l,k} \tqXi \uqT_{i,j} \tqPsi \right) %\nonumber
     + \frac{\omega^2}{n_k} \sum_{i,j} (\talpha_{i,j} - \te_{i,j}) \tr\left(\uqT_{l,k} \tqXi \bqH^H \qPhi \uqR_{i,j} \qTheta \bqH \tqPsi\right) + \tvarepsilon_{l,k}. \label{eq:talphae}
\end{align}
\end{subequations}
Using the same approach as in the proof in Proposition \ref{Prop:SXi}, we prove that $(\alpha_{l,k} -\beta_{l,k} e_{l,k})$'s and $(\talpha_{l,k}-\te_{l,k})$'s
converge towards $0$. From \eqref{eq:traceXiPsi} and \eqref{eq:tracetXiPsi}, we complete the proof of Proposition \ref{Prop:XiPsi}.

\subsection{Proof of Proposition \ref{Prop:SXiPsi_order}}\label{Appendix:Proof of Proposition 3}
We first establish \eqref{eq:Proposition3a} and \eqref{eq:Proposition3b}. The equations (\ref{eq:ttauRec}) and (\ref{eq:tauRec}) can be rewritten as
\begin{subequations}
\begin{align}
    \talpha_{l,k}-\ttau_{l,k}& = \frac{\omega}{n_k} \sum_{i,j} \left(\alpha_{i,j}-\tau_{i,j}\right) \tr\left(\uqT_{l,k} \Ex \{\tcalqS \}\uqT_{i,j} \tqXi \right) + \frac{\omega}{n_k}\tr\left(\uqT_{l,k} \tqDelta \tqXi \right) \nonumber\\
      &+ \frac{\omega^2}{n_k} \sum_{i,j} \left(\talpha_{i,j}-\ttau_{i,j}\right) \tr\left(\uqT_{l,k} \tqTheta \bqH^H \Ex\left\{\calqS \right\} \uqR_{i,j} \qXi \bqH \tqTheta \right)  +  \frac{\omega^2}{n_k} \tr\left(\uqT_{l,k} \tqTheta \bqH^H \qDelta \qXi \bqH \tqTheta \right), \label{eq:talphattau}\\
    \alpha_{l,k}-\tau_{l,k} &= \frac{\omega}{n_k} \sum_{i,j} \left(\talpha_{i,j}-\ttau_{i,j}\right) \tr\left(\uqR_{l,k} \Ex\left\{\calqS \right\}\uqR_{i,j} \qXi \right) + \frac{\omega}{n_k}\tr\left(\uqR_{l,k} \qDelta \qXi \right)\nonumber\\
      &+ \frac{\omega^2}{n_k} \sum_{i,j} \left(\alpha_{i,j}-\tau_{i,j}\right) \tr\left(\uqR_{l,k} \qTheta \bqH \Ex \{\tcalqS \} \uqT_{i,j} \tqXi \bqH^H \qTheta \right)  +  \frac{\omega^2}{n_k} \tr\left(\uqR_{l,k} \qTheta \bqH \tqDelta \tqXi \bqH^H \qTheta \right). \label{eq:alphatau}
\end{align}
\end{subequations}
We can write these two equations in matrix form:
\begin{equation}\label{eq:qeta2}
    \qeta = \qGamma' \qeta + \qepsilon',
\end{equation}
where $\qepsilon' \triangleq \left[{\tt vec}(\qC'_1)^T, {\tt vec}(\qC'_2)^T \right]^T,
 \qGamma' \triangleq \left[
                      \begin{array}{cc}
                          \qGamma'_{11} & \qGamma'_{12} \\
                          \qGamma'_{21} & \qGamma'_{22}
                      \end{array} \right]$, with $\qC'_1, \qC'_2 \in \bbC^{L \times K}, \qGamma'_{11},\qGamma'_{12}, \qGamma'_{21}, \qGamma'_{22} \in \bbC^{LK \times LK}$, and
\begin{subequations}
\begin{align}
    [\qC'_1]_{l,k}&= \frac{\frac{\omega}{n_k} \tr\left( \uqT_{l,k} \tqDelta \tqXi \right) + \frac{\omega^2}{n_k}\tr \left( \uqT_{l,k} \tqTheta \bqH^H \qDelta \qXi \bqH \tqTheta \right)}{1-\frac{\omega^2}{n_k} \tr \left( \uqT_{l,k} \tqTheta \bqH^H \Ex\left\{\calqS \right\} \uqR_{l,k} \qXi \bqH \tqTheta \right)},\\
    [\qC'_2]_{l,k}&= \frac{\frac{\omega}{n_k} \tr\left( \uqR_{l,k} \qDelta \qXi  \right) + \frac{\omega^2}{n_k}\tr \left( \uqR_{l,k} \qTheta \bqH \tqDelta \tqXi \bqH^H \qTheta \right)}{1-\frac{\omega^2}{n_k} \tr \left( \uqR_{l,k} \qTheta \bqH \Ex \{ \tcalqS \} \uqT_{l,k} \tqXi \bqH^H \qTheta \right)},\\
    [\qGamma'_{11}]_{lk,ij} &= \left\{
    \begin{aligned} 0, ~~~~~~~~~~~~~~~~~~~~~~~~ & ~\mbox{for }(i,j)\neq (l,k); \\
     \frac{\frac{\omega^2}{n_k}\tr \left( \uqT_{l,k} \tqTheta \bqH^H \Ex\left\{\calqS \right\} \uqR_{i,j} \qXi \bqH \tqTheta \right)}{1-\frac{\omega^2}{n_k} \tr \left( \uqT_{l,k} \tqTheta \bqH^H \Ex\left\{\calqS \right\} \uqR_{l,k} \qXi \bqH \tqTheta \right)},  & ~\mbox{for }(i,j)= (l,k) ,
    \end{aligned}
    \right.\\
%\end{align}
%\begin{align}
    [\qGamma'_{12}]_{lk,ij} &=  \frac{ \frac{\omega}{n_k} \tr \left(\uqT_{l,k} \Ex \{ \tcalqS \} \uqT_{i,j} \tqXi \right)}{1-\frac{\omega^2}{n_k} \tr \left( \uqT_{l,k} \tqTheta \bqH^H \Ex\left\{\calqS \right\} \uqR_{l,k} \qXi \bqH \tqTheta \right)}, \\
    [\qGamma'_{21}]_{lk,ij} &=  \frac{ \frac{\omega}{n_k} \tr \left(\uqR_{l,k} \Ex \left\{\calqS \right\}\uqR_{i,j} \qXi \right)}{1-\frac{\omega^2}{n_k} \tr \left( \uqR_{l,k} \qTheta \bqH \Ex \{ \tcalqS \} \uqT_{l,k} \tqXi \bqH^H \qTheta \right)},\\
    [\qGamma'_{22}]_{lk,ij} &= \left\{
    \begin{aligned} 0,~~~~~~~~~~~~~~~~~~~~~~~~   & ~\mbox{for }(i,j)\neq (l,k); \\
     \frac{\frac{\omega^2}{n_k} \tr \left( \uqR_{l,k} \qTheta \bqH \Ex \{ \tcalqS \} \uqT_{i,j} \tqXi \bqH^H \qTheta \right)}{1-\frac{\omega^2}{n_k} \tr \left( \uqR_{l,k} \qTheta \bqH \Ex \{ \tcalqS \} \uqT_{l,k} \tqXi \bqH^H \qTheta \right)},  & ~\mbox{for }(i,j) = (l,k).
    \end{aligned}
    \right.
\end{align}
\end{subequations}

Let $\qGamma''$ be the matrix by replacing $\Ex \{\calqS \}, \Ex \{\tcalqS \}, \qXi, \tqXi, \qTheta$ and $\tqTheta$ in $\qGamma'$ with $\qPsi,\tqPsi,\qPsi,\tqPsi,
\qPhi$ and $\tqPhi$, respectively. Using Propositions \ref{Prop:SXi} and \ref{Prop:XiPsi}, we immediately obtain
\begin{equation}\label{eq:qGammatoqGamma}
    \qGamma' = \qGamma'' + \qdelta,
\end{equation}
where all entries of $\qdelta$ converge to $0$ as $\largeN \rightarrow \infty$, and $\qGamma''$ is given by
\begin{equation}\label{eq:qGammapp}
\qGamma'' = \left[
                      \begin{array}{cc}
                          \qGamma''_{11} & \qGamma''_{12} \\
                          \qGamma''_{21} & \qGamma''_{22}
                      \end{array} \right],
\end{equation}
with $\qGamma''_{11},\qGamma''_{12}, \qGamma''_{21}, \qGamma''_{22} \in \bbC^{LK \times LK}$, and
\begin{subequations}
\begin{align}
    [\qGamma''_{11}]_{lk,ij} &= \left\{
    \begin{aligned} 0, ~~~~~~  & ~\mbox{for }(i,j)\neq (l,k); \\
     \frac{u^{(2)}_{lk,ij}}{1-u^{(2)}_{lk,lk}},  & ~\mbox{for }(i,j)= (l,k) ,
    \end{aligned}
    \right. ~~~
%\end{align}
%\begin{align}
    [\qGamma''_{12}]_{lk,ij} =  \frac{v^{(1)}_{lk,ij}}{1-u^{(2)}_{lk,lk}},\\
     [\qGamma''_{21}]_{lk,ij}& =  \frac{u^{(1)}_{lk,ij}}{1-v^{(2)}_{lk,lk}},
    ~~~[\qGamma''_{22}]_{lk,ij} = \left\{
    \begin{aligned} 0, ~~~~~~ & ~\mbox{for }(i,j)\neq (l,k); \\
     \frac{v^{(2)}_{lk,ij}}{1-v^{(2)}_{lk,lk}},  & ~\mbox{for }(i,j) = (l,k),
    \end{aligned}
    \right.\\
    u^{(1)}_{lk,ij} &=  \frac{\omega}{n_k} \tr \left( \uqR_{l,k} \qPsi \uqR_{i,j} \qPsi  \right),
    ~~~u^{(2)}_{lk,ij} =  \frac{\omega^2}{n_k}\tr \left( \uqT_{l,k} \tqPhi \bqH^H \qPsi \uqR_{i,j} \qPsi \bqH \tqPhi \right), \\
    v^{(1)}_{lk,ij} &=  \frac{\omega}{n_k} \tr \left(\uqT_{l,k} \tqPsi \uqT_{i,j} \tqPsi \right),
    ~~~v^{(2)}_{lk,ij} =  \frac{\omega^2}{n_k} \tr \left( \uqR_{l,k} \qPhi \bqH \tqPsi \uqT_{i,j} \tqPsi \bqH^H \qPhi \right).
\end{align}
\end{subequations}

\begin{Lemma}
Let $\qGamma''$ be the matrix defined by \eqref{eq:qGammapp}. Then, we have
\begin{subequations}
\begin{align}
    \sup_N \left[\rho\left(\qGamma'' \right)\right] \leq 1 - \frac{\lambda_0\omega^2}{\left(\omega+\lambda'_0\right)^2} & < 1, \label{eq:suprho_qGamma}  \\
    \sup_N \left[ \mnorm{\left(\qI- \qGamma''\right)^{-1}}_{\infty}  \right] &\leq \frac{\left(\omega+\lambda'_0\right)^2}{\lambda_0\omega^2}, \label{eq:supinvqGamma}
\end{align}
\end{subequations}
for some constants $\lambda_0,\lambda'_0$.
\end{Lemma}

\begin{proof}
From \eqref{eq:Solutionete}, a direct calculation yields
\begin{align}
     \beta_{l,k} e_{l,k} =&   \frac{1}{n_k} \tr\left( \uqR_{l,k} \qPsi \qPsi^{-1} \qPsi\right)  \nonumber \\
      \mathop =\limits^{(i)} &  \frac{1}{n_k} \tr\left( \uqR_{l,k} \qPsi \left( \omega\qI_N + \omega \sum_{i,j} \te_{i,j} \uqR_{i,j} + \omega \bqH \tqPhi \bqH^H\right) \qPsi\right)  \nonumber \\
      = &   \frac{1}{n_k} \tr\left( \uqR_{l,k} \qPsi \left( \omega\qI_N + \omega \sum_{i,j} \te_{i,j} \uqR_{i,j} + \omega \bqH \tqPhi \tqPhi^{-1} \tqPhi \bqH^H\right) \qPsi\right)  \nonumber \\
     \mathop =\limits^{(ii)} & \frac{\omega}{n_k} \sum_{i,j}  \te_{i,j}  \tr\left(\uqR_{l,k} \qPsi \uqR_{i,j} \qPsi \right)  +  \frac{\omega}{n_k}\tr\left(\uqR_{l,k} \qPsi \qPsi \right)   \nonumber \\
        &+  \frac{\omega^2}{n_k} \sum_{i,j}  \beta_{i,j} e_{i,j}  \tr\left(\uqR_{l,k} \qPsi \bqH \tqPhi \uqT_{i,j} \tqPhi \bqH^H \qPsi \right)  +  \frac{\omega^2}{n_k} \tr\left(\uqR_{l,k} \qPsi \bqH \tqPhi \tqPhi \bqH^H \qPsi \right), \label{eq:Rebetae}
\end{align}
where $(i)$ and $(ii)$ are obtained by expanding $\qPsi^{-1}$ and $\tqPhi^{-1}$, respectively. Similarly, we can get
\begin{multline}\label{eq:Rete}
     \te_{l,k}  = \frac{\omega}{n_k} \sum_{i,j}  \beta_{i,j} e_{i,j}  \tr\left(\uqT_{l,k} \tqPsi \uqT_{i,j} \tqPsi \right)  +  \frac{\omega}{n_k}\tr\left(\uqT_{l,k} \tqPsi \tqPsi \right)   \\
       +  \frac{\omega^2}{n_k} \sum_{i,j}  \te_{i,j}  \tr\left(\uqT_{l,k} \tqPsi \bqH^H \qPhi \uqR_{i,j} \qPhi \bqH \tqPsi \right)  +  \frac{\omega^2}{n_k} \tr\left(\uqT_{l,k} \tqPsi \bqH^H \qPhi \qPhi \bqH \tqPsi \right).
\end{multline}
The equations (\ref{eq:Rebetae}) and (\ref{eq:Rete}) can be rewritten as
\begin{align}
     \frac{n_k}{n} \beta_{l,k} e_{l,k} = &  \sum_{i,j}  \frac{n_j}{n} \te_{i,j}  \frac{\omega}{n_j} \tr\left(\uqR_{i,j} \qPsi \uqR_{l,k} \qPsi \right)  +  \frac{\omega}{n} \tr\left(\uqR_{l,k} \qPsi \qPsi \right)   \nonumber \\
        &+   \sum_{i,j} \frac{n_j}{n} \beta_{i,j} e_{i,j}  \frac{\omega^2}{n_j} \tr\left(\uqT_{i,j} \tqPhi \bqH^H \qPsi \uqR_{l,k} \qPsi \bqH \tqPhi \right)  +  \frac{\omega^2}{n} \tr\left(\uqR_{l,k} \qPsi \bqH \tqPhi \tqPhi \bqH^H \qPsi \right) \nonumber \\
       =& \sum_{i,j}  \frac{n_j}{n} \te_{i,j}  u^{(1)}_{ij,lk}  +  \sum_{i,j} \frac{n_j}{n} \beta_{i,j} e_{i,j}  u^{(2)}_{ij,lk} + \frac{\omega}{n} \tr\left(\uqR_{l,k} \qPsi \qPsi \right) + \frac{\omega^2}{n} \tr\left(\uqR_{l,k} \qPsi \bqH \tqPhi \tqPhi \bqH^H \qPsi \right) \label{eq:Rebetae2}
\end{align}
and
\begin{align}
     \frac{n_k}{n} \te_{l,k}  = & \sum_{i,j} \frac{n_j}{n} \beta_{i,j} e_{i,j}  \frac{\omega}{n_j} \tr\left(\uqT_{i,j} \tqPsi \uqT_{l,k} \tqPsi \right)  + \frac{\omega}{n} \tr\left(\uqT_{l,k} \tqPsi \tqPsi \right)   \nonumber \\
       &+ \sum_{i,j} \frac{n_j}{n} \te_{i,j}   \frac{\omega^2}{n_j} \tr\left(\uqR_{i,j} \qPhi \bqH \tqPsi \uqT_{l,k} \tqPsi \bqH^H \qPhi \right)  +  \frac{\omega^2}{n} \tr\left(\uqT_{l,k} \tqPsi \bqH^H \qPhi \qPhi \bqH \tqPsi \right) \nonumber \\
      =& \sum_{i,j} \frac{n_j}{n} \beta_{i,j} e_{i,j}  v^{(1)}_{ij,lk} + \sum_{i,j} \frac{n_j}{n} \te_{i,j}  v^{(2)}_{ij,lk} + \frac{\omega}{n} \tr\left(\uqT_{l,k} \tqPsi \tqPsi \right) +  \frac{\omega^2}{n} \tr\left(\uqT_{l,k} \tqPsi \bqH^H \qPhi \qPhi \bqH \tqPsi \right).    \label{eq:Rete2}
\end{align}
Now, let $\qxi \triangleq \left[ {\tt vec}(\qA_3)^T, {\tt vec}(\qA_4)^T\right]^T$, $\qb \triangleq \left[{\tt vec}(\qC_3)^T, {\tt vec}(\qC_4)^T \right]^T$,
$\qGamma''' \triangleq \left[
                      \begin{array}{cc}
                          \qGamma'''_{11} & \qGamma'''_{12} \\
                          \qGamma'''_{21} & \qGamma'''_{22}
                      \end{array} \right]$,
where $\qA_3, \qA_4$, $\qC_3, \qC_4 \in \bbC^{L \times K},
\qGamma'''_{11},\qGamma'''_{12}, \qGamma'''_{21}, \qGamma'''_{22}
\in \bbC^{LK \times LK}$ with
\begin{subequations}
\begin{align}
    [\qA_3]_{l,k}&= \frac{n_k}{n} \beta_{l,k} e_{l,k}, ~~[\qA_4]_{l,k} =  \frac{n_k}{n} \te_{l,k},\\
    [\qC_3]_{l,k}&= \frac{\frac{\omega}{n} \tr\left(\uqR_{l,k} \qPsi \qPsi \right) + \frac{\omega^2}{n} \tr\left(\uqR_{l,k} \qPsi \bqH \tqPhi \tqPhi \bqH^H \qPsi \right)}{1-u^{(2)}_{ij,ij}},\\
    [\qC_4]_{l,k}&= \frac{\frac{\omega}{n} \tr\left(\uqT_{l,k} \tqPsi \tqPsi \right) +  \frac{\omega^2}{n} \tr\left(\uqT_{l,k} \tqPsi \bqH^H \qPhi \qPhi \bqH \tqPsi \right)}{1-v^{(2)}_{ij,ij}},\\
    [\qGamma'''_{11}]_{lk,ij} &= \left\{
    \begin{aligned} 0,~~~~~  & (i,j)\neq (l,k); \\
     \frac{u^{(2)}_{ij,lk}}{1-u^{(2)}_{ij,ij}},  & (i,j)= (l,k) ,
    \end{aligned}
    \right.
    ~~~[\qGamma'''_{12}]_{lk,ij} =  \frac{u^{(1)}_{ij,lk}}{1-u^{(2)}_{ij,ij}},\\
     [\qGamma'''_{21}]_{lk,ij}& =  \frac{v^{(1)}_{ij,lk}}{1-v^{(2)}_{ij,ij}},
    ~~~[\qGamma'''_{22}]_{lk,ij} = \left\{
    \begin{aligned} 0,~~~~~  & (i,j)\neq (l,k); \\
     \frac{v^{(2)}_{ij,lk}}{1-v^{(2)}_{ij,ij}},  &(i,j) = (l,k).
    \end{aligned}
    \right.
\end{align}
\end{subequations}
Thus, from (\ref{eq:Rebetae2}) and (\ref{eq:Rete2}), we have
\begin{equation}\label{eq:qxi}
    \qxi = \qGamma''' \qxi + \qb.
\end{equation}
Using the matrix inversion lemma (Lemma \ref{Lemma:MatrixInversion}), we obtain $\tqPhi \bqH^H \qPsi = \tqPsi \bqH^H \qPhi$. This implies that
$u^{(2)}_{ij,ij}=v^{(2)}_{ij,ij}$, for $\forall i,j$. We immediately get
\begin{equation}\label{eq:qGammaqGamma}
    \qGamma''' = \left(\qGamma''\right)^T.
\end{equation}
Now, define
\begin{equation}\label{eq:qV}
    \qLambda = \diag\left(1-u^{(2)}_{11,11},\ldots,1-u^{(2)}_{LK,LK},1-v^{(2)}_{11,11},\ldots,1-v^{(2)}_{LK,LK}\right).
\end{equation}
Multiplying both sides of \eqref{eq:qxi} by $\qLambda$ gives
\begin{equation}\label{eq:qVqxi}
    \qLambda\qxi = \qLambda\qGamma''' \qxi + \qLambda\qb.
\end{equation}
For $\omega \in \bbR^+$, the entries of $\qxi$, $\qLambda\qGamma'''$ and $\qLambda\qb$ are positive. Thus, the entries of $\qLambda\qxi$ are positive. Since the
entries of $\qxi$ are positive, we conclude that $1-u^{(2)}_{ij,ij} >0$ and $1-v^{(2)}_{ij,ij} >0$, for $\forall l,k$. From \eqref{eq:qxi}, we obtain that the
entries of $\qGamma'''$ and $\qb$ are positive, for $\omega \in \bbR^+$. Lemma \ref{Lemma:SpectralRadius} implies $\rho\left(\qGamma'''\right)\leq 1-\frac{\min
\qb_l}{\max \qxi_l}$.

Using Lemma \ref{Lemma:TraceNorm}, \eqref{eq:asumRTH}, and the fact that $\|\qPsi\|, \|\tqPsi\|\leq\frac{1}{\omega}$, we have
\begin{equation}
     \frac{n_k}{n} \beta_{l,k} e_{l,k} \leq \frac{N_l C_{\rm max}}{n \omega} \leq \frac{\beta_0 C_{\rm max}}{\omega}
\end{equation}
and
\begin{equation}
     \frac{n_k}{n} \te_{l,k}  \leq \frac{n_k C_{\rm max}}{n \omega} \leq \frac{C_{\rm max}}{\omega},
\end{equation}
where $\beta_0 \triangleq \max_{k,l}\, \{ \beta_{l,k}(N) \}$. From \eqref{eq:asymptoticregime}, we have
\begin{equation}\label{eq:supmaxqxi}
      \sup_N \max \qxi_l \leq \sup_N \max \left\{1, \beta_{0}\right\} \frac{C_{\rm max}}{\omega} <+\infty.
\end{equation}

For $\qb_l$, we have
\begin{align}
     b_{l,k} \geq & \frac{\omega}{n} \tr\left(\uqR_{l,k} \qPsi \qPsi \right) \mathop \geq \limits^{(i)} \frac{\omega}{n} \frac{\left(\tr\left(\uqR_{l,k} \qPsi  \right)\right)^2}{\tr\left( \uqR_{l,k} \right)}
             \mathop \geq \limits^{(ii)}  \frac{\omega}{n} \frac{\left(\tr\left(\uqR_{l,k} \right)\right)^3}{\left(\tr\left(\uqR_{l,k} \qPsi^{-1}  \right)\right)^2} \mathop \geq \limits^{(iii)}  \frac{\omega}{n} \frac{ \tr\left(\uqR_{l,k} \right) }{\|\qPsi^{-1}\|^2} \nonumber \\
             \geq & \frac{\omega}{n} \frac{ \tr\left(\qR_{l,k} \right) }{\left(\omega + \max\{1,\beta_0\}LKC_{\rm max}^2 + LKC_{\rm max}\right)^2}.
\end{align}
where $(i)$ and $(ii)$ follow from $1)-a)$ of Lemma \ref{Lemma:TraceNorm}, i.e., $\left(\tr(\qA\qB)\right)^2 \leq \tr(\qA\qA^H)\tr(\qB\qB^H)$, $(iii)$ is due to
$2)$ of Lemma \ref{Lemma:TraceNorm}. Similarly,
\begin{equation}
     \tb_{l,k} \geq \frac{\omega}{n} \frac{ \tr\left(\qT_{l,k} \right) }{\left(\omega + \max\{1,\beta_0\}LKC_{\rm max}^2 + LKC_{\rm max}\right)^2}.
\end{equation}
As a consequence, we have
\begin{equation}\label{eq:infminqb}
     \inf_N \min \qb_l  \geq \frac{ \omega C_5 }{\left(\omega + \sup_N \max\{1,\beta_0\}LKC_{\rm max}^2 + LKC_{\rm max}\right)^2},
\end{equation}
where $C_5 = \inf_N \max \{ \frac{1}{n}\tr\left(\qT_{l,k}\right), \frac{1}{n}\tr\left(\qT_{l,k} \right)\}>0$.

Combining \eqref{eq:supmaxqxi} and \eqref{eq:infminqb}, we obtain
\begin{equation}\label{eq:qGammaleq1}
    \sup_N \left[ \rho \left(\qGamma'''\right) \right] \leq 1 - \frac{\lambda_0\omega^2}{\left(\omega+\lambda'_0\right)^2}  < 1.
\end{equation}
According to \eqref{eq:qGammaqGamma}, \eqref{eq:suprho_qGamma} holds true. It is easy to get \eqref{eq:supinvqGamma} by $\rho \left(\qGamma''\right)< 1$. A similar
proof can be found in \cite{Dupuy-11IT,Hachem08AAP}, and is therefore omitted.
\end{proof}

Applying this lemma and \eqref{eq:qGammatoqGamma}, there exists $N_0$ such that $(\qI-\qGamma')$ is invertible, for each $N>N_0$, and
$\sup_{N>N_0}\left[\mnorm{(\qI-\qGamma')^{-1}}_{\infty}\right] \leq \frac{\left(\omega+\lambda'_0\right)^2}{\lambda_0\omega^2}$. Note that $\qepsilon' =
O\left(\frac{1}{N^2}\right){\bf 1}$. Hence, from \eqref{eq:qeta2}, we obtain $(\alpha_{l,k}-\tau_{l,k})$'s and $(\talpha_{l,k}-\ttau_{l,k})$'s are of
$O\left(\frac{1}{N^2}\right)$. This establishes \eqref{eq:Proposition3a} and \eqref{eq:Proposition3b}.

From \eqref{eq:Proposition3a}, \eqref{eq:Proposition3b}, \eqref{eq:varepsilon}, and \eqref{eq:tvarepsilon}, we have $\varepsilon_{l,k} =
O\left(\frac{1}{N^2}\right)$ and $\tvarepsilon_{l,k} = O\left(\frac{1}{N^2}\right)$. \eqref{eq:alphae} and \eqref{eq:talphae} can be rewritten as a matrix form
similar to \eqref{eq:qeta2}. Using the same approach as in the proof of \eqref{eq:Proposition3a} and \eqref{eq:Proposition3b}, we prove that
$(\alpha_{l,k}-\beta_{l,k}e_{l,k})$'s and $(\talpha_{l,k}-\te_{l,k})$'s are of $O\left(\frac{1}{N^2}\right)$. This shows that \eqref{eq:Proposition3c} and
\eqref{eq:Proposition3d} are established and the proof is completed.

\section{Proof of $\Ex \{ m_{\qB_N} \} - \Ex \{ m_{\calqB_N}\} = O\left( \frac{1}{\sqrt{N}}\right)$ in Theorem \ref{mainTh_Stj}}\label{Appendix: Proof of GtoNG}
The aim of this appendix is to prove
\begin{equation}\label{eq:Step1aim}
    \left|\Ex \{ m_{\qB_N}(\omega) \} - \Ex \{ m_{\calqB_N}(\omega) \}\right| = O\left(\frac{1}{\sqrt{N}}\right).
\end{equation}
We mainly make use of the generalized Lindeberg principle given below.

\begin{Lemma}\label{Lma:Lindeberg}
{\rm (Generalized Lindeberg Principle \cite{Korada-11IT})} Let $\qv= [v_i] \in \bbR^n$ and $\tqv = [\tv_i] \in \bbR^n$ be two random vectors with mutually
independent components. Define $\left\{a_i\right\}_{1\leq i \leq n}$ and $\left\{b_i\right\}_{1\leq i \leq n}$ with
\begin{equation}\label{eq:Def_ab}
    a_i \triangleq |\Ex \{ v_i \} - \Ex \{ \tv_i \}|,~~\mbox{and}~~b_i \triangleq |\Ex \{ v_i^2 \} - \Ex \{ \tv_i^2 \}|.
\end{equation}
Then, given a twice continuously differentiable function $f: \bbR^n\rightarrow \bbR$, we have
\begin{multline}\label{eq:Lindeberg}
    \left|\Ex\left\{ f(\qv) \right\} - \Ex\left\{ f(\tqv) \right\}\right| \leq \sum^n_{i=1} \biggl[ a_i\Ex \left\{|\partial_i f\left(\qv^{i-1}_1, 0, \tqv^{n}_{i+1}\right)|\right\} + \frac{1}{2}b_i\Ex \left\{|\partial_i^2 f\left(\qv^{i-1}_1, 0, \tqv^{n}_{i+1}\right)|\right\} \\
             + \frac{1}{2} \Ex \left\{\int^{v_i}_0|\partial_i^3 f\left(\qv^{i-1}_1, s, \tqv^{n}_{i+1}\right)| \left(v_i-s\right)^2 ds \right\} \\
             + \frac{1}{2} \Ex \left\{\int^{\tv_i}_0|\partial_i^3 f\left(\qv^{i-1}_1, s, \tqv^{n}_{i+1}\right)| \left(\tv_i-s\right)^2 ds \right\} \biggr],
\end{multline}
where $\partial_i^p$ is the $p$-fold derivative in the $i$-th coordinate, $\qv^{i-1}_1 = \left(v_1,\dots,v_{i-1}\right)$, and $\tqv^{n}_{i+1} =
\left(\tv_{i+1},\ldots,\tv_{n}\right)$.
\end{Lemma}

As $\uqX_{l,k}$'s and $\ucalqX_{l,k}$'s are matrices with entries satisfying \eqref{eq:assX}, we have $a^{(l,k)}_{i,j} =b^{(l,k)}_{i,j} = 0$ for $i = 1,\dots,N$
and $j = 1,\dots,n$. Therefore, the remaining challenge is to evaluate the third and fourth terms of the right-hand side of inequality \eqref{eq:Lindeberg}. Since
the real and imaginary parts of $X_{ij}^{(l,k)}$ are independent, all the results established in the real case can be directly applied for the complex case. Thus,
without loss of generality, we only take the derivative with respect to the real part of $X_{ij}^{(l,k)}$ in \eqref{eq:Lindeberg}. Before proceeding, we remark
that because of the finite $6$-th order moment assumption of $X_{ij}^{(l,k)}$'s, the following proof is much simpler than that in \cite{Wen-11IT}.

Let
\begin{equation}\label{eq:function_f}
    f\left( \left\{ \qA_{l,k}\right\}_{\forall l,k} \right) = \frac{1}{N}\tr \left( \qG + \omega\qI_N \right)^{-1}
\end{equation}
where
\begin{equation}
\qG = \left( \sum_{l,k} \left( \uqR_{l,k}^{\frac{1}{2}} \qA_{l,k}\uqT_{l,k}^{\frac{1}{2}}+ \ubqH_{l,k}\right)\right) \left(\sum_{l,k} \left( \uqR_{l,k}^{\frac{1}{2}} \qA_{l,k}
\uqT_{l,k}^{\frac{1}{2}}+ \ubqH_{l,k} \right)\right)^H,
\end{equation}
for any $\qA_{l,k} \in \bbR^{N \times n},$ for $l = 1,\dots,L$ and $k = 1,\dots,K$. As such, we have $m_{\qB_N}(\omega) = f\left(\left\{ \qX_{l,k}\right\}_{\forall
l,k} \right)$ and $m_{\calqB_N}(\omega) = f\left( \left\{ \calqX_{l,k}\right\}_{\forall l,k} \right)$. To use \eqref{eq:Lindeberg},
$\left\{\qA_{l,k}\right\}_{\forall l,k}$ will take the form $\left\{\qA_{l,k} = [A^{(l,k)}_{ij}(l_0,k_0,r,c,s)]\right\}_{\forall l,k}$ with
\begin{align}
 A^{(l,k)}_{i,j}(l_0,k_0,r,c,s) = \left\{
    \begin{array}{cl} \frac{X_{ij}^{(l,k)}}{\sqrt{n_k}}, & \mbox{if}~l<l_0,~\mbox{or}~l=l_0, k<k_0,~\mbox{or}~l=l_0, k=k_0,i<r,\\
                                       &~\mbox{or}~l=l_0, k=k_0,i=r,j<c; \\
     s, & \mbox{if}~ (l,k)=(l_0,k_0)~ \mbox{and}~ (i,j)=(r,c); \\
     \frac{\calX_{ij}^{(l,k)}}{\sqrt{n_k}}, & \mbox{otherwise}.
    \end{array}
    \right.  \label{eq:Adef}
\end{align}
Taking the third-fold partial derivative of \eqref{eq:function_f} with respect to $A^{(l,k)}_{i,j}$, denoted by $\partial^{(l,k)3}_{ij}$, we have
\begin{multline}
     \partial^{(l,k)3}_{ij}f =-\frac{6}{N}\tr \left( (\partial^{(l,k)}_{ij}\qG) \left(\qG+\omega\qI_N\right)^{-1} (\partial^{(l,k)}_{ij}\qG) \left(\qG+\omega\qI_N\right)^{-1} (\partial^{(l,k)}_{ij}\qG) \left(\qG+\omega\qI_N\right)^{-2}  \right)\\
              + \frac{3}{N}\tr \left( (\partial^{(l,k)2}_{ij}\qG) \left(\qG+\omega\qI_N\right)^{-1} (\partial^{(l,k)}_{ij}\qG) \left(\qG+\omega\qI_N\right)^{-2}  \right) \\
             + \frac{3}{N}\tr \left( (\partial^{(l,k)}_{ij}\qG) \left(\qG+\omega\qI_N\right)^{-1} (\partial^{(l,k)2}_{ij}\qG) \left(\qG+\omega\qI_N\right)^{-2}  \right),
\end{multline}
where
\begin{subequations}
\begin{align}
     \partial^{(l,k)}_{ij}\qG  =& \left( \uqR_{l,k}^{\frac{1}{2}} \qE_{ij} \uqT_{l,k}^{\frac{1}{2}} \right) \sum_{l_1,k_1} \left( \uqR_{l_1,k_1}^{\frac{1}{2}} \qA_{l_1,k_1} \uqT_{l_1,k_1}^{\frac{1}{2}}+ \ubqH_{l_1,k_1} \right)^H  \nonumber \\
                                &+  \left(\sum_{l_1,k_1} \left( \uqR_{l_1,k_1}^{\frac{1}{2}} \qA_{l_1,k_1} \uqT_{l_1,k_1}^{\frac{1}{2}}+ \ubqH_{l_1,k_1}\right) \right) \left( \uqT_{l,k}^{\frac{1}{2}} \qE_{ji} \uqR_{l,k}^{\frac{1}{2}}\right), \label{eq:Gpar1}  \\
     \partial^{(l,k)2}_{ij}\qG  = & 2 \uT^{(l,k)}_{jj} \uqR_{l,k}^{\frac{1}{2}} \qE_{ii} \uqR_{l,k}^{\frac{1}{2}}. \label{eq:Gpar2}
\end{align}
\end{subequations}
Here, $\qE_{ij}$ denotes the matrix which has its entries being all $0$'s except for the $(i, j)$-th entry as $1$.

Using Lemma \ref{Lemma:TraceNorm}, the first term of $\partial^{(l,k)3}_{ij}f $ can be bounded by
\begin{equation}
    \left|\tr \left( (\partial^{(l,k)}_{ij}\qG) \left(\qG+\omega\qI_N\right)^{-1} (\partial^{(l,k)}_{ij}\qG) \left(\qG+\omega\qI_N\right)^{-1} (\partial^{(l,k)}_{ij}\qG) \left(\qG+\omega\qI_N\right)^{-2}  \right)\right| \leq \frac{1}{\omega^4} \|(\partial^{(l,k)}_{ij}\qG)\|^{3}_{\rm F},
\end{equation}
and the second and third terms of $\partial^{(l,k)3}_{ij}f$ can be bounded by
\begin{equation}
    \left| \tr \left( (\partial^{(l,k)2}_{ij}\qG) \left(\qG+\omega\qI_N\right)^{-1} (\partial^{(l,k)}_{ij}\qG) \left(\qG+\omega\qI_N\right)^{-2}  \right)\right| \leq \frac{1}{\omega^3}\| (\partial^{(l,k)2}_{ij}\qG) \|_{\rm F} \|(\partial^{(l,k)}_{ij}\qG)\|_{\rm F}.
\end{equation}
From \eqref{eq:Gpar1} and \eqref{eq:Gpar2}, using Lemma \ref{Lemma:TraceNorm} and \eqref{eq:asumRTH}, we obtain
\begin{align}
     \|(\partial^{(l,k)}_{ij}\qG)\|_{\rm F} \mathop \leq \limits^{(i)} & 2 \sum_{l_1,k_1} \left(\|\uqR_{l,k}^{\frac{1}{2}} \qE_{ij} \uqT_{l,k}^{\frac{1}{2}} \uqT_{l_1,k_1}^{\frac{1}{2}} \qA^H_{l_1,k_1} \uqR_{l_1,k_1}^{\frac{1}{2}}\|_{\rm F}+ \|\uqR_{l,k}^{\frac{1}{2}} \qE_{ij} \uqT_{l,k}^{\frac{1}{2}} \ubqH^H_{l_1,k_1}\|_{\rm F} \right)  \nonumber \\
       \mathop \leq \limits^{(ii)} & 2 \sum_{l_1,k_1} \left( C_{\rm max} \| \qE_{ij} \uqT_{l,k}^{\frac{1}{2}} \uqT_{l_1,k_1}^{\frac{1}{2}} \qA^H_{l_1,k_1} \|_{\rm F}+ \|\uqR_{l,k}^{\frac{1}{2}} \qE_{ij} \uqT_{l,k}^{\frac{1}{2}} \ubqH^H_{l_1,k_1}\|_{\rm F} \right)    \nonumber \\
                           =  & 2 C_{\rm max} \sum_{l_1,k_1}  \left[ \tr\left(  \qE_{ij} \uqT_{l,k}^{\frac{1}{2}} \uqT_{l_1,k_1}^{\frac{1}{2}} \qA^H_{l_1,k_1} \qA_{l_1,k_1} \uqT_{l_1,k_1}^{\frac{1}{2}} \uqT_{l,k}^{\frac{1}{2}} \qE_{ji} \right) \right]^\frac{1}{2}        \nonumber \\
                              &+ 2 \sum_{l_1,k_1} \left[ \tr\left( \uqR_{l,k}^{\frac{1}{2}} \qE_{ij} \uqT_{l,k}^{\frac{1}{2}} \ubqH^H_{l_1,k_1} \ubqH_{l_1,k_1} \uqT_{l,k}^{\frac{1}{2}}  \qE_{ji} \uqR_{l,k}^{\frac{1}{2}}\right) \right]^\frac{1}{2}      \nonumber \\
       \mathop \leq \limits^{(iii)} & 2 C^2_{\rm max} \sum_{l,k} \left( \sum_{i=1}^{N} \left(A^{(l,k)}_{ij}\right)^2\right)^\frac{1}{2} + 2LKC^{\frac{3}{2}}_{\rm max}
\end{align}
and
\begin{align}
     \| (\partial^{(l,k)2}_{ij}\qG) \|_{\rm F} =  2 \uT^{(l,k)}_{jj} \uR^{(l,k)}_{ii}  \leq  2\|\qT_{l,k}\|\|\qR_{l,k}\| \leq 2C_{\rm max},
\end{align}
where $(i)$ is obtained by the triangle inequality of the Frobenius norm, $(ii)$ follows from $1(b)$ of Lemma \ref{Lemma:TraceNorm} and \eqref{eq:asumRTH}, and
$(iii)$  follows from $2$ of Lemma \ref{Lemma:TraceNorm} and \eqref{eq:asumRTH}. Combining everything together, we get
\begin{align}
     \Ex \left\{\partial^{(l_0,k_0)3}_{rc}f \right\} \leq &\frac{C_1}{N}\Ex \left\{ \left( \sum_{l,k} \left( \sum_{i=1}^{N} \left(A^{(l,k)}_{ic}\right)^2\right)^\frac{1}{2} + C_2 \right)^3\right\}  \nonumber \\
             \mathop \leq \limits^{(i)} &  \frac{C_1(LK+1)^2}{N}\Ex \left\{ \sum_{l,k} \left( \sum_{i=1}^{N} \left(A^{(l,k)}_{ic}\right)^2\right)^\frac{3}{2} + C_2^3 \right\}  \nonumber \\
             \mathop \leq \limits^{(ii)} &  \frac{C_3}{N} \left(|s|^3 + \Ex \left\{ \left( \sum_{i\neq r}^{N} \left(A^{(l_0,k_0)}_{ic}\right)^2\right)^\frac{3}{2} \right\} + \Ex \left\{ \sum_{l\neq l_0}^{L}\sum_{k\neq k_0}^{K} \left( \sum_{i=1}^{N} \left(A^{(l,k)}_{ic}\right)^2\right)^\frac{3}{2} \right\} + C_2^3 \right)   \nonumber \\
             \mathop = \limits^{(iii)} &  \frac{C_3}{N} \left(|s|^3 + C_4 \right), \label{eq:Eabsparf}
\end{align}
where $C_1,C_2,C_3,C_4$ denote constants, $(i)$ is obtained by Lemma \ref{Lemma:AbsIneq}, $(ii)$ follows from the definition of $A^{(l,k)}_{ij}$ \eqref{eq:Adef},
and $(iii)$ is due to the fact that $X_{ij}^{(l,k)}$ and $\calX_{ij}^{(l,k)}$ have finite $6$-th order moment, thus giving the second and third terms of third line
of \eqref{eq:Eabsparf} as $O(1)$.

Finally, using \eqref{eq:Lindeberg} and \eqref{eq:Eabsparf}, we obtain
\begin{align}
    \left|\Ex\left\{ \Re \left\{m_{\qB_N}(\omega)\right\} \right\} - \Ex\left\{ \Re \left\{m_{\calqB_N}(\omega)\right\} \right\}\right|
              \leq & \frac{C_3}{2N} \sum_{l,k} \sum^N_{r=1}\sum^n_{c=1} \left(  \Ex \left\{\int^{|X_{rc}^{(l,k)}|/\sqrt{n_{k}}}_0 \left(|s|^3 + C_4 \right) \left(\frac{X_{rc}^{(l,k)}}{\sqrt{n_{k}}}-s\right)^2 ds \right\} \right. \nonumber \\
                   & \left. +  \Ex \left\{\int^{|\calX_{rc}^{(l,k)}|/\sqrt{n_{k}}}_0 \left(|s|^3 + C_4 \right) \left(\frac{\calX_{rc}^{(l,k)}}{\sqrt{n_{k}}}-s\right)^2 ds \right\} \right)   \nonumber \\
              \leq & \frac{C_3}{2N} \sum_{l,k} \sum^N_{r=1}\sum^n_{c=1} \left( \frac{1}{6} \Ex \left\{ \left(\frac{|X_{rc}^{(l,k)}|}{\sqrt{n_{k}}}\right)^6\right\} + \frac{C_4}{3} \Ex \left\{ \left(\frac{|X_{rc}^{(l,k)}|}{\sqrt{n_{k}}}\right)^3\right\} \right.  \nonumber \\
                   & \left. + \frac{1}{6} \Ex \left\{ \left(\frac{|\calX_{rc}^{(l,k)}|}{\sqrt{n_{k}}}\right)^6\right\} + \frac{C_4}{3} \Ex \left\{ \left(\frac{|\calX_{rc}^{(l,k)}|}{\sqrt{n_{k}}}\right)^3\right\}  \right)   \nonumber \\
               =&O\left(\frac{1}{\sqrt{N}}\right).
\end{align}
The quantity $\left|\Ex\left\{ \Im \left\{m_{\qB_N}(\omega)\right\} \right\} - \Ex\left\{ \Im \left\{m_{\calqB_N}(\omega)\right\} \right\}\right|$ also admits the
same upper bound. Thus, \eqref{eq:Step1aim} is true.

\section{Existence and Uniqueness}\label{Appendix: Existence and Uniquenss}
%In this appendix, we will consider existence and uniqueness of solution to \eqref{eq:Solutionete}.

\subsection{Existence}
Following \cite{Couillet-11IT} and using Proposition \ref{Prop:SXi} the existence of $(e_{l,k},\te_{l,k})_{\forall l,k}$ can be shown.

\subsection{Uniqueness}
Let $(e_{l,k}, \te_{l,k})$ and $(e_{l,k}^{\circ}, \te_{l,k}^{\circ})$ be two solutions satisfying \eqref{eq:Solutionete}, and $\qPsi^{\circ}, \tqPsi^{\circ},
\qPhi^{\circ}, \tqPhi^{\circ}$ be the matrices obtained by replacing $e_{l,k}(\omega)$'s and $\te_{l,k}(\omega)$'s in $\qPsi, \tqPsi, \qPhi, \tqPhi$ with
$e_{l,k}^{\circ}(\omega)$'s and $\te_{l,k}^{\circ}(\omega)$'s respectively. To prove the uniqueness, we need to show that $e_{l,k}- e_{l,k}^{\circ} = 0$ and
$\te_{l,k} - \te_{l,k}^{\circ} = 0$, for any $l$ and $k$. Our proof is inspired by \cite{Dupuy-11IT}.

A standard calculation involving Lemma 11 yields
\begin{subequations}
\begin{align}
\beta_{l,k}(e_{l,k} - e_{l,k}^{\circ}) =  & -\frac{\omega}{n_k} \sum_{i,j} (\te_{i,j} - \te_{i,j}^{\circ}) \tr \left(\uqR_{l,k} \qPsi \uqR_{i,j} \qPsi^{\circ} \right) \nonumber \\
&+ \frac{\omega^2}{n_k} \sum_{i,j} \beta_{i,j}(e_{i,j} - e_{i,j}^{\circ}) \tr \left(\uqR_{l,k} \qPsi \bqH \tqPhi \uqT_{i,j} \tqPhi^{\circ} \bqH^H \qPsi^{\circ} \right),\label{eq:ee}\\
\te_{l,k} - \te_{l,k}^{\circ} =  & -\frac{\omega}{n_k} \sum_{i,j} \beta_{i,j}(e_{i,j} - e_{i,j}^{\circ}) \tr \left(\uqT_{l,k} \tqPsi \uqT_{i,j} \tqPsi^{\circ} \right) \nonumber \\
&+ \frac{\omega^2}{n_k} \sum_{i,j} (\te_{i,j} - \te_{i,j}^{\circ}) \tr \left(\uqT_{l,k} \tqPsi \bqH^H \qPhi \uqR_{i,j} \qPhi^{\circ} \bqH \tqPsi^{\circ} \right). \label{eq:tete}
\end{align}
\end{subequations}
Now, let
%\begin{equation}
$\qzeta \triangleq \left[ {\tt vec}(\qA_5)^T, {\tt vec}(\qA_6)^T\right]^T,~~ \qPi \triangleq \left[\begin{array}{cc}
\qPi_{11} & \qPi_{12} \\
\qPi_{21} & \qPi_{22}
\end{array} \right]$,
%\end{equation}
where
%\begin{equation}
$\qA_5, \qA_6 \in \bbC^{L \times K},\qPi_{11}, \qPi_{12}, \qPi_{21}, \qPi_{22} \in \bbC^{LK \times LK}$
%\end{equation}
with
\begin{subequations}
\begin{align}
    [\qA_5]_{l,k}&= \beta_{l,k}(e_{l,k} - e_{l,k}^{\circ}), ~~~~~~~~
    [\qA_6]_{l,k} = \te_{l,k} - \te_{l,k}^{\circ},\\
    [\qPi_{11}]_{lk,ij} &= \left\{
    \begin{aligned} 0,~~~~~~~~~~~~~~~~~~~~~~~  & ~\mbox{for }(i,j)\neq (l,k); \\
     \frac{\frac{\omega^2 }{n_k} \tr \left(\uqR_{l,k} \qPsi \bqH \tqPhi \uqT_{i,j} \tqPhi^{\circ} \bqH^H \qPsi^{\circ} \right)}{1 - \frac{\omega^2 }{n_k} \tr \left(\uqR_{l,k} \qPsi \bqH \tqPhi \uqT_{l,k} \tqPhi^{\circ} \bqH^H \qPsi^{\circ} \right)},  & ~\mbox{for }(i,j)= (l,k) ,
    \end{aligned}
    \right.\\
    [\qPi_{12}]_{lk,ij} &=  \frac{-\frac{\omega}{n_k} \tr \left(\uqR_{l,k} \qPsi \uqR_{i,j} \qPsi^{\circ} \right)}{1-\frac{\omega^2 }{n_k} \tr \left(\uqR_{l,k} \qPsi \bqH \tqPhi \uqT_{l,k} \tqPhi^{\circ} \bqH^H \qPsi^{\circ} \right)},\\
    [\qPi_{21}]_{lk,ij} &=  \frac{ -\frac{\omega}{n_k} \tr \left(\uqT_{l,k} \tqPsi \uqT_{i,j} \tqPsi^{\circ} \right)}{1 - \frac{\omega^2}{n_k} \tr \left(\uqT_{l,k} \tqPsi \bqH^H \qPhi \uqR_{l,k} \qPhi^{\circ} \bqH \tqPsi^{\circ} \right)},\\
    [\qPi_{22}]_{lk,ij} &= \left\{
    \begin{aligned} 0,~~~~~~~~~~~~~~~~~~~~~  & ~\mbox{for }(i,j)\neq (l,k); \\
     \frac{\frac{\omega^2}{n_k} \tr \left(\uqT_{l,k} \tqPsi \bqH^H \qPhi \uqR_{i,j} \qPhi^{\circ} \bqH \tqPsi^{\circ} \right)}{1 - \frac{\omega^2}{n_k} \tr \left(\uqT_{l,k} \tqPsi \bqH^H \qPhi \uqR_{l,k} \qPhi^{\circ} \bqH \tqPsi^{\circ} \right)},  & ~\mbox{for }(i,j) = (l,k).
    \end{aligned}
    \right.
\end{align}
\end{subequations}
Thus, \eqref{eq:ee} and \eqref{eq:tete} can be written together as
\begin{equation}\label{eq:qzeta}
    \qzeta = \qPi \qzeta.
\end{equation}

To complete the proof, it remains to prove that $\rho(\qPi)<1$. To do so, we first write \eqref{eq:Rebetae} and \eqref{eq:Rete} in matrix form as follows:
\begin{equation}\label{eq:qxip}
    \qxi' = \qK \qxi' + \qb',
\end{equation}
where
%\begin{equation}
$\qxi' = \left[ {\tt vec}(\qA_7)^T, {\tt vec}(\qA_8)^T\right]^T, \qb' = \left[{\tt vec}(\qC_7)^T, {\tt vec}(\qC_8)^T \right]^T, \qK = \left[
                      \begin{array}{cc}
                          \qK_{11} & \qK_{12} \\
                          \qK_{21} & \qK_{22}
                      \end{array} \right]$,
%\end{equation}
and
%\begin{equation}
$\qA_7, \qA_8, \qC_7, \qC_8 \in \bbC^{L \times K}, \qK_{11},\qK_{12}, \qK_{21}, \qK_{22} \in \bbC^{LK \times LK}$
%\end{equation}
with
\begin{subequations}
\begin{align}
    [\qA_7]_{l,k}& =  \beta_{l,k} e_{l,k}, ~~~~~~~~
    [\qA_8]_{l,k} =  \te_{l,k} , \\
    [\qC_7]_{l,k}& = \frac{\frac{\omega}{n_k}\tr\left(\uqR_{l,k} \qPsi \qPsi \right) + \frac{\omega^2}{n_k} \tr\left(\uqR_{l,k} \qPsi \bqH \tqPhi \tqPhi \bqH^H \qPsi \right)}{1-u'^{(2)}_{lk,lk}},\\
    [\qC_8]_{l,k}& = \frac{\frac{\omega}{n_k}\tr\left(\uqT_{l,k} \tqPsi \tqPsi \right) + \frac{\omega^2}{n_k} \tr\left(\uqT_{l,k} \tqPsi \bqH^H \qPhi \qPhi \bqH \tqPsi \right)}{1-v'^{(2)}_{lk,lk}},\\
    [\qK_{11}]_{lk,ij} &= \left\{
    \begin{aligned} 0,~~~~~  & ~\mbox{for }(i,j)\neq (l,k); \\
     \frac{u'^{(2)}_{lk,ij}}{1-u'^{(2)}_{lk,lk}},  & ~\mbox{for }(i,j)= (l,k) ,
    \end{aligned}\right.
    ~~~[\qK_{12}]_{lk,ij} =  \frac{u'^{(1)}_{lk,ij}}{1-u'^{(2)}_{lk,lk}}, \\
     [\qK_{21}]_{lk,ij} &=  \frac{v'^{(1)}_{lk,ij}}{1-v'^{(2)}_{lk,lk}},
    ~~~[\qK_{22}]_{lk,ij} = \left\{
    \begin{aligned} 0,~~~~~  & ~\mbox{for }(i,j)\neq (l,k); \\
     \frac{v'^{(2)}_{lk,ij}}{1-v'^{(2)}_{lk,lk}},  & ~\mbox{for }(i,j) = (l,k),
    \end{aligned}\right.\\
    u'^{(1)}_{lk,ij}&= \frac{\omega}{n_k} \tr\left(\uqR_{l,k} \qPsi \uqR_{i,j} \qPsi \right), ~~~
    u'^{(2)}_{lk,ij}= \frac{\omega^2}{n_k} \tr\left(\uqR_{l,k} \qPsi \bqH \tqPhi \uqT_{i,j} \tqPhi \bqH^H \qPsi \right),\\
    v'^{(1)}_{lk,ij}&= \frac{\omega}{n_k} \tr\left(\uqT_{l,k} \tqPsi \uqT_{i,j} \tqPsi \right), ~~~
    v'^{(2)}_{lk,ij}= \frac{\omega^2}{n_k} \tr\left(\uqT_{l,k} \tqPsi \bqH^H \qPhi \uqR_{i,j} \qPhi \bqH \tqPsi \right).
\end{align}
\end{subequations}

Using a similar approach of \eqref{eq:qxi}, we get that $1-u'^{(2)}_{lk,lk} >0, 1-v'^{(2)}_{lk,lk} >0, \forall l,k$, and the entries of $\qxi', \qK$ and $\qb'$ are
positive, for $\omega \in \bbR^+$. Therefore, from \eqref{eq:qxip} and Lemma \ref{Lemma:SpectralRadius2}, we have $\rho(\qK)<1$. Similarly, we also have
$\rho(\qK^{\circ})<1$, where $\qK^{\circ}$ as well as $\qK^{\circ}_{11}, \qK^{\circ}_{12}, \qK^{\circ}_{21}$, and $\qK^{\circ}_{22}$ are the matrices by replacing
$\qPsi, \tqPsi, \qPhi$, and $\tqPhi$ with $\qPsi^{\circ}, \tqPsi^{\circ}, \qPhi^{\circ}$, and $\tqPhi^{\circ}$, respectively.

For the denominator of $[\qPi_{11}]_{lk,ij}$, applying Lemma \ref{Lemma:traceineq} with $\qA = \omega\sqrt{\frac{1}{n_k}} \uqR_{l,k}^{\frac{1}{2}} \qPsi \bqH
\tqPhi \uqT_{l,k}^{\frac{1}{2}} $  and $\qB = \omega \sqrt{\frac{1}{n_k}} \uqT_{l,k}^{\frac{1}{2}} \tqPhi^{\circ} \bqH^H \qPsi^{\circ}\uqR_{l,k}^{\frac{1}{2}}$
satisfying $\tr(\qA\qA^H)=u'^{(2)}_{lk,lk}<1$ and $\tr(\qB\qB^H)=v'^{(2)}_{lk,lk}<1$, we have
\begin{multline}\label{eq:denominator_qPi11}
1 - \frac{\omega^2 }{n_k} \tr \left(\uqR_{l,k} \qPsi \bqH \tqPhi \uqT_{l,k} \tqPhi^{\circ} \bqH^H \qPsi^{\circ} \right)\\
    \geq  \left(1 - \frac{\omega^2}{n_k} \tr \left(\uqR_{l,k} \qPsi \bqH \tqPhi \uqT_{l,k} \tqPhi \bqH^H \qPsi \right) \right)^{\frac{1}{2}}  \left(1 - \frac{\omega^2}{n_k} \tr \left(\uqR_{l,k} \qPsi^{\circ} \bqH \tqPhi^{\circ} \uqT_{l,k} \tqPhi^{\circ} \bqH^H \qPsi^{\circ} \right) \right)^{\frac{1}{2}} .
\end{multline}
Applying the Cauchy-Schwarz inequality to the numerator of $[\qPi_{11}]_{lk,ij}$ and from \eqref{eq:denominator_qPi11}, we obtain
\begin{align}%\label{eq:qPi11}
   |[\qPi_{11}]_{lk,ij}|  & \leq  \left(\frac{\frac{\omega^2}{n_k} \tr \left(\uqR_{l,k} \qPsi \bqH \tqPhi \uqT_{i,j} \tqPhi \bqH^H \qPsi \right)}{1 - \frac{\omega^2}{n_k} \tr \left(\uqR_{l,k} \qPsi \bqH \tqPhi \uqT_{l,k} \tqPhi \bqH^H \qPsi \right)} \right)^{\frac{1}{2}}  \left(\frac{\frac{\omega^2}{n_k} \tr \left(\uqR_{l,k} \qPsi^{\circ} \bqH \tqPhi^{\circ} \uqT_{i,j} \tqPhi^{\circ} \bqH^H \qPsi^{\circ} \right)}{1 - \frac{\omega^2}{n_k} \tr \left(\uqR_{l,k} \qPsi^{\circ} \bqH \tqPhi^{\circ} \uqT_{l,k} \tqPhi^{\circ} \bqH^H \qPsi^{\circ} \right)} \right)^{\frac{1}{2}}  \nonumber\\
                          & = \left| \frac{u'^{(2)}_{lk,ij}}{1-u'^{(2)}_{lk,lk}} \right|^{\frac{1}{2}} \left| \frac{u'^{\circ(2)}_{lk,ij}}{1-u'^{\circ(2)}_{lk,lk}} \right|^{\frac{1}{2}}  = \left| [\qK_{11}]_{lk,ij} \right|^{\frac{1}{2}} \left| [\qK^{\circ}_{11}]_{lk,ij} \right|^{\frac{1}{2}}.
\end{align}
Likewise, we have
\begin{subequations}
\begin{align}%\label{eq:qPi122122}
   |[\qPi_{12}]_{lk,ij}| \leq & \left| [\qK_{12}]_{lk,ij} \right|^{\frac{1}{2}} \left| [\qK^{\circ}_{12}]_{lk,ij} \right|^{\frac{1}{2}},\\
   |[\qPi_{21}]_{lk,ij}| \leq & \left| [\qK_{21}]_{lk,ij} \right|^{\frac{1}{2}} \left| [\qK^{\circ}_{21}]_{lk,ij} \right|^{\frac{1}{2}},\\
   |[\qPi_{22}]_{lk,ij}| \leq & \left| [\qK_{22}]_{lk,ij} \right|^{\frac{1}{2}} \left| [\qK^{\circ}_{22}]_{lk,ij} \right|^{\frac{1}{2}}.
\end{align}
\end{subequations}
Using Lemma \ref{Lemma:SpectralRadius3} and Lemma \ref{Lemma:SpectralRadius4}, we obtain
\begin{equation}\label{eq:rhoqPi}
    \rho(\qPi) \leq \rho(|\qPi|) \leq \rho(\qK)^{\frac{1}{2}} \rho(\qK^{\circ})^{\frac{1}{2}} <1.
\end{equation}
This contradicts to the statement that $\qPi$ has an eigenvalue equal to $1$. Therefore, we have $e_{l,k} - e_{l,k}^{\circ} = 0$ and $\te_{l,k} - \te_{l,k}^{\circ}
= 0$, for any $l,k$ and $\omega \in \bbR^+$.

\section{Mathematical Tools}\label{Appendix: Mathematical Tools}
In this appendix, we provide some mathematical tools used in the proof of the appendices.

\begin{Lemma}\label{Lemma:TraceNorm}{\rm \cite{Horn-91}}
\begin{enumerate}
  \item Let $\qA = [A_{ij}]$ and $\qB$ be any matrices such that the product is a square matrix. Then,
 \begin{enumerate}
  \item $|\tr(\qA\qB)| \leq \|\qA\|_{\rm F}\|\qB\|_{\rm F}$,
  \item $\|\qA\qB\|_{\rm F} \leq \|\qA\|_{\rm F}\|\qB\|$,
  \item $\|\qA\qB\|_{\rm F} \leq \|\qA\|_{\rm F}\|\qB\|_{\rm F}$,
  \item $|A_{ij}| \leq \|\qA\|$.
 \end{enumerate}
  \item If $\qA$ is nonnegative definite, we have $|\tr(\qA\qB)| \leq \|\qB\|\tr(\qA)$.
  \item Let $\qA$ be any matrix such that the product $\qA\qB$ exists. Then, $\|\qA\qB\| \leq \|\qA\|\|\qB\|$.
\end{enumerate}
\end{Lemma}

\begin{Lemma}\label{Lemma:AbsIneq}
For any $p\geq 1$ and real numbers $a_i$'s, we have
\begin{equation}
\left|\sum^n_{i=1}a_i\right|^p \leq n^{p-1} \sum^n_{i=1}|a_i|^p.
\end{equation}
\end{Lemma}

\begin{Lemma}\label{Lemma:Weyl} {\rm \cite[Theorem 4.3.1]{Horn-90}}
Let $\qA$ and $\qB$ be Hermitian matrix and let the eigenvalues $\lambda_i(\qA)$, $\lambda_i(\qB)$, and $\lambda_i(\qA+\qB)$ be arranged in decreasing order. For
each $k = 1,2,\dots,n$, we have
\begin{equation}
\lambda_k(\qA)+\lambda_n(\qB) \leq \lambda_k(\qA+\qB) \leq \lambda_k(\qA)+\lambda_1(\qB).
\end{equation}
\end{Lemma}

\begin{Lemma}\label{Lemma:SpNormHHupb}
Let matrix $\qA_{l,k} \in\bbC^{N_l \times n_k}$ for $l=1,\ldots,L,k=1,\ldots,K$, and let $\qA_k = \left[ \qA_{1,k}^T \cdots \qA_{L,k}^T \right]^T \in \bbC^{N\times
n_k}, \qA = \left[\qA_1, \cdots, \qA_K\right]  \in \bbC^{N \times n}$, with $N =\sum_{l=1}^{L} N_l$ and  $n=\sum_{k=1}^{K} n_k$. If $\|\qA_{l,k}\qA_{l,k}^H\|  \leq
C$, then we have  $\|\qA\qA^H\|  \leq LK C$.
\end{Lemma}

\begin{proof}
Notice that $\qA\qA^H$ and $\qA_{l,k}\qA_{l,k}^H$ are Hermitian matrices. Therefore, a standard computation involving Lemma \ref{Lemma:Weyl} yields
\begin{align}
     \|\qA\qA^H\| & = \lambda_1\left(\qA\qA^H\right) = \lambda_1\left(\sum^K_{k=1}\qA_k\qA_k^H\right) \nonumber\\
                    &\leq \sum^K_{k=1} \lambda_1\left(\qA_k\qA_k^H\right) = \sum^K_{k=1} \lambda_1\left(\qA_k^H\qA_k\right) \nonumber\\
                    & = \sum^K_{k=1}  \lambda_1\left(\sum^L_{l=1}\qA_{l,k}^H\qA_{l,k}\right) \leq \sum_{l,k} \lambda_1\left(\qA_{l,k}^H\qA_{l,k}\right) \nonumber\\
                    & = \sum_{l,k}  \|\qA_{l,k}^H\qA_{l,k}\|  \leq  LKC.
\end{align}
\end{proof}

\begin{Lemma}\label{Lemma:ResolventIdentity}{\rm (Resolvent Identity)}
For invertible $\qA$ and $\qB$ matrices, we have the identity
\begin{equation}
\qA^{-1}-\qB^{-1} = \qA^{-1}(\qB-\qA)\qB^{-1}.
\end{equation}
\end{Lemma}

\begin{Lemma}\label{Lemma:MatrixInversion}{\rm (Matrix Inversion)}
   For invertible $\qA,\qB$ and $\qR$ matrices, suppose that $\qB=\qA+\qX\qR\qY$, then
   \begin{equation}
     \qB^{-1} = \qA^{-1}-\qA^{-1}\qX(\qR^{-1}+\qY\qA^{-1}\qX)^{-1}\qY\qA^{-1}.  \nonumber
   \end{equation}
\end{Lemma}

\begin{Lemma}\label{Lemma:optmalCovLemma}
   Assume that $\qA$ is a positive seme-definite $M\times M$ matrix and $\qB = \diag(\qB_1,\ldots,\qB_K)$ is a block-diagonal matrix, where $\qB_k$ is a positive seme-definite $M_k\times M_k$ matrix and $M=\sum_{k=1}^K M_k$. Let
   $\qC_{k} = \ang{(\qI+\qA\qB_{\backslash k})^{-1}\qA}_k,$ where $\qB_{\backslash k} = \diag (\qB_1,\ldots,\qB_{k-1},\qzero,\qB_{k+1},\ldots,\qB_K)$. Then, we have
   \begin{equation}\label{eq:optmalCovLemma1}
     \ang{(\qI+\qA\qB)^{-1}\qA}_k = (\qI+\qC_k\qB_k)^{-1}\qC_k.
   \end{equation}
\end{Lemma}
\begin{proof}
Letting $\uqB_k = \diag(\qzero, \ldots,\qzero,\qB_k,\qzero,\ldots,\qzero)$, we have
\begin{align}
    \left(\qI+\qA\qB\right)^{-1}\qA =& \left(\qI+\qA\qB_{\backslash k}+\qA\uqB_k\right)^{-1}\qA \nonumber \\
   \mathop = \limits^{(i)}& \qC - \qC\left(\qI+\uqB_k\qC\right)^{-1}\uqB_k\qC  %\nonumber \\
                        = \qC \left(\qI- \left((\uqB_k\qC)^{-1} + \qI\right)^{-1}\right) \nonumber \\
   \mathop = \limits^{(ii)}& \qC \left((\uqB_k\qC)^{-1}\left((\uqB_k\qC)^{-1} + \qI\right)^{-1}\right)  %\nonumber \\
                        = \left(\qI + \qC\uqB_k \right)^{-1}\qC \label{eq:optmalCovLemma2}
\end{align}
where $(i)$ follows from Lemma \ref{Lemma:MatrixInversion} and defining $\qC = \left(\qI+\qA\qB_{\backslash k}\right)^{-1}\qA$,  $(ii)$ is due to Lemma
\ref{Lemma:ResolventIdentity}. Substituting \eqref{eq:optmalCovLemma2} into \eqref{eq:optmalCovLemma1}, we obtain
 \begin{equation}
     \ang{(\qI+\qA\qB)^{-1}\qA}_k = \ang{\left(\qI + \qC\uqB_k \right)^{-1}\qC}_k = (\qI+\qC_k\qB_k)^{-1}\qC_k, \nonumber
 \end{equation}
where the last step is obtained by calculating the inverse of $(\qI + \qC\uqB_k)^{-1}$.
\end{proof}

\begin{Lemma}\label{Lemma:SpectralRadius}{\rm \cite[Corollary 8.1.29]{Horn-90}}
   Let $\qA \in \bbR^{n\times n}$, $\qx \in \bbR^{n}$, for $\qA\geq 0$ and $\qx >0$. If $\alpha, \beta \geq 0$ are such that $\alpha\qx\leq \qA\qx \leq\beta\qx$, then $\alpha\leq \rho(\qA)\leq\beta$. If $\alpha\qx < \qA\qx$, then $\alpha< \rho(\qA)$. If $\qA\qx < \beta\qx$, then $\rho(\qA)<\beta$.
\end{Lemma}

\begin{Lemma}\label{Lemma:SpectralRadius2}{\rm \cite[Lemma 9]{Couillet-11IT}}
   If the components of $\qC, \qx$, and $\qb$ are all positive, then $\qx = \qC \qx + \qb$ implies $\rho(\qC)<1$.
\end{Lemma}

\begin{Lemma}\label{Lemma:traceineq}{\rm \cite[Lemma 16]{Wen-11IT}}
   Let $\qA$ and $\qB$ be any matrices such that $\qA\qB^H$ exists and is a squared matrix. If $\tr(\qA\qA^H)\leq 1$ and $\tr(\qB\qB^H)\leq 1$, then
   \begin{equation}
    |1-\tr(\qA\qB^H)| \geq \left(1-\tr(\qA\qA^H)\right)^{\frac{1}{2}}\left(1-\tr(\qB\qB^H)\right)^{\frac{1}{2}}.
   \end{equation}
\end{Lemma}

\begin{Lemma}\label{Lemma:SpectralRadius3}{\rm \cite[Theorem 8.1.18]{Horn-90}}
   Let $\qA=[A_{ij}]$ and $\qB=[B_{ij}]$ be square matrices. If $|A_{ij}|\leq B_{ij}, \forall i,j$, then $\rho(\qA) \leq \rho(|\qA|) \leq \rho(\qB)$.
\end{Lemma}

\begin{Lemma}\label{Lemma:SpectralRadius4}{\rm \cite[Lemma 5.7.9]{Horn-91}}
   Let $\qA=[A_{ij}]$ and $\qB=[B_{ij}]$ be matrices with nonnegative elements. Then $\rho([A^{\frac{1}{2}}_{ij} B^{\frac{1}{2}}_{ij}]) \leq \rho(\qA)^{\frac{1}{2}}\rho(\qB)^{\frac{1}{2}}$.
\end{Lemma}

\section*{\sc Acknowledgment}
We thank the reviewers for the careful reviews and for their suggestions which helped in improving the quality of the paper.

{\renewcommand{\baselinestretch}{1.1}
\begin{footnotesize}
\bibliographystyle{IEEEtran}

\begin{thebibliography}{10}
\providecommand{\url}[1]{#1} \csname url@samestyle\endcsname \providecommand{\newblock}{\relax} \providecommand{\bibinfo}[2]{#2}
\providecommand{\BIBentrySTDinterwordspacing}{\spaceskip=0pt\relax} \providecommand{\BIBentryALTinterwordstretchfactor}{4}
\providecommand{\BIBentryALTinterwordspacing}{\spaceskip=\fontdimen2\font plus \BIBentryALTinterwordstretchfactor\fontdimen3\font minus
  \fontdimen4\font\relax}
\providecommand{\BIBforeignlanguage}[2]{{%
\expandafter\ifx\csname l@#1\endcsname\relax
\typeout{** WARNING: IEEEtran.bst: No hyphenation pattern has been}%
\typeout{** loaded for the language `#1'. Using the pattern for}%
\typeout{** the default language instead.}%
\else \language=\csname l@#1\endcsname \fi #2}} \providecommand{\BIBdecl}{\relax} \BIBdecl

\bibitem{Fos-98}
G.~J. Foschini and M.~J. Gans, ``On limits of wireless communications in a
  fading environment when using multiple antennas,'' \emph{Kluwer Academic
  Publishers--Wireless Per. Commun.}, vol.~6, pp. 311--335, 1998.

\bibitem{Tel-99}
{\.I. E. Telatar}, ``Capacity of multi-antenna gaussian channels,'' \emph{Euro.
  Trans. Telecom.}, vol.~10, pp. 585--595, 1999.

\bibitem{Marzetta10TWC}
{T. L. Marzetta}, ``{Noncooperative cellular wireless with unlimited numbers of
  base station antennas},'' \emph{IEEE Trans. Wireless Commun.}, vol.~9,
  no.~11, pp. 3590--3600, Nov. 2010.

\bibitem{Jose11TWC}
{J. Jose, A. Ashikhmin, T. Marzetta, and S. Vishwanath}, ``{Pilot contamination
  and precoding in multi-cell TDD systems},'' \emph{IEEE Trans. Wireless
  Commun.}, vol.~10, no.~8, pp. 2640--2651, Aug. 2011.

\bibitem{Ngo11ACASP}
{H. Q. Ngo, E. G. Larsson, and T. L. Marzetta}, ``{Analysis of the pilot
  contamination effect in very large multicell multiuser MIMO systems for
  physical channel models},'' in \emph{Proc. IEEE International Conference on
  Acoustics, Speech and Signal Processing}, Prague, Czech Repulic, May 2011,
  pp. 3464--3467.

\bibitem{Hoydis11ACC}
{ J. Hoydis, S. ten Brink, and M. Debbah}, ``{Massive MIMO: How many antennas
  do we need?}'' in \emph{Proc. 49th Allerton Conference on Communication,
  Control, and Computing}, Urbana-Champaign, Illinois, USA, Sep. 2011.

\bibitem{Rusek11SPM}
\BIBentryALTinterwordspacing { F. Rusek, D. Persson, B. K. Lau, E. G. Larsson, T. L. Marzetta, O. Edfors,
  and F. Tufvesson}, ``{Scaling up MIMO: Opportunities and challenges with very
  large arrays},'' \emph{IEEE Sig. Proc. Mag.}, 2011. [Online]. Available:
  \url{http://arxiv.org/abs/1201.3210.}
\BIBentrySTDinterwordspacing

\bibitem{ZTE-01}
{}, ``{ZTE Green Technology Innovations},'' \emph{White Paper}, 2011.

\bibitem{Huh-11}
{H. Huh and S.-H. Moon, Y.-T. Kim, I. Lee and G. Caire}, ``Multi-cell mimo
  downlink with cell cooperation and fair scheduling: A large-system limit
  analysis,'' \emph{{IEEE} Trans. Inf. Theory}, vol.~57, no.~12, pp.
  7771--7786, Dec. 2011.

\bibitem{Moustakas03IT}
{A. L. Moustakas, S. Simon, and A. M. Sengupta}, ``{MIMO capacity through
  correlated channels in the presence of correlated interferers and noise: A
  (not so) large N analysis},'' \emph{IEEE Trans. Inf. Theory}, vol.~49,
  no.~10, pp. 2545--2561, Oct. 2003.

\bibitem{Tulino-04}
{A. M. Tulino and S. Verd\'u}, ``{Random Matrix Theory and Wireless
  Communications},'' \emph{Found. Trends Commun. Inf. Theory}, vol.~1, pp.
  1--182, Jun. 2004.

\bibitem{Hachem-07AAP}
{W. Hachem, P. Loubaton, and J. Najim}, ``{Deterministic equivalents for
  certain functionals of large random matrices},'' \emph{Ann. App. Probab.},
  vol.~17, no.~3, pp. 875--930, 2007.

\bibitem{Hachem-08IT}
{W. Hachem, O. Khorunzhiy, P. Loubaton, J. Najim, and L. Pastur}, ``{A new
  approach for mutual information analysis of large dimensional multi-antenna
  channels},'' \emph{IEEE Trans. Inf. Theory}, vol.~54, no.~9, pp. 3987--4004,
  Sep. 2008.

\bibitem{Taricco08IT}
{ G. Taricco}, ``{Asymptotic mutual information statistics of
  separately-correlated Rician fading MIMO channels},'' \emph{IEEE Trans. Inf.
  Theory}, vol.~54, no.~8, pp. 3490--3504, Aug. 2008.

\bibitem{Dumont-10IT}
{J. Dumont, S. Lasaulce, W. Hachem, Ph. Loubaton and J. Najim}, ``{On the
  capacity achieving covariance matrix for Rician MIMO channels: an asymptotic
  approach},'' \emph{IEEE Trans. Inf. Theory}, vol.~56, no.~3, pp. 1048--1069,
  Mar. 2010.

\bibitem{Couillet-11IT}
{R. Couillet, M. Debbah, and J. W. Silverstein}, ``{A deterministic equivalent
  for the capacity analysis of correlated MIMO multiple access channels},''
  \emph{IEEE Trans. Inf. Theory}, vol.~57, no.~6, pp. 3493--3514, Jun. 2011.

\bibitem{Couillet-11Book}
{R. Couillet and M. Debbah}, \emph{Random Matrix Methods for Wireless
  Communications}.\hskip 1em plus 0.5em minus 0.4em\relax Cambridge University
  Press, 2011.

\bibitem{Dupuy-11IT}
{F. Dupuy and P. Loubaton}, ``{On the capacity achieving covariance matrix for
  frequency selective MIMO channels using the asymptotic approach},''
  \emph{IEEE Trans. Inf. Theory}, vol.~57, no.~9, pp. 5737--5753, Sep. 2011.

\bibitem{Wen-11IT}
\BIBentryALTinterwordspacing {C. K. Wen, G. Pan, K.-K. Wong, M. H. Guo, and J. C. Chen}, ``{A deterministic
  equivalent for the analysis of non-Gaussian correlated MIMO multiple access
  channels},'' preprint 2011. [Online]. Available:
  \url{http://arxiv.org/abs/1108.4096.}
\BIBentrySTDinterwordspacing

\bibitem{Molisch05}
{A. F. Molisch et al}, ``{IEEE 802.15.4a channel model -- Final report},'' in
  \emph{Tech. Rep}, Document IEEE 802.1504-0062-02-004a, 2005.

\bibitem{Foerster03}
\BIBentryALTinterwordspacing {J. R. Foerster, M. Pendergrass, and A. F. Molisch}, ``{ A channel model for
  ultrawideband indoor communication},'' 2003. [Online]. Available:
  \url{http://www.merl.com/reports/docs/TR2003-73.pdf}
\BIBentrySTDinterwordspacing

\bibitem{Bai-10}
Z.~Bai and J.~W. Silverstein, \emph{Spectral Analysis of Large Dimensional
  Random Matrices}.\hskip 1em plus 0.5em minus 0.4em\relax Springer Series in
  Statistics, 2010.

\bibitem{Huh-12}
{H. Huh, A. M. Tulino, and G. Caire}, ``{Network MIMO With Linear Zero-Forcing
  Beamforming: Large System Analysis, Impact of Channel Estimation, and
  Reduced-Complexity Scheduling},'' \emph{{IEEE} Trans. Inf. Theory}, vol.~58,
  no.~5, pp. 2911--2934, May 2012.

\bibitem{Hoy-11c}
{J. Hoydis, M. Kobayashi, and M. Debbah}, ``{Optimal channel training in uplink
  network MIMO systems},'' \emph{IEEE Trans. Sig. Proc.}, vol.~59, no.~6, pp.
  2824--2833, Jun 2011.

\bibitem{Hoy-12}
{J. Hoydis, A. M\:uller, R. Couillet, M. Debbah}, ``Analysis of multicell
  cooperation with random user locations via deterministic equivalents,'' in
  \emph{Eighth Workshop on Spatial Stochastic Models for Wireless Networks},
  Paderborn, Germany, 2012.

\bibitem{Wagner-11IT}
{S. Wagner, R. Couillet, M. Debbah, and D. T. M. Slock}, ``{Large system
  analysis of linear precoding in correlated MISO broadcast channels under
  limited feedback},'' \emph{IEEE Trans. Info. Theory}, vol.~58, no.~7, pp.
  4509--4537, Jul. 2012.

\bibitem{Pastur-05AAP}
{L. A. Pastur}, ``{A simple approach to the global regime of Gaussian ensembles
  of random matrices},'' \emph{Ukrainian Math. J.}, vol.~57, pp. 936--966,
  2005.

\bibitem{Chatterjee-06AAP}
{S. Chatterjee}, ``{A generalization of the Lindeberg principle},'' \emph{Ann.
  Appl. Probab.}, vol.~34, no.~6, pp. 2061--2076, 2006.

\bibitem{Lytova-09AAP}
{ A. Lytova and L. A. Pastur}, ``{Central limit theorem for linear eigenvalues
  of statistics of random matrices with independent entries},'' \emph{Ann.
  Appl. Probab.}, vol.~37, no.~5, pp. 1778--1840, 2009.

\bibitem{Korada-11IT}
{S. Korada and A. Montanari}, ``{Applications of the Lindeberg principle in
  communications and statistical learning},'' \emph{IEEE Trans. Inf. Theory},
  vol.~57, no.~4, pp. 2440--2450, Apr. 2011.

\bibitem{Shiu-00TCOM}
{D. Shiu, G. J. Foschini, M. J. Gans, and J. M. Kahn}, ``{Fading correlation
  and its effect on the capacity of multi-element antenna systems},''
  \emph{IEEE Trans. Commun.}, vol.~48, no.~3, pp. 502--513, Mar. 2000.

\bibitem{Goldsmith-03JSAC}
{A. Goldsmith, S. A. Jafar, N. Jindal, and S. Vishwanath}, ``{Capacity limits
  of MIMO channels},'' \emph{IEEE J. Sel. Areas Commun.}, vol.~21, no.~5, pp.
  684--702, June 2003.

\bibitem{Wen-07TCOM}
{C. K. Wen, K. K. Wong, and J. C. Chen}, ``{Spatially correlated MIMO
  multiple-access systems with macrodiversity: Asymptotic analysis via
  statistical physics},'' \emph{IEEE Trans. Commun.}, vol.~55, no.~3, pp.
  477--488, Mar. 2007.

\bibitem{Taricco-11IT}
{G. Taricco and E. Riegler}, ``{On the ergodic capacity of correlated Rician
  fading MIMO channels with interference},'' \emph{IEEE Trans. Inf. Theory},
  vol.~57, no.~7, pp. 4123--4137, Jul. 2011.

\bibitem{Yu-04IT}
{ W. Yu, W. Rhee, S. Boyd, and J. M. Cioffi}, ``{Iterative water-filling for
  Gaussian vector multiple access channels},'' \emph{{IEEE} Trans. Inf.
  Theory}, vol.~50, no.~1, pp. 145--151, Jan. 2004.

\bibitem{Boyd04}
{S. Boyd and L. Vandenberghe}, \emph{Convex Optimization}.\hskip 1em plus 0.5em
  minus 0.4em\relax Cambridge Univ. Press, 2004.

\bibitem{Vu-05PAC}
{M. Vu and A. Paulraj}, ``{Capacity optimization for Rician correlated MIMO
  wireless channels},'' in \emph{Proc. Asilomar Conf. Sig., Sys. and Comp.},
  Pacific Grove, CA, Nov. 2005, pp. 133--138.

\bibitem{Wen-11TCOM}
{C. K. Wen, S. Jin, and K.-K. Wong}, ``{On the sum-rate of multiuser MIMO
  uplink channels with jointly-correlated Rician fading},'' \emph{IEEE Trans.
  Commun.}, vol.~59, no.~10, pp. 2883--2895, Oct. 2011.

\bibitem{Wen06TCOM}
{C. K. Wen, P. Ting, and J. T. Chen}, ``{Asymptotic analysis of MIMO wireless
  systems with spatial correlation at the receiver},'' \emph{IEEE Trans.
  Commun.}, vol.~54, no.~2, pp. 349--363, Feb. 2006.

\bibitem{Bol-02TCOM}
{H. B\"olcskei, D. Gesbert, and A. Paulraj}, ``{On the capacity of OFDM based
  spatial multiplexing systems},'' \emph{IEEE Trans. Commun.}, vol.~50, no.~2,
  pp. 225--234, Feb. 2002.

\bibitem{Choi07ICC}
{S. H. Choi, P. Smith, B. Allen, W. Q. Malik, and M. Shafi}, ``Severely fading
  mimo channels: Models and mutual information,'' in \emph{Proc. IEEE Int.
  Conf. Commun}, Glasgow, Scotland, Jun. 2007, pp. 4628--4633.

\bibitem{Horn-90}
{R. Horn and C. Johnson}, \emph{Matrix Analysis}.\hskip 1em plus 0.5em minus
  0.4em\relax Cambridge Univ. Press, 1990.

\bibitem{Cartan78}
H.~Cartan, ``Th\'eorie elementaire des fonctions analytiques d'une ou plusieurs
  variables complexes,'' \emph{Hermann}, 1978.

\bibitem{Hachem08AAP}
{W. Hachem, P. Loubaton, and J. Najim}, ``{A CLT for information-theoretic
  statistics of Gram random matrices with a given variance profile},''
  \emph{Ann. Appl. Probab.}, vol.~18, no.~6, pp. 2071--2130, Dec. 2008.

\bibitem{Horn-91}
{R. Horn and C. Johnson}, \emph{Topics in Matrix Analysis}.\hskip 1em plus
  0.5em minus 0.4em\relax Cambridge Univ. Press, 1991.

\end{thebibliography}

\end{footnotesize}}
\end{document}